\documentclass{amsart}
\usepackage[utf8]{inputenc}
\usepackage[title]{appendix}
\usepackage{graphicx}%
\usepackage{epstopdf}
\graphicspath{{./figure/}} 
\usepackage{mathtools}
\usepackage{multirow}%
\usepackage{amsmath,amssymb,amsfonts}%
\usepackage{amsthm}%
\usepackage{mathrsfs}%
\usepackage{xcolor}%
\usepackage{textcomp}%
\usepackage{manyfoot}%
\usepackage{booktabs}%
\usepackage{caption,subcaption}
\usepackage{float}
\usepackage{listings}%

\usepackage{hyperref}
\usepackage{cleveref}
\usepackage{enumerate}
\usepackage{mathtools}
\newtheorem{lemma}{Lemma}
\newtheorem{definition}{Definition}
\newtheorem{proposition}{Proposition}
\newtheorem{remark}{Remark}
\newtheorem{corollary}{Corollary}
\newtheorem{theorem}{Theorem}
\newtheorem{assump}{Assumption}

\numberwithin{equation}{section}
\numberwithin{figure}{section}
\numberwithin{theorem}{section}
\numberwithin{lemma}{section}
\numberwithin{proposition}{section}
\numberwithin{corollary}{section}
\numberwithin{remark}{section}
\numberwithin{definition}{section}
\numberwithin{assump}{section}

\def\eu{\ensuremath{\mathrm{e}}}
\def\iu{\ensuremath{\mathrm{i}}}
\def\du{\ensuremath{\mathrm{d}}}
\def\wD{\ensuremath{\widetilde{D}}}
\def\wY{\ensuremath{\widetilde{Y}}}
\def\wV{\ensuremath{\widetilde{V}}}
\def\fu{\ensuremath{\mathrm{full}}}
\def\wI{\ensuremath{\widetilde{I}}}
\def\wcV{\ensuremath{\widetilde{\mathcal{V}}}}
\def\wcW{\ensuremath{\widetilde{\mathcal{W}}}}
\def\cV{\ensuremath{\mathcal{V}}}
\def\cW{\ensuremath{\mathcal{W}}}

\begin{document}

\title[Tight-binding Approach for Computing Guided Modes]{A Tight-binding Approach for Computing Subwavelength Guided Modes in Crystals with Line Defects}

\author{Habib Ammari}
\address{Department of Mathematics, ETH Z\"{u}rich, Rämistrasse 101, CH-8092 Z\"{u}rich, Switzerland}
\email{habib.ammari@math.ethz.ch}

\author{Erik Orvehed Hiltunen}
\address{Department of Mathematics, University of Oslo, 
	Moltke Moes vei 35, 0851 Oslo, Norway}
\email{erikhilt@math.uio.no}

\author{Ping Liu}
\address{School of Mathematical Sciences, Zhejiang University, 310027 Hangzhou, China and Institute of Fundamental and Transdisciplinary Research, Zhejiang University, 310027 Hangzhou, China}
\email{pingliu@zju.edu.cn}

\author{Borui Miao}
\address{Yau Mathematical Sciences Center, Tsinghua University, 100084 Beijing, China}
\email{mbr@mail.tsinghua.edu.cn}

\author{Yi Zhu}
\address{Yau Mathematical Sciences Center, Tsinghua University, 100084 Beijing, China and Beijing Institute of Mathematical Sciences and Applications, 101408 Beijing, China}
\email{yizhu@tsinghua.edu.cn}


\date{}


\begin{abstract}
In this paper, we develop a framework for computing subwavelength guided modes in high-contrast periodic media with line defects, based on a tight-binding approximation. 
The physical problem is formulated as an eigenvalue problem for the Helmholtz equation with high-contrast parameters. By employing layer potential theory on unbounded domains, we characterize the subwavelength frequencies via the quasi-periodic capacitance matrix.
Our main contribution is the proof of exponential decay of the off-diagonal elements of the associated full and quasi-periodic capacitance matrices. These decay properties provide error bounds for the banded approximation of the capacitance matrices, thereby enabling a tight-binding approach for computing the spectral properties of subwavelength resonators with non-compact defects. Numerical experiments illustrate the decay and compare full and banded capacitance matrices for defect and dimerized interface configurations. 

\end{abstract}

\maketitle

\section{Introduction} 
\subsection{Motivation}
The study of line defects in subwavelength bandgap materials has had a major impact on various technological applications. The capability of subwavelength devices to guide waves significantly below the diffraction limit is pivotal to developing the next generation of communication, biomedical, quantum, and sensing technologies \cite{integrated-photonics,defect1,defect2,defect3,phononic1}. These subwavelength structures are often referred to as \emph{metamaterials}. Similar attention has been devoted to photonic and phononic crystals, where wave guiding phenomena induced by line defects have been investigated both numerically and analytically \cite{santosa,ong1,ong2,brown1,brown2,brown3,brown4,fliss1,fliss2,fliss3,figotin}. However, in these classical settings, the guiding phenomenon typically occurs at scales comparable to the unit cell size of the microstructure.

In the subwavelength regime, as demonstrated in \cite{Ammari_2017}, bandgaps corresponding to wavelengths much larger than the unit cell of the microstructure can be found in a broad class of high-contrast periodic resonator systems. Here, the density contrast parameter $\delta$ is small; for air bubbles in water, for example, $\delta$ is of order $10^{-3}$. By introducing defects into the periodic structure, resonant frequencies may emerge within the subwavelength bandgap of the unperturbed structure. These resonances, referred to as subwavelength bandgap frequencies, correspond to modes that are spatially localized near the defect site. In \cite{jems2021}, it is shown that a line defect creates a \emph{band} of subwavelength bandgap frequencies corresponding to the modes propagating along the defect line. More precisely, using the fictitious source superposition method originally introduced in \cite{effective-source}, the authors in \cite{jems2021} model the defect sites as unperturbed resonators with additional fictitious monopole and dipole sources. The size of the perturbation needed to ensure that the entire defect band lies inside the bandgap is explicitly quantified. Moreover, it is shown in \cite{jems2021} that the defect band is nowhere flat and therefore supports propagating modes along the line defect. \par 

From a computational standpoint, defect bands can be evaluated by several numerical approaches, including finite element methods (FEM) \cite{Guo2021,Guo2021a}, boundary element methods (BEM) \cite{Barnett2010} and schemes based on computing characteristic values of families of quasi-periodic operators. In the latter approach, the operators are discretized using truncated Fourier bases, and the corresponding characteristic values are calculated by applying root-finding algorithms, such as Muller's method, to the determinants of the truncated operators \cite{Ammari2018}. However, standard FEM or BEM discretizations in the high-contrast subwavelength regime can be computationally demanding. Resolving material interfaces and the geometry of resonators can lead to large discrete systems, with substantial costs for matrix assembly, storage, and solution \cite{Guo2021a,Wout2022}.
  \par 

An effective alternative to direct discretization is the capacitance matrix formulation \cite{paa,cbms,Ammari2018}, which gives a discrete approximation of the subwavelength frequencies of the continuous partial differential equation model in the high-contrast limit $\delta\to 0$. The resulting capacitance matrix is typically dense, however, making both explicit analysis and efficient numerical computation difficult. To justify efficient truncation schemes, we therefore need quantitative estimates of the decay of its off-diagonal elements.

In this paper, we establish a computationally efficient and accurate approach for determining subwavelength guided modes in crystals with line defects. Our strategy relies on a quasi-periodic capacitance matrix formulation derived from layer potential theory on unbounded domains. By establishing the exponential decay of off-diagonal elements, we provide a quantitative justification for approximating the continuous Helmholtz eigenvalue problem on an unbounded domain by banded capacitance matrices. This banded matrix model, which is referred to as a tight-binding model, is widely used to design metamaterials and analyze wave transport and localization properties; see, for example, \cite{TB0,TB1,TB2,TB3,integrated-photonics}.
This tight-binding approach for computing guided modes generalizes the decay analysis in \cite[Proposition 2.5 and Section 3.2]{ammari2025analysisnonlinearresonancesresonator} for the capacitance matrix associated with the unperturbed crystal and provides quantitative estimates of local decay rates. 
Moreover, \Cref{cor:banded_operator_error} allows the bandwidth to be chosen as a function of $\delta$ to balance the truncation error with the asymptotic error; see \Cref{rem:combined_error}.

\subsection{Main contributions}
The main results of this paper can be summarized as follows.
\begin{enumerate}
	\item \textbf{Exponential decay of off-diagonal elements of the full and quasi-periodic capacitance matrices.}
	In \Cref{thm:ExpConv} and \Cref{thm:ExpConv_quasi}, we prove that the off-diagonal elements of the full and quasi-periodic capacitance matrices decay exponentially away from the diagonal. These results are proved by the truncation method; see \Cref{sec:Decay_periodic} and \Cref{sec:estimate_full}. In addition, the exponential decay provides a quantitative justification for truncating the capacitance matrix with explicit error control and allows the interaction bandwidth to be chosen according to the desired accuracy. This establishes a tight-binding approximation of subwavelength systems. 
	
	\item \textbf{Asymptotic expansion of subwavelength frequencies.}
	For fixed $\alpha\neq 0$, we apply the capacitance reduction of \cite[Theorem 2.6 and Section 4]{ammari2025analysisnonlinearresonancesresonator} to the waveguide fiber. We prove spectral proximity and, for guided frequencies near an isolated positive generalized eigenvalue, the asymptotic formula
	\[
	\omega(\alpha)=\mu_{_1}\sqrt{\delta\,\lambda^\alpha}+\mathcal{O}(\delta),
	\]
	where $\lambda^\alpha$ is a generalized eigenvalue of the quasi-periodic capacitance matrix; see \eqref{eqn:genEigProb} and the neighborhood condition in \Cref{thm:capamat_approx}. The proof uses layer potential theory on unbounded domains; see \Cref{apsec:invert} and \Cref{apsec:Invert_NP}. It gives a quantitative connection between the continuous spectral problem and the effective discrete model. 
\end{enumerate}
From a computational viewpoint, these results provide a \emph{rigorous and quantitative relation} between the continuous high-contrast Helmholtz eigenvalue problem and a finite-banded matrix generalized eigenvalue problem, thereby enabling the efficient computation of defect bands and guided modes.

\subsection{Organization of the paper}
The paper is organized as follows. In \Cref{sec:setting} we introduce the geometric setting of the resonators, the subwavelength problem, and the definition of the full capacitance matrix and the corresponding quasi-periodic capacitance matrix. In \Cref{sec:Decay_periodic} and \Cref{sec:estimate_full}, we prove the exponential decay of off-diagonal elements of the full capacitance matrix. In \Cref{sec:Subwavelength_capa} we further derive the capacitance matrix formulation for the subwavelength problem. Numerical experiments are provided in \Cref{sec:numerics} to illustrate the theoretical results and assess the nearest-neighbor approximation. We also consider a dimerized interface configuration. Concluding remarks and potential extensions of the method are discussed in \Cref{sec:conclude}. Finally, the appendices contain necessary preliminaries and some supporting theoretical results.

\section{Problem setting and preliminaries}\label{sec:setting}
We begin with the geometry of the waveguide. Let $v_{_1},v_{_2}\in\mathbb R^2$ be linearly independent lattice vectors. The lattice and its cells are
\begin{align}
 \Lambda &\coloneqq \{mv_{_1}+nv_{_2}:m,n\in\mathbb Z\},\label{eqn:lattice}\\
 Y_{m,n}&\coloneqq\{(m+s)v_{_1}+(n+t)v_{_2}:s,t\in[0,1)\}.\label{eqn:cells}
\end{align}
Specifically we let 
\begin{equation*}
 v_{_1}=(1,0),\qquad v_{_2}=(b,h),\qquad h>0.
\end{equation*}

Let $D_{m,n}$ denote the inclusion in each cell $Y_{m,n}$. The region occupied by all inclusions is
\begin{equation}\label{eqn:inclusionregions}
 D_{\fu}\coloneqq\bigcup_{m,n\in\mathbb Z}D_{m,n},\qquad D_{m,n}=D_{\fu}\cap Y_{m,n}.
\end{equation}
In this article, we suppose the inclusions satisfy following assumptions.
\begin{assump}\label{def:geometry}
	The inclusions $ D_{\fu}\subset \mathbb{R}^2 $ satisfy the following conditions:
	\begin{itemize}
		\item[(1)] For any $ Y_{m,n} $, the inclusion within the cell $ D_{m,n} $, is a nonempty open set with smooth boundary. The set $ D_{m,n} $ is connected and satisfies $ D_{m,n}\Subset  Y_{m,n} $ for all $ m,n\in \mathbb{Z} $;
		\item[(2)] The overall inclusion $ D_{\fu} $ is $ v_{_1} $-periodic.
		\[D_{\fu} + mv_{_1} = D_{\fu},\qquad  m\in \mathbb{Z}.\]
	\end{itemize}
\end{assump}
The left panel of \Cref{fig:CellExample} illustrates the geometry of the waveguide. It follows that $\partial D_{m,n}$ has zero Lebesgue measure in $\mathbb{R}^2$ for $m,n\in\mathbb{Z}$.
\begin{figure}[htbp] 
	\centering
	\includegraphics[width=0.45\textwidth]{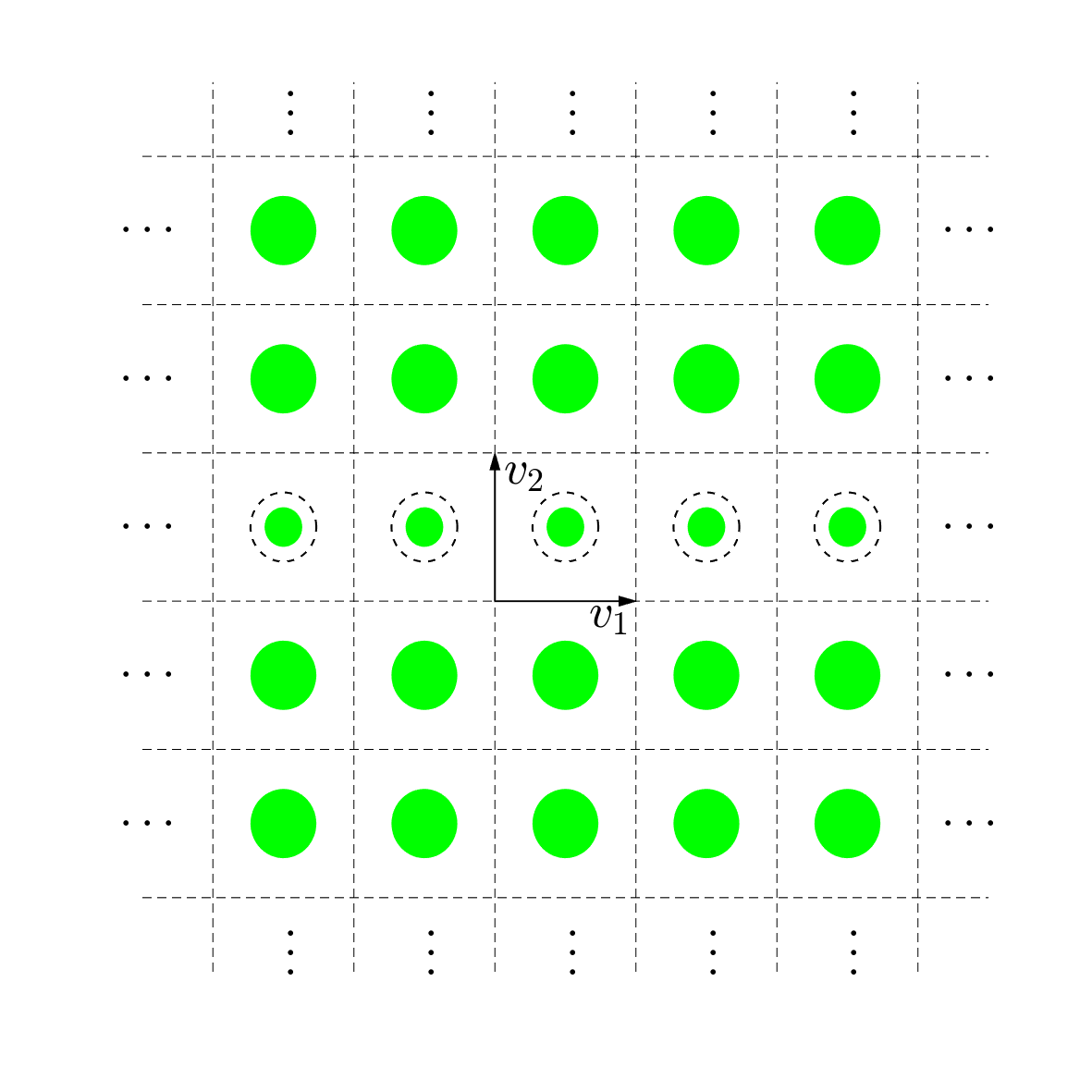}
	\hspace{-0.3cm}
	\includegraphics[width=0.49\textwidth]{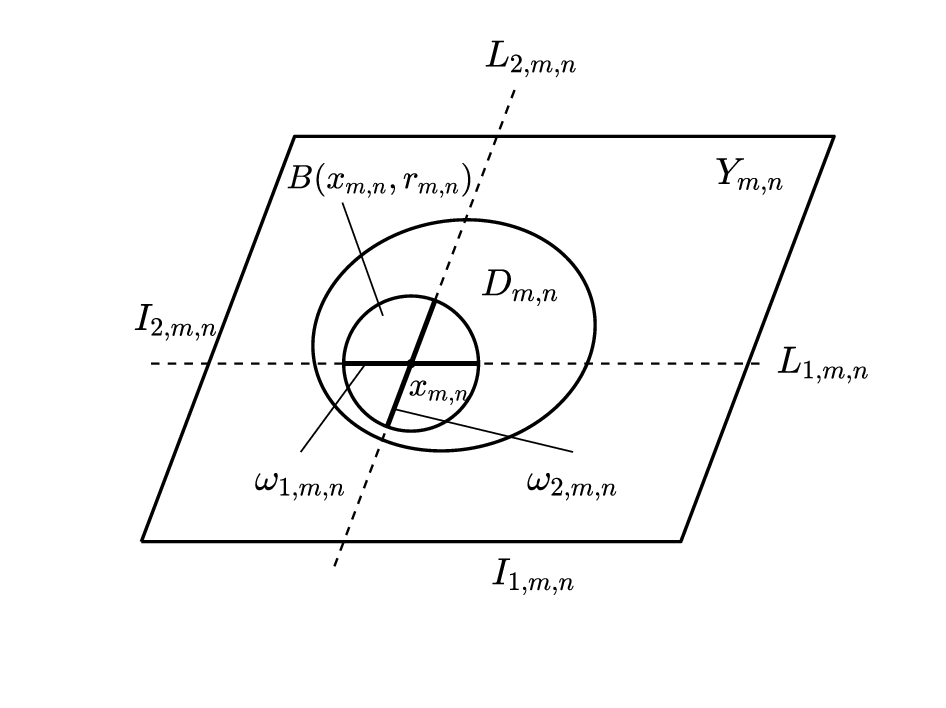}
	\caption{Left panel: an example of a waveguide system (adapted from \cite{jems2021}). Right panel: an example of oblique cell $Y_{m,n}$ containing a noncircular inclusion $D_{m,n}$. }
	\label{fig:CellExample}
\end{figure}

The left and bottom sides of $Y_{m,n}$ are given by 
\begin{align}
 I_{1,m,n}&\coloneqq\{(m+s)v_{_1}+nv_{_2}:s\in[0,1)\},\\
 I_{2,m,n}&\coloneqq\{mv_{_1}+(n+t)v_{_2}:t\in[0,1)\}.
\end{align}
Their lengths are $1$ and $\sqrt{b^2+h^2}$, respectively. 

Throughout this article, a tilde indicates that the corresponding object is obtained by collecting the relevant cells, sides or inclusions over all $n\in\mathbb{Z}$. 
\begin{equation}
	\wY_m\coloneqq\bigcup_{n\in\mathbb Z}Y_{m,n},\qquad
	\wI_{2,m}\coloneqq\bigcup_{n\in\mathbb Z}I_{2,m,n},\qquad
	\wD_m\coloneqq D_{\fu}\cap\wY_m.
\end{equation}
The infinite strip $\wY_m$ lies between the parallel lines $\wI_{2,m}$ and $\wI_{2,m+1}$. Translation by $mv_{_1}$ maps $\wD_0$ to $\wD_m$. We use $\wY\coloneqq\wY_0 $ and $\wD\coloneqq\wD_0 $ as the representative periodic strip and the inclusions it contains. 
 
The following additional assumptions are required to establish the decay result for the coefficients of the capacitance matrix.
\begin{assump}\label{def:decaycond2}
	The inclusion in each cell $ D_{m,n} $ satisfies the following conditions for all $ m,n\in \mathbb{Z} $:
	\begin{itemize}
		\item[(1)]  There exists $ r>0 $ such that $ B(y,r)\cap D_{m,n} = \emptyset $ for every $ y\in \partial Y_{m,n} $, where $ B(y,r) $ denotes the ball of radius $ r $ centered at $ y $; 
		\item[(2)] There exist balls $B(x_{m,n},r_{m,n})$ such that $ B(x_{m,n},r_{m,n})\Subset  D_{m,n} $ and 
		\[ \min_{m,n\in \mathbb{Z}} r_{m,n}\ge c>0 . \]
		\item[(3)] The boundaries of the inclusions are smooth and uniformly Lipschitz: there exist $\epsilon>0$ and $L\ge1$, independent of $m,n$, such that, for all $y\in\partial D_{m,n}$, there exists a rigid motion $\mathcal R_{m,n,y}$ sending $y$ to $0$,
		\[
		\mathcal R_{m,n,y}(D_{m,n})\cap Q_{\epsilon,L}
		=\{(s,t)\in Q_{\epsilon,L}:t>g_{m,n,y}(s)\},
		\]
		where $Q_{\epsilon,L}=(-\epsilon,\epsilon)\times(-2(1+L)\epsilon,2(1+L)\epsilon)$. The function $g_{m,n,y}:(-\epsilon,\epsilon)\to\mathbb R$ satisfies $g_{m,n,y}(0)=0$ and $\operatorname{Lip}(g_{m,n,y})\le L$.
	\end{itemize}
\end{assump}

Assumption~(1) gives a uniform separation from the cell boundaries, while assumption~(2) prevents the inclusions from shrinking to points by requiring an interior ball of uniformly positive radius. The point $x_{m,n}$ is the center of the interior ball $B(x_{m,n},r_{m,n})$. The right panel of \Cref{fig:CellExample} illustrates an example of unit cell and the inclusion within. These assumptions also imply that the boundary of the overall inclusions $\partial D_{\fu}$ is a disjoint union of countably many sets $\partial D_{m,n}$ and therefore has zero Lebesgue measure in $\mathbb{R}^2$. 

The lines $L_{i,m,n}$ that pass through the ball center parallel to the two cell sides, and their intersections with the ball $ \omega_{i,m,n} $, are denoted by
\begin{align}
 L_{i,m,n}\coloneqq x_{m,n}+\mathbb R v_{_i},\qquad \omega_{i,m,n}\coloneqq L_{i,m,n}\cap B(x_{m,n},r_{m,n}),\qquad i=1,2.
\end{align}
One can directly verify that
\[ |\omega_{i,m,n}|\ge 2c>0,\quad m,n\in \mathbb{Z},i=1,2. \]
The half-space formulas below use Euclidean tangential and normal coordinates. For a cell side $I$, we write
\[
d(x,I)\coloneqq\operatorname{dist}(x,\operatorname{aff}I),
\]
where $\operatorname{aff}I$ is the supporting line. Thus the distances in the Poisson multipliers are perpendicular distances, rather than distances to the endpoints of a segment. For lines parallel to $v_{_2}$, the tangential coordinate is arclength in the direction $v_{_2}/|v_{_2}|$. Each pair of parallel lines uses a common tangential origin. The half-strip problems in \Cref{sec:Decay_periodic} are understood after periodic extension in the $v_{_1}$ direction. Constants may depend on the fixed lattice geometry.
Since each inclusion $D_{m,n}$ lies within its corresponding cell $ Y_{m,n}$, there exists a positive constant $C$ such that
\[ 0<r<d(x_{m,n},I_{i,m,n})<C, \quad \forall\; m,n\in \mathbb{Z},i=1,2. \]

We let $\mathbb{H}^{s}(\partial D_{m,n})$ denote the standard Sobolev space of order $s$. Throughout this paper, we define the Hilbert space $\mathbb{H}^s(\partial \wD)$ on $ \partial \wD$ as 
\begin{gather}
	\begin{aligned}
		\mathbb{H}^s(\partial \wD)\coloneqq \Big\{ \phi &: \phi|_{\partial D_{0,n}}\in \mathbb{H}^s(\partial D_{0,n}) , \\ &\Vert \phi  \Vert_{\mathbb{H}^{s}(\partial\wD)} \coloneqq \Big( \sum_{n\in \mathbb{Z}} \Vert \phi|_{\partial D_{0,n}} \Vert^2_{\mathbb{H}^s(\partial D_{0,n})} \Big)^{\frac{1}{2}} <+\infty\Big\}.
	\end{aligned}
\end{gather} 
The duality pairing of $\mathbb{H}^{-\frac12}(\partial \wD)$ and $\mathbb{H}^{\frac12}(\partial \wD)$  is denoted by $\langle\cdot,\cdot\rangle$ and is conjugate-linear in the first argument. We define $\mathbb H^s(\partial D_{\fu})$ by the analogous Hilbert direct sum over $(m,n)\in\mathbb Z^2$. 

\subsection{Wave guiding in the subwavelength regime}

We consider the following Helmholtz system
\begin{equation}\label{eqn:OriginalProb}
	-\kappa\nabla \cdot \left( \rho^{-1} \nabla u \right) =\omega^2 u.
\end{equation}
Here, the density and bulk modulus functions $ \rho(x),\kappa(x) $ are given by 
\begin{equation}
	\rho(x) \coloneqq \rho_{_0}\chi_{_{\mathbb{R}^2\setminus D_{\fu}}}(x) + \rho_{_1}\chi_{_{D_{\fu}}}(x),\quad \kappa(x)\coloneqq \kappa_{_0}\chi_{_{{\mathbb{R}^2\setminus D_{\fu}}}}(x) + \kappa_{_1}\chi_{_{D_{\fu}}}(x),
\end{equation}
where $ \chi_{_{D_{\fu}}} $ and $ \chi_{_{\mathbb{R}^2\setminus D_{\fu}}} $ denote the characteristic functions of $ D_{\fu} $ and $ \mathbb{R}^2\setminus D_{\fu} $, respectively. The positive constants $ \rho_{_0},\kappa_{_0} $ and $ \rho_{_1},\kappa_{_1} $ represent the densities and bulk moduli outside and inside the inclusions $D_{\fu}$, respectively. The corresponding wave speeds outside and inside the inclusions are
\begin{equation}\label{eqn:mu_defini}
	\mu_{_0}\coloneqq  \sqrt{\frac{\kappa_{_0}}{\rho_{_0}}},\qquad \mu_{_1}\coloneqq\sqrt{\frac{\kappa_{_1}}{\rho_{_1}}}.
\end{equation}
And we assume that the contrast parameter $\delta$ satisfies
\begin{equation}
	\delta \coloneqq \frac{\rho_{_1}}{\rho_{_0}} \ll 1.
\end{equation} 

Throughout the high-contrast limit, the positive wave speeds $\mu_{_0}$ and $\mu_{_1}$ are fixed. The full-space operator associated with \eqref{eqn:OriginalProb}, acting in $\mathbb L^2(\mathbb R^2,\kappa^{-1}\du x)$, is given by
\begin{equation}\label{eqn:fullspaceoperator}
 \mathcal A_\delta u\coloneqq-\kappa\nabla\cdot(\rho^{-1}\nabla u),
\end{equation}
with domain
\begin{equation}\label{eqn:fullspacedomain}
 \mathcal D(\mathcal A_\delta)\coloneqq
 \{u\in\mathbb H^1(\mathbb R^2):\nabla\cdot(\rho^{-1}\nabla u)\in\mathbb L^2(\mathbb R^2)\}.
\end{equation}
This is the self-adjoint operator associated with the closed form $\int_{\mathbb R^2}\rho^{-1}\overline{\nabla u}\cdot\nabla v\,\du x$ on $\mathbb H^1(\mathbb R^2)$. By the Bloch decomposition in the $v_{_1}$ direction, we can define the quasi-periodic operator $\mathcal A^\alpha_\delta v\coloneqq-\kappa\nabla\cdot(\rho^{-1}\nabla v)$
with domain
\begin{equation}\label{eqn:fiberdomain}
	\begin{split}
		\mathcal D(\mathcal A^\alpha_\delta)\coloneqq
		\big\{v\in\mathbb H^1_\alpha(\wY):\;&\nabla\cdot(\rho^{-1}\nabla v)\in\mathbb L^2_{\mathrm{loc}}(\mathbb R^2),\\
		&\|\nabla\cdot(\rho^{-1}\nabla v)\|_{\mathbb L^2(\wY)}<\infty\big\}.
	\end{split}
\end{equation}
Here the space for quasi-periodic functions is given by 
\begin{equation}\label{eqn:quasiperiodicenergyspace}
\begin{split}
 \mathbb H^1_\alpha(\wY)\coloneqq
 \big\{v|_{\wY}:v\in\mathbb H^1_{\mathrm{loc}}(\mathbb R^2), v(x+mv_{_1})=\eu^{\iu\alpha m}v(x),\|v\|_{\mathbb H^1(\wY)}<\infty\big\}.
\end{split}
\end{equation}
The fiber acts in $\mathbb L^2(\wY,\kappa^{-1}\du x)$. A guided mode is a non-zero eigenfunction $u^\alpha\in\mathcal D(\mathcal A^\alpha_\delta)$ with eigenvalue $\omega^2$. It satisfies the following Helmholtz equations with jump conditions on the boundaries 
\begin{gather}\label{eqn:ProbFormq}
	\left\{\begin{aligned}
		&\Delta u^{\alpha} + \frac{\omega^2}{\mu_{_0}^2} u^{\alpha} = 0,\quad x\in \wY\setminus \overline{\wD},\\
		&\Delta u^{\alpha} +\frac{\omega^2}{\mu_{_1}^2} u^{\alpha} = 0,\quad x\in \wD, \\
		&u^{\alpha}|_{+}-u^{\alpha}|_{-}=0,\quad y\in \partial\wD,\\
		&\delta\left.\nu_y\cdot \nabla u^{\alpha} \right|_{+} - \left.\nu_y\cdot \nabla u^{\alpha}\right|_{-}=0,\quad y\in \partial \wD,\\
		&u^{\alpha}(x+mv_{_1}) = \eu^{\iu \alpha m}u^{\alpha}(x),\quad x\in \wY,m\in \mathbb{Z}.
	\end{aligned}\right.
\end{gather}   
The normal vector $\nu$ points out of each inclusion, and the signs $\cdot|_+$ and $\cdot|_-$ denote exterior and interior traces, respectively. 

Mathematically, we study eigenvalues $\omega^2\in\sigma_{\mathrm p}(\mathcal A^\alpha_\delta)$ for fixed $\alpha\ne0$ with $\omega^2=\mathcal O(\delta)$ as $\delta\to0$. This is equivalent to studying subwavelength frequencies $\omega=\mathcal O(\sqrt\delta)$. 
Here localization means square integrability over the transverse fiber $\wY$. Its solution can be represented using single-layer potentials:
\begin{gather}\label{eqn:RepresentationLayer}
	u^\alpha(x) = \left\{\begin{aligned}
		&\mathcal{S}^{\alpha,k_1}_{\wD}[\phi](x),\quad x\in \wD,\\
		&\mathcal{S}^{\alpha,k_0}_{\wD}[\psi](x),\quad x\in \wY\setminus \overline{\wD},
	\end{aligned}\right.
\end{gather}
for some densities $\phi,\ \psi \in \mathbb{H}^{-\frac{1}{2}}(\partial \wD)$, where the wave numbers are defined as
\[ k_0 \coloneqq \frac{\omega}{\mu_{_0}},\quad k_1 \coloneqq \frac{\omega}{\mu_{_1}}. \]
The single-layer potential $\mathcal{S}^{\alpha,k}_{\wD}$ is defined in \eqref{eqn:SingleLayer}. The densities $\phi$ and $\psi $ satisfy the following system of boundary integral equations:
\begin{gather}\label{eqn:BoundaryMatch}
	\left\{ \begin{aligned}
		&\mathcal{S}^{\alpha,k_1}_{\wD}[\phi](y) = \mathcal{S}_{\wD}^{\alpha,k_0}[\psi](y),\\
		&-\frac{\phi(y)}{2} + (\mathcal{K}^{-\alpha,k_1}_{\wD})^{\ast}[\phi](y) = \delta\left( \frac{\psi(y)}{2} + (\mathcal{K}^{-\alpha,k_0}_{\wD})^{\ast}[\psi](y) \right),
	\end{aligned}\right.\quad y\in \partial \wD. 
\end{gather}
Here, $(\mathcal{K}^{-\alpha,\omega}_{\wD})^{\ast} $ denotes the Neumann--Poincar\'e operator; see \Cref{apsec:prelim} for details. 

As a preliminary step, we define the quasi-periodic harmonic functions $ \{U_{q}^{\alpha}\}_{q\in\mathbb{Z}} $ as the solutions to the boundary value problem:
\begin{gather}\label{eqn:Harfuncp}
	\left\{\begin{aligned}
		&\Delta U^{\alpha}_{q}(x) = 0,\quad x \in \wY\setminus\overline{\wD},\\
		&U^{\alpha}_{q}(y)  = \delta_{nq},\quad y \in \partial D_{0,n},\\
		&U^{\alpha}_{q}(x+mv_{_1}) = \eu^{\iu \alpha m}U_{q}^{\alpha}(x),\quad x\in \wY\setminus\overline{\wD},m\in \mathbb{Z}. 
	\end{aligned}\right.
\end{gather}
We require $U_q^\alpha\in\mathbb H^1_\alpha(\wY\setminus\overline{\wD})$. Its existence and uniqueness are proved in Appendix~\ref{apsec:poincare}.
And the quasi-periodic capacitance matrix $ \widehat{\mathbf{C}}^{\alpha} = \{ \widehat{C}^{\alpha}_{n,q} \}_{n,q\in \mathbb{Z}}$ is defined by
\begin{equation}\label{eqn:PCapaMat}
	\widehat{C}_{n,q}^{\alpha}\coloneqq -\int_{\partial D_{0,n}} \nu_{y}\cdot \left.\nabla U_{q}^{\alpha}(y)\right|_{+}\du\sigma(y).
\end{equation}
To investigate the infinite quasi-periodic capacitance matrix $ \widehat{\mathbf{C}}^{\alpha}$, in particular its decay property, we first introduce the full capacitance matrix $\mathfrak{C}$ corresponding to \eqref{eqn:OriginalProb}. Let the nonnegative harmonic functions $ \{U_{p,q}\}_{p,q\in\mathbb{Z}} $ be the solutions to the following boundary value problem:
\begin{gather}\label{eqn:Harfunc}
	\left\{\begin{aligned}
		&\Delta U_{p,q}(x) = 0,\quad x \in \mathbb{R}^2\setminus \overline{D_{\fu}},\\
		&U_{p,q}(y)  = \delta_{pm}\delta_{qn},\quad y \in \partial D_{m,n}.\\
	\end{aligned}\right.
\end{gather}
We require $U_{p,q}\in\mathbb H^1(\mathbb R^2\setminus\overline{D_{\fu}})$. We prove existence and uniqueness in this class in Appendix~\ref{apsec:poincare}.
The full capacitance matrix $\mathfrak{C} \coloneqq \{ C_{m,n}^{p,q} \}_{m,n\in \mathbb{Z}}^{p,q\in \mathbb{Z}}$ is then given by 
\begin{equation}\label{eqn:FullCapaMat}
	C_{m,n}^{p,q} \coloneqq -\int_{\partial D_{m,n}} \nu_y\cdot \left.\nabla U_{p,q}(y)\right|_{+}\du\sigma(y).
\end{equation}
By the translation invariance of $D_{\fu}$, one has 
\[ U_{p,q}(x-mv_{_1}) = U_{p+m,q}(x). \]
And it follows from definition that 
\[ C^{p,q}_{m,n} = -\int_{\partial D_{0,n}} \nu_y\cdot \left.\nabla U_{p-m,q}(y)\right|_{+}\du\sigma(y)=C^{p-m,q}_{0,n}. \]

Our goal is to prove the following \Cref{thm:ExpConv} and \Cref{thm:ExpConv_quasi}, which establish the exponential decay of the full and quasi-periodic capacitance matrices, respectively. In particular, \Cref{thm:ExpConv_quasi} can be derived from \Cref{thm:ExpConv}, which will be proved in the subsequent sections.
\begin{theorem}\label{thm:ExpConv}
	There exist $C>0$, $\rho\in(0,1)$ such that for any $ (p,q)\in\mathbb{Z}^2$ and $(m,n)\in \mathbb{Z}^2\setminus \{(p,q)\} $, the elements of the full capacitance matrix $\mathfrak{C}$ satisfy:
	\begin{equation}\label{eqn:ExpConvergence}
		|C_{m,n}^{p,q}|\le C\min\{ \rho^{|m-p|},\rho^{|n-q|} \}.
	\end{equation}  
\end{theorem}

The following estimate for the quasi-periodic capacitance matrix is a direct consequence. 
\begin{theorem}\label{thm:ExpConv_quasi}
	There exist $C>0$, $\widetilde{\rho}\in(0,1)$ such that for any $\alpha\in [-\pi,\pi)$, $ q\in\mathbb{Z} $ and $n\in\mathbb{Z}\setminus \{q\}$, the elements of the quasi-periodic capacitance matrix $\widehat{\mathbf{C}}^\alpha$ satisfy: 
	\begin{equation}\label{eqn:ExpConvergence_quasi}
		|\widehat{C}^{\alpha}_{n,q}| \le C\widetilde{\rho}^{|n-q|}.
	\end{equation}
	
\end{theorem}
\begin{proof}
	We first establish the Bloch representation independently of \eqref{eqn:ExpConvergence}. The weak maximum principle gives $0\le U_{j,n}\le1$ and $0\le U_n^0\le1$. For every finite set $J\subset\mathbb Z$, bounded comparison on the periodically extended exterior gives
	\[
	0\le\sum_{j\in J}U_{j,n}\le U_n^0.
	\]
	Here and below we use the bounded maximum principle for planar domains; see \cite[Theorem 2.7]{Axler2001}. The increasing limit over $J$ is a bounded periodic harmonic function with the boundary values of $U_n^0$. The boundary values follow by comparison with the individual source field and with $U_n^0$. Bounded uniqueness therefore gives $\sum_jU_{j,n}=U_n^0$.
	For every $\alpha$, define
	\begin{equation}\label{eqn:Harmonic_quasi}
	\widetilde U_n^\alpha(x)\coloneqq\sum_{j\in\mathbb Z}\eu^{\iu\alpha j}U_{j,n}(x).
	\end{equation}
	The series converges absolutely pointwise, with $|\widetilde U_n^\alpha|\le U_n^0$. The nonnegative tails at $\alpha=0$ tend locally uniformly to zero. Interior and local boundary estimates then give convergence of the series and its normal derivatives near each fixed inclusion, after separating its possible source term. Translation of the series shows quasi-periodicity. Moreover, $\widetilde U_n^\alpha\in\mathbb L^2(\wY)$ by domination. Subtracting a cell-supported boundary lift and applying the Caccioppoli estimate with transverse cutoffs gives $\widetilde U_n^\alpha\in\mathbb H^1_\alpha(\wY\setminus\overline{\wD})$: the cutoff terms are bounded by its $\mathbb L^2$ norm and the norm of the lift, while the side terms cancel by quasi-periodicity. Uniqueness yields $\widetilde U_n^\alpha=U_n^\alpha$.
	Taking the flux on $\partial D_{0,n}$ is now justified and gives
	\begin{equation}\label{eqn:transform}
	\widehat C_{n,q}^\alpha=\sum_{p\in\mathbb Z}\eu^{\iu\alpha p}C_{0,n}^{p,q}.
	\end{equation}
	The coefficient series is absolutely convergent: except for the possible source term, its terms at $\alpha=0$ have one sign by the maximum principle, and their sum is the finite flux of $U_q^0$. Thus \eqref{eqn:transform}, including its use in \Cref{lem:SimpleForm_estimate}, does not depend on either decay theorem.
	
	From \Cref{thm:ExpConv}, we directly have the estimate:
	\[
	|\widehat{C}^{\alpha}_{n,q}| \le C\sum_{p\in \mathbb{Z}}\min\{ \rho^{|p|},\rho^{|n-q|} \}\le \widetilde{C}\rho^{|n-q|}(1+|n-q|).
	\]
	for all $n\neq q$. The second inequality holds since
	\begin{gather*}
		\begin{aligned}
			\sum_{p\in \mathbb{Z}}\min\{ \rho^{|p|},\rho^{|n-q|} \} &=  \sum_{|p|\le |n-q|}\rho^{|n-q|}+\sum_{|p|> |n-q|} \rho^{|p|}\\
			&= \rho^{|n-q|}\Big( 2|n-q|+1+\frac{2\rho}{1-\rho} \Big).
		\end{aligned}
	\end{gather*}
	And there exists $M\ge 0$ such that
	\[  \rho^{|n-q|/2}(1+|n-q|)\le M<+\infty. \]
	 We let $\widetilde{\rho} = \sqrt{\rho}$ and it follows that
	\[ |\widehat{C}^{\alpha}_{n,q}| \le \widetilde{C}M\widetilde{\rho}^{|n-q|}. \]
	Here the constant $\widetilde{C}$ is independent of $n,q$.
\end{proof}
For completeness, choose a smooth lift $\eta_n$ equal to one on $D_{0,n}$ and supported inside $Y_{0,n}$, with uniformly bounded $\mathbb H^1$ norm. The separation assumption allows this choice independently of $n$ and $\alpha$. The energy-minimizing property gives
\[
0\le\widehat C_{n,n}^\alpha
=\int_{\wY\setminus\overline{\wD}}|\nabla U_n^\alpha|^2\du x
\le\int_{\wY}|\nabla\eta_n|^2\du x\le C.
\]
Together with \eqref{eqn:ExpConvergence_quasi}, this bounds all row and column sums. Green's identity also gives
\[
\widehat C_{n,q}^\alpha
=\int_{\wY\setminus\overline{\wD}}\overline{\nabla U_n^\alpha}\cdot\nabla U_q^\alpha\du x.
\]
Consequently, the matrix is Hermitian and nonnegative, and the Schur test yields the following corollary.

\begin{corollary}
	For $\alpha\in [-\pi,\pi)$, the quasi-periodic capacitance matrix $\widehat{\mathbf{C}}^{\alpha}$ is a bounded operator on $l^2(\mathbb{Z})$. Here the action on vector $\mathbf{v}\in l^2(\mathbb{Z})$ is given by: 
	\begin{equation}
		\widehat{\mathbf{C}}^{\alpha}\mathbf{v}\coloneqq \Big\{\sum_{q\in\mathbb{Z}}\widehat{C}^{\alpha}_{n,q}v_{q}\Big\}_{n\in\mathbb{Z}}.
	\end{equation}
\end{corollary}

\begin{corollary}\label{cor:banded_operator_error}
	Define the banded matrix for $N\in\mathbb N_0$
	\begin{equation}\label{eqn:truncate_mat}
		\widehat{\mathbf{C}}_{t,N}^{\alpha}
		\coloneqq
		\left\{\begin{aligned}
			&\widehat C_{n,q}^{\alpha},|n-q|\le N,\\
			&0,|n-q|>N.
		\end{aligned}\right.
	\end{equation}
	With $C$ and $\widetilde\rho$ from \Cref{thm:ExpConv_quasi}, we have the following estimate
	\[
	\sup_{\alpha\in[-\pi,\pi)}\big\|\widehat{\mathbf C}^{\alpha}-\widehat{\mathbf{C}}^{\alpha}_{t,N}\big\|_{\mathcal L(l^2(\mathbb Z))}
	\le\frac{2C\widetilde\rho^{N+1}}{1-\widetilde\rho}.
	\]
\end{corollary}
\begin{proof}
	Let $\mathbf{E} \coloneqq \widehat{\mathbf C}^{\alpha}-\widehat{\mathbf{C}}^{\alpha}_{t,N} $. By \Cref{thm:ExpConv_quasi},
	\[
	\max\Big\{\sup_{n\in\mathbb{Z}}\sum_{q\in\mathbb{Z}}|E_{n,q}|,
	\sup_{q\in\mathbb{Z}}\sum_{n\in\mathbb{Z}}|E_{n,q}|\Big\}
	\le2C\sum_{j=N+1}^{\infty}\widetilde\rho^j
	=\frac{2C\widetilde\rho^{N+1}}{1-\widetilde\rho}.
	\]
	For $\mathbf v\in l^2(\mathbb{Z})$, Cauchy-Schwarz inequality gives
	\[
	\|\mathbf E\mathbf v\|_{l^2(\mathbb{Z})}^2
	\le\sum_{n\in\mathbb{Z}}\Big(\sum_{q\in\mathbb{Z}}|E_{n,q}|\Big)
	\Big(\sum_{q\in\mathbb{Z}}|E_{n,q}||v_q|^2\Big)
	\le\frac{4C^2\widetilde\rho^{2N+2}}{(1-\widetilde\rho)^2}\|\mathbf v\|_{l^2(\mathbb{Z})}^2.
	\]
\end{proof}

\begin{remark}
	If the region $D_{\fu}$ satisfies the $v_{_2}$-periodicity condition, we can define the full periodic capacitance matrix as
	\[C^{\alpha,\beta} \coloneqq \sum_{q\in\mathbb{Z}} \eu^{\iu \beta q} \widehat{C}^{\alpha}_{0,q}, \]
	where $(\alpha, \beta) \in [-\pi, \pi)^2$. It plays a significant role in the analysis of the band structure of periodic systems with multiple inclusions; see, for example, \cite{Ammari2018}.
\end{remark}

%

\subsection{Sufficient conditions for decay}	
The proof of \Cref{thm:ExpConv} relies on the key estimate provided by \Cref{lem:SimpleForm_estimate}.
We first establish some fundamental properties of the elements of the full capacitance matrix $\mathfrak{C}$, whose signs follow from the maximum principle and Hopf's boundary point lemma.
\begin{proposition}\label{prop:PropertyCapa}
	The elements of the full capacitance matrix $\mathfrak{C}$ satisfy the following properties:
	\begin{itemize}
		\item[(1)] For $ (p,q)\neq(m,n)  $, 
		\[ C_{m,n}^{p,q}=-\int_{\partial D_{m,n}}\nu_y\cdot \nabla U_{p,q}(y)\du\sigma(y)<0. \]
		\item[(2)] For any $ (p,q)\in \mathbb{Z}^2  $, 
		\[ C_{p,q}^{p,q}=-\int_{\partial D_{p,q}}\nu_y\cdot \nabla U_{p,q}(y)\du\sigma(y)>0. \]
		Moreover, the following holds:
		\[ C_{p,q}^{p,q}=-\sum_{(p,q)\neq (m,n)}C_{m,n}^{p,q} ,\]
		with an absolutely convergent series on the right-hand side.
	\end{itemize}
\end{proposition}
\begin{proof}
The exterior has a common unbounded component containing the cell-boundary collars. Each inclusion has an outer boundary adjacent to this component. The strong maximum principle and Hopf's lemma there give the strict signs of the integrated fluxes. In any bounded cavity the field is constant when its boundary belongs to a single inclusion, so this does not change the signs. Green's identity gives $C_{m,n}^{p,q}=C_{p,q}^{m,n}$.

For conservation, every finite sum of source fields lies between zero and one by bounded comparison. Its increasing limit has boundary value one on every inclusion and is bounded and harmonic. Bounded uniqueness in the plane gives $\sum_{p,q}U_{p,q}=1$. Near a fixed inclusion, the tails have zero boundary value and tend locally uniformly to zero; local boundary estimates justify taking the normal derivative. Thus $\sum_{p,q}C_{m,n}^{p,q}=0$, with absolute convergence since all but one term have the same sign. Symmetry proves the asserted identity.
\end{proof}

We now prove the following key lemma, which provides sufficient conditions for the decay of off-diagonal elements of the full capacitance matrix. 
\begin{lemma}\label{lem:SimpleForm_estimate}
	Suppose that for all $ p,q\in \mathbb{Z} $, the following estimates hold:
	\begin{align}
		\int_{I_{1,0,q-n}}U^{0}_{q}\du \sigma(y)\le C_1\rho_0^{n},\quad 
		\int_{I_{1,0,q+n+1}}U^{0}_{q}\du \sigma(y)\le C_1\rho_0^n,\quad n\ge 0, \label{eqn:Decay_Condition1}\\
		\int_{\wI_{2,p-m}}U_{p,q}\du \sigma(y) \le C_2\rho_0^{m}, \quad \int_{\wI_{2,p+ m+1}}U_{p,q}\du \sigma(y) \le C_2\rho_0^{m},\quad m \ge 0.\label{eqn:Decay_Condition2}
	\end{align}
	Here, $\rho_0\in (0,1) $ is a positive constant, and $ C_1,C_2 $ are positive real constants independent of $ m,n,p,q $. The harmonic functions $ U^{0}_{q} $ and $ U_{p,q} $ are determined from \eqref{eqn:Harfuncp} and \eqref{eqn:Harfunc}, respectively. Then the estimate \eqref{eqn:ExpConvergence} holds.
\end{lemma}
\begin{proof}
	By the geometrical assumptions and relabeling, we set $p=q=0$. We first investigate \eqref{eqn:Decay_Condition1}. Apply the divergence theorem to the finite perforated strips $\bigcup_{l=n+1}^{N}(Y_{0,l}\setminus\overline{D_{0,l}})$ and their lower counterparts. The periodic side fluxes cancel. On the far cross-section, \Cref{lem:Harnackineq} and \eqref{eqn:Decay_Condition1} bound the flux by $C\rho_0^N$, so it vanishes as $N\to\infty$. The hole fluxes have one sign, which also justifies their limiting sum. Thus
	\begin{align*}
	\int_{I_{1,0,n+1}}\nu_y\cdot\nabla U_0^0\du\sigma
	&=-\sum_{l>n}\widehat C_{l,0}^0,\\
	\int_{I_{1,0,-n}}\nu_y\cdot\nabla U_0^0\du\sigma
	&=-\sum_{l<-n}\widehat C_{l,0}^0.
	\end{align*}
	Here $\nu_y$ on each cross-section points out of the corresponding upper or lower region; on inclusion boundaries it retains the convention in \eqref{eqn:FullCapaMat}.
	The smoothness of the harmonic function $U_0^0$ implies the bounds
	\begin{align*}
		\left|\int_{I_{1,0,n+1}} \nu_{y}\cdot \nabla U^0_0(y)\du\sigma(y)\right| \le \int_{I_{1,0,n+1}}|\nabla U^0_0(y)|\du\sigma(y),\\
		\left|\int_{I_{1,0,-n}} \nu_{y}\cdot \nabla U^0_0(y)\du\sigma(y)\right| \le \int_{I_{1,0,-n}}|\nabla U^0_0(y)|\du\sigma(y).
	\end{align*}
	Applying Harnack's inequality (\Cref{lem:Harnackineq}) and noting that $U^0_0$ is positive, we obtain
	\begin{equation}\label{eqn:value_control}
		\bigg|\sum_{l> n} \widehat{C}^0_{l,0}\bigg| \le \frac{4}{r}\int_{I_{1,0,n+1}}U^0_0(y)\du\sigma(y),\quad \bigg|\sum_{l< -n} \widehat{C}^0_{l,0}\bigg| \le \frac{4}{r}\int_{I_{1,0,-n}}U^0_0(y)\du\sigma(y).
	\end{equation}  
	By \Cref{prop:PropertyCapa}, $C^{0,0}_{k,l}<0 $ for all $l>n$ and $k\in \mathbb{Z}$, therefore 
	\[
	| C^{0,0}_{m,n+1} |\le -\sum_{k\in \mathbb{Z}}\sum_{l> n} C^{0,0}_{k,l} = -\sum_{k\in \mathbb{Z}}\sum_{l> n} C^{-k,0}_{0,l}.
	\]
	Combining this with \eqref{eqn:transform}, we have 
	\[
	| C^{0,0}_{m,n+1} |\le -\sum_{l>n}\widehat{C}^0_{l,0},\quad n\ge 0.
	\]
	And similarly we also have for $n\ge 0$ 
	\[
	| C^{0,0}_{m,-(n+1)} |\le -\sum_{l< -n} \widehat{C}^0_{l,0}.
	\]
	Combining these with the previous inequalities \eqref{eqn:value_control} and \eqref{eqn:Decay_Condition1}, we conclude that
	\[ | C^{0,0}_{m,n+1} |\le C_3 \rho_0^n,\quad | C^{0,0}_{m,-(n+1)} |\le C_3\rho_0^n,\quad n\ge 0. \]
	Defining $C_4 = C_3\rho_0^{-1}$, we obtain the final estimates in the $x_2$ direction:
	\begin{equation}\label{eqn:y_dir_decay}
		| C^{0,0}_{m,n} |\le C_4 \rho_0^{n},\quad | C^{0,0}_{m,-n} |\le C_4\rho_0^n,\quad n\ge 1.
	\end{equation}
	
	We now turn to \eqref{eqn:Decay_Condition2}. Apply the divergence theorem first on a finite union of cells with $m+1\le l\le M$. Let the upper and lower transverse cutoffs tend to infinity. The bound $U_{0,0}\le U_0^0$ and \eqref{eqn:Decay_Condition1}, together with the collar gradient estimate, make those fluxes vanish for each fixed $M$. Next let $M\to\infty$ and use \eqref{eqn:Decay_Condition2} to remove the far longitudinal flux. The analogous argument on the left gives
	\begin{align*}
	\int_{\wI_{2,m+1}}\nu_y\cdot\nabla U_{0,0}\du\sigma
	&=-\sum_{l>m}\sum_{n\in\mathbb Z}C_{l,n}^{0,0},\\
	\int_{\wI_{2,-m}}\nu_y\cdot\nabla U_{0,0}\du\sigma
	&=-\sum_{l<-m}\sum_{n\in\mathbb Z}C_{l,n}^{0,0}.
	\end{align*}
	Applying Harnack's inequality again, we find
	\begin{align*}
		\bigg|\sum_{l> m}\sum_{n\in\mathbb{Z}} C^{0,0}_{l,n}\bigg| = \left|\int_{\wI_{2,m+1}} \nu_{y}\cdot \nabla U_{0,0}(y)\du\sigma(y) \right|&\le\frac{4}{r} \int_{\wI_{2,m+1}}  U_{0,0}(y)\du\sigma(y) ,\\
		\bigg|\sum_{l< -m}\sum_{n\in\mathbb{Z}} C^{0,0}_{l,n}\bigg| =  \left|\int_{\wI_{2,-m}} \nu_{y}\cdot \nabla U_{0,0}(y)\du\sigma(y)\right| &\le \frac{4}{r} \int_{\wI_{2,-m}}  U_{0,0}(y)\du\sigma(y).
	\end{align*}
	Consequently, we have
	\[ |C^{0,0}_{m+1,n}|\le C_5\rho^m_0,\quad |C^{0,0}_{-(m+1),n}| \le C_5\rho^m_0,\quad m\ge 0. \]
	Defining $C_6 = C_5\rho_0^{-1}$, we obtain the estimates in the $x_1$ direction: 
	\begin{equation}\label{eqn:x_dir_decay}
		|C^{0,0}_{m,n}|\le C_6\rho^m_0,\quad |C^{0,0}_{-m,n}| \le C_6\rho^m_0,\quad m\ge 1.
	\end{equation}
	For $m,n\ne0$, combining \eqref{eqn:y_dir_decay} and \eqref{eqn:x_dir_decay} gives \eqref{eqn:ExpConvergence}. If one index is zero, the estimate in the other direction gives the same conclusion. The source diagonal is excluded.
\end{proof}
The proof of \Cref{thm:ExpConv} is then reduced to proving the estimates \eqref{eqn:Decay_Condition1} and \eqref{eqn:Decay_Condition2}. This will be done in the next two sections.

\section[Proof of the transverse decay estimate]{\texorpdfstring{Proof of estimate \eqref{eqn:Decay_Condition1}}{Proof of the transverse decay estimate}}\label{sec:Decay_periodic}

This section establishes the exponential decay of the integral of the periodic harmonic function $U_q^0$ on $I_{1,0,n}$ as $n \to \pm \infty$, i.e., estimate \eqref{eqn:Decay_Condition1}. 

For notational simplicity, we define:
\[
\begin{gathered}
 Y_{n} \coloneqq Y_{0,n},\quad D_{n} \coloneqq D_{0,n},\quad I_{n} \coloneqq I_{1,0,n},\\
 L_{n} \coloneqq L_{1,0,n},\quad \omega_{n} \coloneqq \omega_{1,0,n},\quad x_n\coloneqq x_{0,n}.
\end{gathered}
\]
for all $n\in\mathbb{Z}$ to simplify the notation. The distance between $x_n$ and $I_n$, denoted by $d(x_n,I_n)$, satisfies
\[ d(x_n,I_n)\ge r+c>0. \] 
This fact follows from \Cref{def:decaycond2}. Throughout this section, we let $q = 0$. Similar estimate holds for general $q$ by relabeling and the geometrical assumptions.\par 
By the uniqueness theorem, the restriction of $ U_0^0 $ to $ \cup_{n \ge 1} (Y_n \setminus D_n) $ coincides with the unique solution to the following boundary value problem with periodic boundary conditions:
\begin{gather}\label{eqn:truncateProblem}
	\left\{\begin{aligned}
		&\Delta \widetilde{U}_{+}(x) = 0,\quad x \in \cup_{n\ge 1} (Y_{n}\setminus D_{n}),\\
		&\widetilde{U}_{+}(y)  = 0,\quad y \in \cup_{n\ge1}\partial D_{n},\\
		&\widetilde{U}_{+}(y)  = U_0^0\big|_{I_{1}}(y),\quad y \in I_{1}, \\
		&\widetilde{U}_{+}(x + mv_{_1}) = \widetilde{U}_{+}(x),\quad x \in \cup_{n\ge 1} Y_{n},\; m\in \mathbb{Z}.
	\end{aligned}\right.
\end{gather}
Since $U_0^0\big|_{I_{1}}(y)$ is a positive smooth function, it belongs to
$\mathbb{L}^1(I_{1})\cap \mathbb{L}^2(I_{1})$. In the following, we let $\widetilde{U}_{+}$ be $0$ inside the inclusions $\cup_{n\ge 1}D_{n}$.

\subsection{Truncation method} 
In this part, we introduce the truncation method to estimate $\widetilde{U}_{+}$ on $\cup_{n\ge 1} (Y_{n}\setminus D_{n})$. Before providing details, note that the line $L_n$ divides the cell $Y_{n}$ into two parts, denoted by $Y_{n,+}$ and $Y_{n,-}$, as shown in \Cref{fig:TruncaMethod}. 
\begin{figure}[htbp] 
	\centering
	\includegraphics[width = 0.5\textwidth]{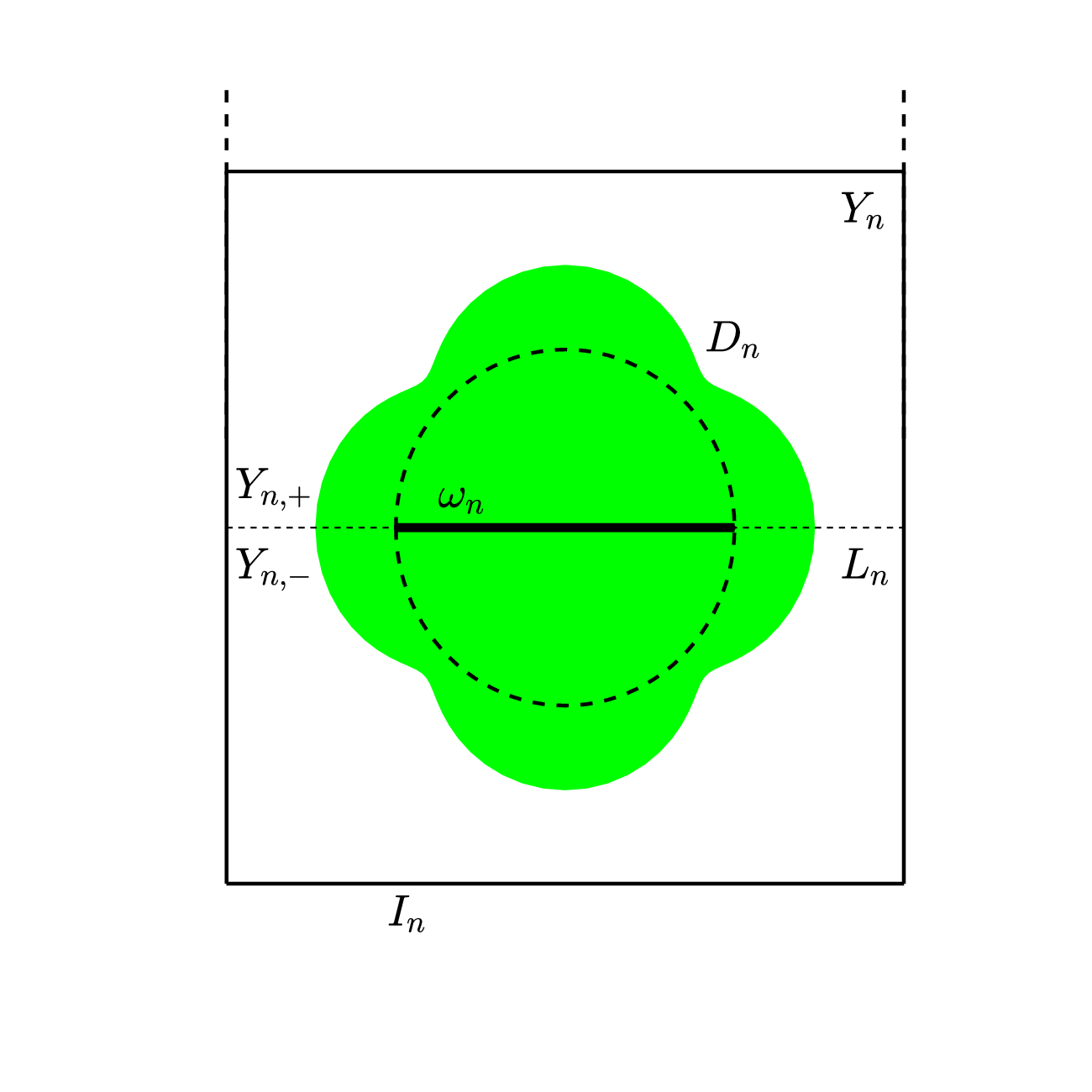}
	\caption{Cell $Y_n$. The line $L_n$ divides the cell into two parts, denoted as $Y_{n,+}$ and $Y_{n,-}$.}
	\label{fig:TruncaMethod}
\end{figure}
The procedure can be divided into three recursive steps. Given the boundary value $ \widetilde{g}_{_{1}}= U_0^0\big|_{I_{1}}>0 $ on $ I_{1} $, for $ n\ge 1 $,
\begin{itemize}
	\item[(1)] First, solve the following boundary value problem for the Laplace equation on the half space:
	\begin{gather}
		\left\{\begin{aligned}
			&\Delta \wcV_n(x) = 0,\quad x\in \cup_{l\ge n} Y_l,\\
			& \wcV_n(y) = \widetilde{g}_{_n}(y),\quad y\in I_{n},\\
			& \wcV_n(x+mv_{_1}) =  \wcV_n(x),\quad x\in \cup_{l\ge n} Y_{l},\; m\in\mathbb{Z},\\
			&\sup_{x\in \cup_{l\ge n} Y_l}|\wcV_n|<+\infty.
		\end{aligned}\right.
	\end{gather}
	Due to the smoothness of $\wcV_n$, one can take its trace on $L_{n}$
	in $ Y_{n} $, denoted by $\widetilde{h}_{n}$.  
	
	\item[(2)] Next, solve another boundary value problem on the half space:
	\begin{gather}
		\left\{\begin{aligned}
			&\Delta \wcW_n(x) = 0,\quad x\in Y_{n,+}\cup(\cup_{l\ge n+1} Y_{l}),\\
			& \wcW_n(y) = \widetilde{h}_n(1-\chi_{_n}) ,\quad y\in L_{n},\\
			& \wcW_n(x+mv_{_1}) =  \wcW_n(x),\quad x\in Y_{n,+}\cup(\cup_{l\ge n+1} Y_{l}),\; m\in\mathbb{Z},\\
			&\sup_{x\in Y_{n,+}\cup(\cup_{l\ge n+1} Y_{l})}|\wcW_n|<+\infty.
		\end{aligned}\right.
	\end{gather}
	Here $\chi_{_n}\in[0,1]$ is smooth and periodic on $L_n$, with its restriction to one cell satisfying 
	\[ \chi_{_n}(y) =1,\quad y\in \omega_n,\quad \operatorname{supp}\chi_{_n}\cap Y_n \subset L_n\cap D_n. \] 
	From this, its trace on $I_{n+1}$ can also be taken, denoted by $\widetilde{g}_{_{n+1}}$.
	\item[(3)] Repeat Steps (1) and (2).
\end{itemize}
The left and right panels of \Cref{fig:TruncateMethodVW} illustrate the truncation method and the functions $\widetilde{\mathcal{V}}_n$ and $\widetilde{\mathcal{W}}_n$, respectively. Their existence and uniqueness are discussed in \Cref{apsec:Poisson}.
\begin{figure}[htbp] 
	\centering
	\includegraphics[width = 0.45\textwidth]{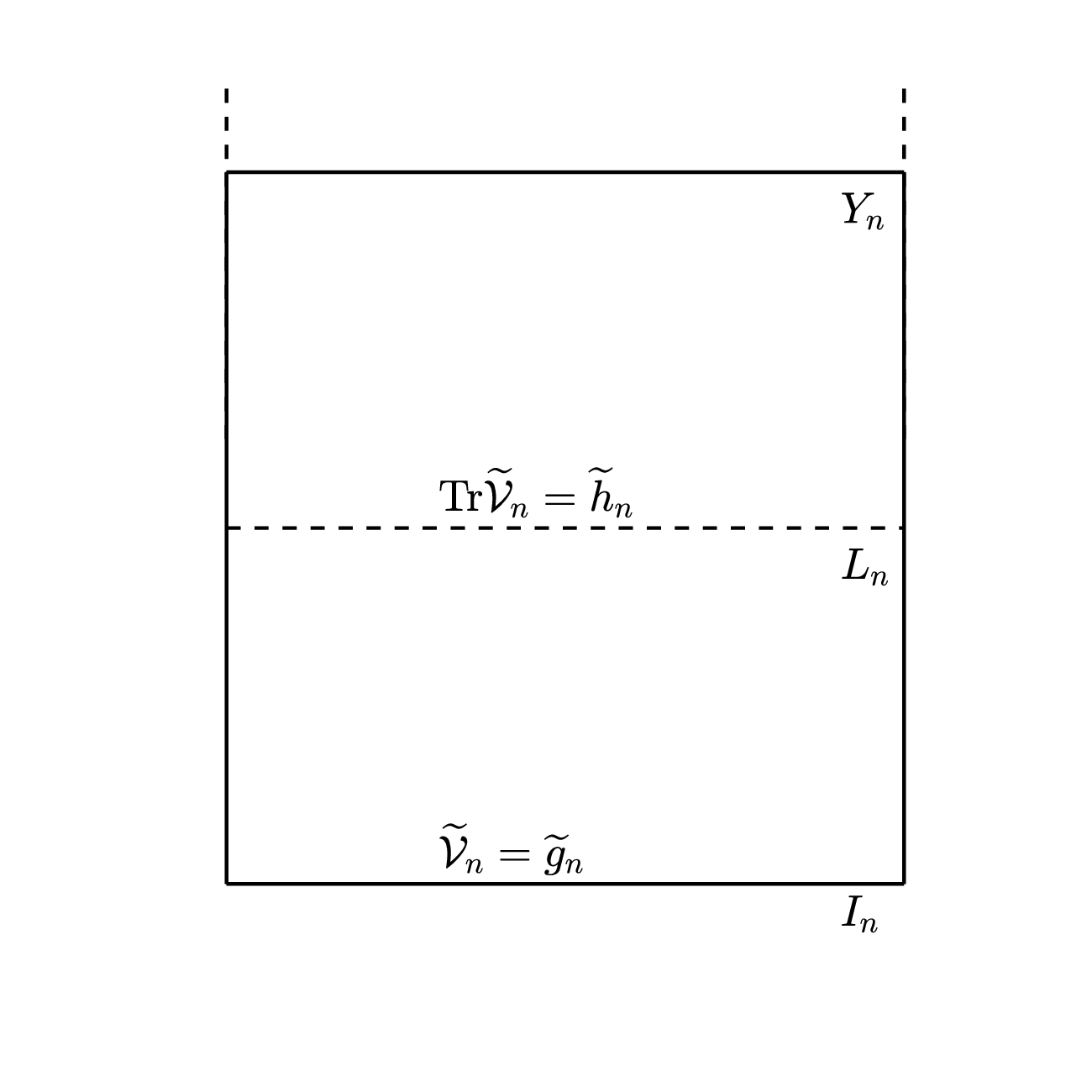}
	\hspace{-0.4cm}
	\includegraphics[width = 0.45\textwidth]{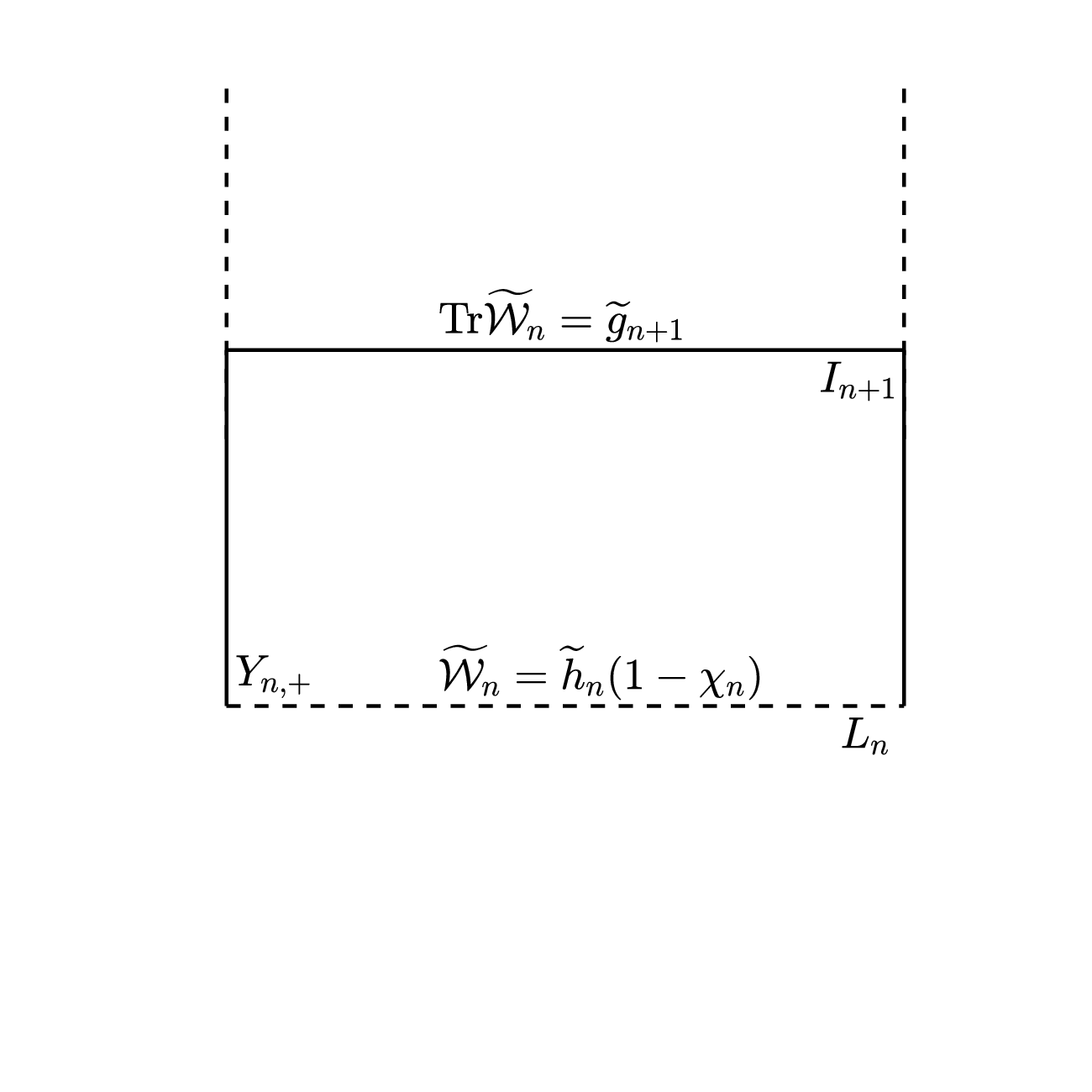}
	\caption{Left panel: Illustration of $\wcV_n$ and the derivation of $\widetilde{h}_n$. Right panel: Illustration of $\wcW_n$ and the derivation of $\widetilde{g}_{n+1}$.}
	\label{fig:TruncateMethodVW}
\end{figure}
\begin{remark}
	We emphasize that the equations for $\{ \wcV_n,\wcW_n \}_{n\ge 1}$ are solved on the half space without inclusions $\{ D_n \}_{n\ge 1}$.
	Therefore, they can be explicitly represented by Fourier series; see Appendix~\ref{apsec:Poisson}. 
\end{remark}
Based on this, we can prove the following theorem.
\begin{theorem}\label{thm:Control_of_Trunc}
	For the solutions $\{\wcV_n,\wcW_n\}_{n\ge1}$ defined by the truncation method and every $n\ge1$, 
	\[ \widetilde{U}_{+}(y)\le \widetilde{g}_{_n}(y),\quad y\in I_{n}. \]
\end{theorem}
\begin{proof}
	We prove this by induction. First, $ \widetilde{U}_+(y)= \widetilde{g}_{_1}(y) $ on $ I_{1} $. Assuming $ \widetilde{U}_{+}(y)\le \widetilde{g}_{_n}(y) $ on $I_{n}$,
	consider the harmonic function $ \wcV_n-\widetilde{U}_{+} $, which satisfies
	\begin{gather}
		\left\{\begin{aligned}
			&\Delta (\wcV_n-\widetilde{U}_{+}) = 0,\quad  x \in\cup_{l\ge n} \big(Y_{l} \setminus \overline{ D_{l}}\big),\\
			&\wcV_n-\widetilde{U}_{+} \ge 0,\quad y \in I_{n},\\
			&\wcV_n-\widetilde{U}_{+} > 0,\quad y\in \partial \big(\cup_{l\ge n}  D_{l}\big),\\
			&\wcV_n-\widetilde{U}_{+} \text{ satisfies the periodic boundary condition.} 
		\end{aligned}\right.
	\end{gather}
		Since $\widetilde g_{_n}\ge0$ is not identically zero, positivity of the Poisson kernel gives $\wcV_n>0$. Hence
		\[
		\wcV_n(y)-\widetilde U_+(y)=\wcV_n(y)>0, \quad  y\in \partial D_l, l\ge n.
		\]
		By maximum principle, $\widetilde U_+\le\wcV_n$ in the perforated region $\cup_{l\ge n} \big(Y_{l} \setminus \overline{ D_{l}}\big)$. Then taking its trace on $L_n$ gives 
		\[
		\widetilde U_+|_{L_n}\le\wcV_n|_{L_n} =\widetilde h_n.
		\]
		By the support condition on $\chi_{_n}$ and $\widetilde U_+=0$ in $D_n$, we have
		\[
		\widetilde h_n(1-\chi_{_n})-\widetilde U_+
		=
		\begin{cases}
			\widetilde h_n-\widetilde U_+\ge0,
			&\text{on }L_n\setminus\operatorname{supp}\chi_{_n},\\
			\widetilde h_n(1-\chi_{_n})\ge0,
			&\text{on }\operatorname{supp}\chi_{_n}.
		\end{cases}
		\]
	Next, consider the function $\wcW_n-\widetilde{U}_{+}$. By definition, it can be verified that 
	\begin{gather}
		\left\{\begin{aligned}
			&\Delta (\wcW_n-\widetilde{U}_{+}) = 0,\quad x \in (Y_{n,+}\setminus \overline{D_n})\cup(\cup_{l\ge n+1} Y_{l}\setminus \overline{ D_l}),\\
			&\wcW_n-\widetilde{U}_{+} \ge 0,\quad y \in L_n,\\
			&\wcW_n-\widetilde{U}_{+} \ge 0,\quad y\in \partial  D_{n}\cap Y_{n,+},\\
			& \wcW_n-\widetilde{U}_{+} \ge 0,\quad y\in \partial  \big(\cup_{l\ge n+1} D_{l}\big),\\
			&\wcW_n-\widetilde{U}_{+} \text{ satisfies the periodic boundary condition.}
		\end{aligned}\right.
	\end{gather}
	Hence, we conclude that $ \wcW_n\ge \widetilde{U}_{+} $ for all $ x\in Y_{n,+}\setminus D_n $. Denoting the trace of $ \wcW_n $ on $ I_{n+1} $ by $ \widetilde{g}_{_{n+1}} $, we obtain $ \widetilde{g}_{_{n+1}}\ge \widetilde{U}_{+} $. This completes the proof.
\end{proof}
A direct consequence of the above theorem is 
\[ \int_{I_{n}} U_0^0\,\du\sigma(y) = \int_{I_{n}} \widetilde{U}_{+}\,\du\sigma(y) \le \int_{I_{n}} \widetilde{g}_{_{n}}\,\du\sigma(y),\quad n\ge1.  \]
The theorem also implies that, for $n\ge1$,  
\begin{equation}
	U^0_0(x)\left\{\begin{aligned}
		&\le \wcV_n(x),\quad x\in Y_{n,-},\\
		& \le \wcW_n(x),\quad x\in Y_{n,+}.
	\end{aligned}\right. 
\end{equation}
In the following, we estimate the $\mathbb{L}^1(I_n)$ norm of $\{\widetilde{g}_{_n}\}_{n\ge 1}$ as $n\to +\infty$.
\subsection[Exponential decay of the transverse integral]{\texorpdfstring{Exponential decay of the integral on $I_{n}$}{Exponential decay of the transverse integral}}
In this subsection, we prove the decay of the integrals of $ U_0^0 $ on $I_{\pm n}$. This relies substantially on the following theorem characterizing the decay of the integrals of $\widetilde{g}_n$. 
\begin{theorem}\label{thm:DecayPer}
	Given the periodic functions $\{ \widetilde{g}_{_n} \}_{n\ge 1}$ determined by the truncation method, we have 
	\begin{equation}\label{eqn:Decay_estimate1}
		\int_{I_{n+1}}\widetilde{g}_{_{n+1}}\du\sigma(y) = \int_{I_{n}}\widetilde{g}_{_n}(1-\widetilde{V}_{n})\du\sigma(y).
	\end{equation}
	Here, $\widetilde{V}_{n}$ is the \emph{auxiliary function} determined by the following boundary value problem on the lower half space:
	\begin{gather}
		\left\{\begin{aligned}
			&\Delta \widetilde{V}_n = 0,\quad x\in Y_{n,-}\cup \Big(\cup_{l\le n-1}Y_l\Big),\\
			&\widetilde{V}_n = \chi_{_n},\quad y\in L_n,\\
			&\widetilde{V}_n(x+mv_{_1}) = \widetilde{V}_n(x),\quad x\in Y_{n,-}\cup \Big(\cup_{l\le n-1}Y_l\Big),\; m\in\mathbb{Z},\\
			& \sup_{x\in Y_{n,-}\cup (\cup_{l\le n-1}Y_l)}|\widetilde{V}_{n}|<+\infty. 
		\end{aligned}\right.
	\end{gather}
	Moreover, the function $1-\wV_n$ satisfies 
	\[ 1-\wV_n \le c_{\ast}<1,\quad \forall x\in I_n,\; n\ge 1. \]
\end{theorem}
\begin{proof}
	From the truncation method, we have from \eqref{apeqn:periodpoisson}
	\[ \widetilde{h}_n(x) = \sum_{m\in \mathbb{Z}} \widehat{g}_n(m)\eu^{-2\pi |m|d(x_n,I_n) }\eu^{\iu 2\pi m x_1 }, \]
	where $\widehat{g}_n(m)$ is the $m$\textsuperscript{th} Fourier coefficient.
	Recall that the distance between $x_n$ and $I_n$ is denoted by $d(x_n,I_n)$. Thus, the $m$\textsuperscript{th}  Fourier coefficient of $\widetilde{h}_n$ is 
	\[ \widehat{h}_n(m) =  \widehat{g}_n(m)\eu^{-2\pi |m|d(x_n,I_n) }.\]
	Suppose that the Fourier coefficients of $1-\chi_{_{n}}$ are given by $\{\widehat{b}_{n}(m)\}_{m\in\mathbb{Z}}$, then we can directly compute 
	\begin{equation*}
		\int_{I_{n+1}} \widetilde{g}_{_{n+1}}\du\sigma(y) = \sum_{m\in\mathbb{Z}} \widehat{g}_n(m)\widehat{b}_n(-m)\eu^{-2\pi |m|d(x_n,I_n) }.
	\end{equation*}
	From the definition of the auxiliary problem, we have $\widetilde{V}_n>0$ on $I_{n}$, and 
	\[ 1-\widetilde{V}_n(x_1) = \sum_{m\in\mathbb{Z}} \widehat{b}_n(m)\eu^{-2\pi|m|d(x_n,I_n)}\eu^{\iu2\pi m x_1}. \]
	It follows that 
	\[ \widetilde{g}_n (1-\widetilde{V}_n) = \sum_{\ell\in \mathbb{Z}}\bigg( \sum_{m\in\mathbb{Z}} \widehat{g}_n(m)\widehat{b}_n(\ell-m)\eu^{-2\pi|\ell-m|d(x_n,I_n)}\bigg)\eu^{\iu 2\pi \ell x_1}. \] 
	Integrating over $I_n$ gives the desired result.
	
	By the maximum principle, $\widetilde{V}_n\in(0,1)$, it only remains to show that there exists a uniform upper bound for $1-\widetilde{V}_n$ with respect to $n\ge 1$. By \Cref{def:decaycond2},
	$d(x_n,I_n)\ge r+c$ and $|\omega_n|\ge2c$. And the periodic Poisson's kernel satisfies
	\begin{align*}
		P_{d(x_n,I_n)}^{\mathrm{per}}(x_1)&=\frac{\sinh[2\pi d(x_n,I_n) ]}
		{\cosh[2\pi d(x_n,I_n)]-\cos(2\pi x_1)}\\
		&\ge \frac{\sinh[2\pi d(x_n,I_n)]}{\cosh[2\pi d(x_n,I_n)]+1}
		=\frac{2\sinh[\pi d(x_n,I_n)]\cosh[\pi d(x_n,I_n)]}{2\cosh^2[\pi d(x_n,I_n)]}\\
		&\ge \tanh[\pi d(x_n,I_n)].
	\end{align*}
	Since $\chi_n\ge0$ and $\chi_n=1$ on $\omega_n$,
	the Poisson representation in \Cref{appprop:representation} gives
	\[
	\widetilde V_n(x)\ge
	2c\tanh\bigl(\pi(r+c)\bigr)>0,
	\qquad x\in I_n.
	\]
	Therefore, there exist $c_{\ast}$ such that $1-\widetilde V_n\le c_{\ast}<1$,
	uniformly in $n$.

\end{proof} 

Therefore, we have 
\begin{equation}
	\int_{I_{n+1}}\widetilde{g}_{_{n+1}}\du\sigma(y) \le \rho_n   \int_{I_{n}}\widetilde{g}_{_n}\du\sigma(y). 
\end{equation}
Here $\rho_n\coloneqq\|1-\widetilde V_n\|_{\mathbb L^\infty(I_n)}$ and $\rho\coloneqq\sup_{n\ge1}\rho_n<1$.
By \Cref{thm:DecayPer},
\[
\sup_{n\ge1}\rho_n\le c_\ast<1.
\]
From this and from \Cref{thm:Control_of_Trunc}, we are led to 
\[\int_{I_{n+1}}U_0^0\du \sigma(y)\le \int_{I_{n+1}}\widetilde{g}_{_{n+1}}\du \sigma(y) \le \rho^{n}\int_{I_1}U_0^0\du \sigma(y),\quad n\in \mathbb{Z}^+. \]
Since $\widetilde{g}_1$ is not larger than 1 by the maximum principle, we obtain the upper-direction estimate in \eqref{eqn:Decay_Condition1} for all $q\in \mathbb{Z}$.

For $n\le -1$, similar estimates can be obtained by considering the restriction of $U^0_0$ in $ \cup_{n\le -1} Y_n\setminus D_n $, as illustrated in \Cref{fig:TruncateMethodVWm}. 
\begin{figure}[htbp]
	\centering
	\includegraphics[width = 0.45\textwidth]{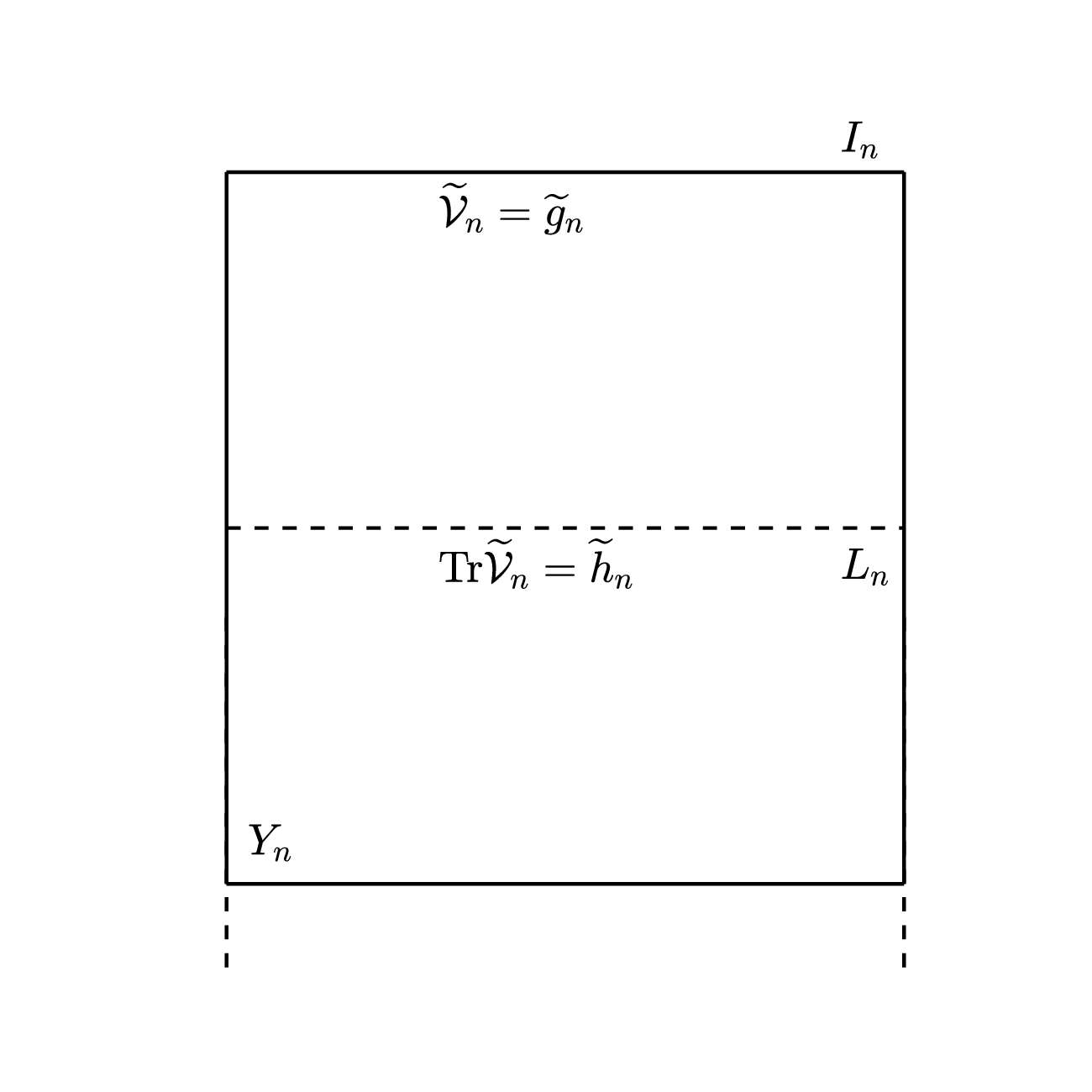}
	\hspace{-0.4cm}
	\includegraphics[width = 0.45\textwidth]{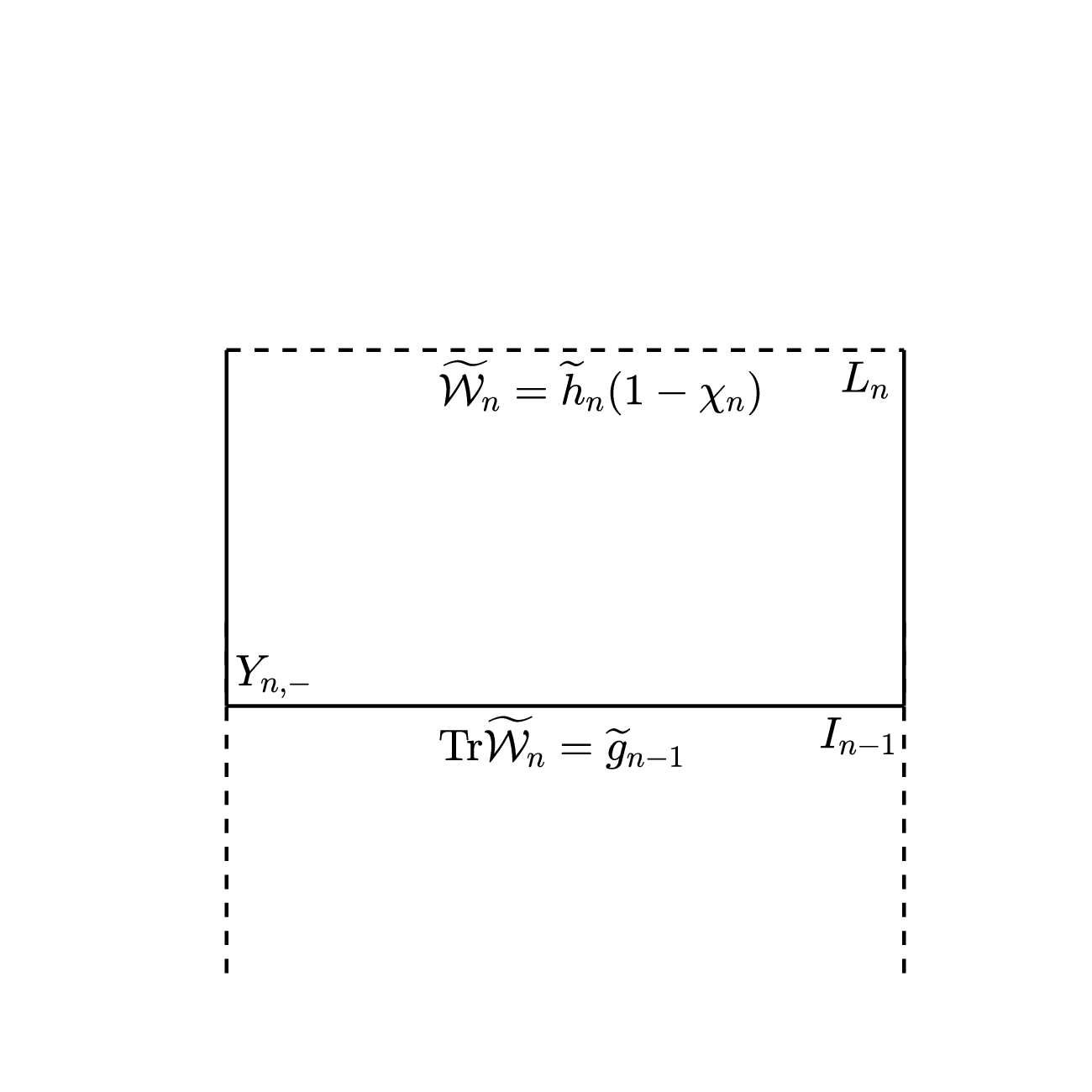}
	\caption{Left panel: Illustration of $\wcV_n$ and the derivation of $\widetilde{h}_n$ for $n<0$. Right panel: Illustration of $\wcW_n$ and the derivation of $\widetilde{g}_{_{n+1}}$ for $n<0$.}
	\label{fig:TruncateMethodVWm}
\end{figure} 
Thus, we can also define $\{ \wcV_n,\wcW_n \}_{n\le -1}$ and the related trace $\{ g_{_{n}} \}_{n\le -1}$ on $ \{ I_{n} \}_{n\le -1}$ and prove that the following estimate holds
\begin{equation}
	\int_{I_{-n}}U^0_0\du \sigma(y)\le \int_{I_{-n}}\widetilde{g}_{_{-n}}\du \sigma(y) \le \rho^{n}\int_{I_0}U_0^0\du \sigma(y),\quad n\in \mathbb{Z}^+.
\end{equation}

Combining the above arguments, we obtain the full estimate \eqref{eqn:Decay_Condition1}.

\subsection[Estimate on the strip boundary]{\texorpdfstring{Estimate on the boundary $\partial \wY$}{Estimate on the strip boundary}}
In this subsection, we estimate the trace of $U^0_0$ on the boundary $\partial \wY$, which will be used in \Cref{sec:estimate_full}. From the previous definitions, it can be expressed as 
\[\partial \wY = \wI_{2,0}\cup \wI_{2,1}.\]
We first present a version of the trace inequality which can be easily derived from the fundamental theorem of calculus.
\begin{proposition}
	The following estimate holds for all $ 1\le  p <\infty $ and $u\in \mathbb{W}^{1,p}(\wY)$:
	\[ \int_{\partial \wY} |u|^p(y)\du\sigma(y) \le C_p\bigg[\int_{\wY} |u|^p(x)\du x  + \int_{\wY} |\partial_{x_1}u|^p(x)\du x \bigg], \]
	where $C_p$ is a positive constant that depends on $p$ and the fixed lattice geometry.
\end{proposition}
Combining the smoothness of $U_0^0$ and the above results, we obtain the following result. 
\begin{lemma}\label{lem:estimatebdry}
	The trace of $U_q^0$ on $\wI_{2,0}$ and $\wI_{2,1} $ belongs to $\mathbb{L}^2(\wI_{2,0})\cap \mathbb{L}^1(\wI_{2,0})$ and $\mathbb{L}^2(\wI_{2,1})\cap \mathbb{L}^1(\wI_{2,1})$, respectively. Moreover, there exists a real constant $M$ such that the following estimates hold for all $q\in\mathbb{Z}$:
	\begin{gather}
		\begin{aligned}
			\Vert U_{q}^0 \Vert_{\mathbb{L}^2(\wI_{2,0})}  + \Vert U_{q}^0 \Vert_{\mathbb{L}^2(\wI_{2,1})} &\le M <+\infty,\\
			\Vert U_{q}^0 \Vert_{\mathbb{L}^1(\wI_{2,0})}  + \Vert U_{q}^0 \Vert_{\mathbb{L}^1(\wI_{2,1})} &\le M <+\infty.
		\end{aligned}
	\end{gather}
\end{lemma}
\begin{proof}
The separation assumption provides an inclusion-free collar of fixed width along every cell side. For $y\in I_{2,0,n}$ or $y\in I_{2,1,n}$, a Harnack chain in this collar connects $y$ to the bottom side $I_n$. The number of balls and the Harnack constant depend only on $r$ and the fixed cell geometry. Comparing with every point on $I_n$ and integrating gives
\[
\sup_{I_{2,0,n}\cup I_{2,1,n}}U_q^0
\le C\int_{I_n}U_q^0\du\sigma
\le C'\rho^{|n-q|},
\]
after increasing $C'$ to cover the two rows adjacent to the source. The last inequality follows from \eqref{eqn:Decay_Condition1}, already proved in this section. Each side segment has length $|v_{_2}|$. Summing this bound and its square over $n\in\mathbb Z$ proves both trace estimates, uniformly in $q$.
\end{proof}

\section[Proof of the longitudinal decay estimate]{\texorpdfstring{Proof of estimates \eqref{eqn:Decay_Condition2}}{Proof of the longitudinal decay estimate}}\label{sec:estimate_full}
In this section, we establish estimate \eqref{eqn:Decay_Condition2}. In the absence of periodicity in the $v_2$ direction, we can analogously define the truncation method on $\mathbb{R}$ together with its associated auxiliary problem. \par 
From Assumption \ref{def:decaycond2}, the distance $d(x_{m,n},I_{2,m,n})$ between $x_{m,n}$ and $I_{2,m,n}$ satisfies 
\[ d(x_{m,n},I_{2,m,n})\ge r+c>0. \]

Again, by \Cref{apthm:existenceHar}, we can show that the unique solution to the following problem is $U_{0,0}$. 
\begin{gather}\label{eqn:truncateProblemfull}
	\left\{\begin{aligned}
		&\Delta U_{+}(x) = 0,\quad x \in \cup_{m\ge 1}\wY_m\setminus \wD_m,\\
		&U_{+}(y)  = 0,\quad y \in \cup_{m\ge 1}\partial \wD_m,\\
		&U_{+}(y)  = U_{0,0}(y),\quad y \in \wI_{2,1}.
	\end{aligned}\right.
\end{gather}
Bounded comparison gives $0\le U_{0,0}\le U_0^0$. Therefore \Cref{lem:estimatebdry} gives 
\[ U_{0,0}\big|_{\wI_{2,1}} \in \mathbb{L}^1(\wI_{2,1})\cap \mathbb{L}^2(\wI_{2,1}).\]
We also let $U_+$ be $0$ inside the inclusions $\wD_m$ for all $m\ge 1$.

\subsection[Truncation method on the line]{\texorpdfstring{Truncation method on $\mathbb{R}$}{Truncation method on the line}}
In this subsection, we introduce the truncation method on $\mathbb{R}$ to estimate $U_+$ on the boundary $\wI_{2,m}$. We first note that each line $L_{2,m,n}$ divides the infinite strip $\wY_{m}$ into two parts, $\wY_{m,n,+}$ and $\wY_{m,n,-}$, as shown in \Cref{fig:TruncateCell}. 
\begin{figure}[htbp]
	\centering
	\includegraphics[width = 0.5\textwidth]{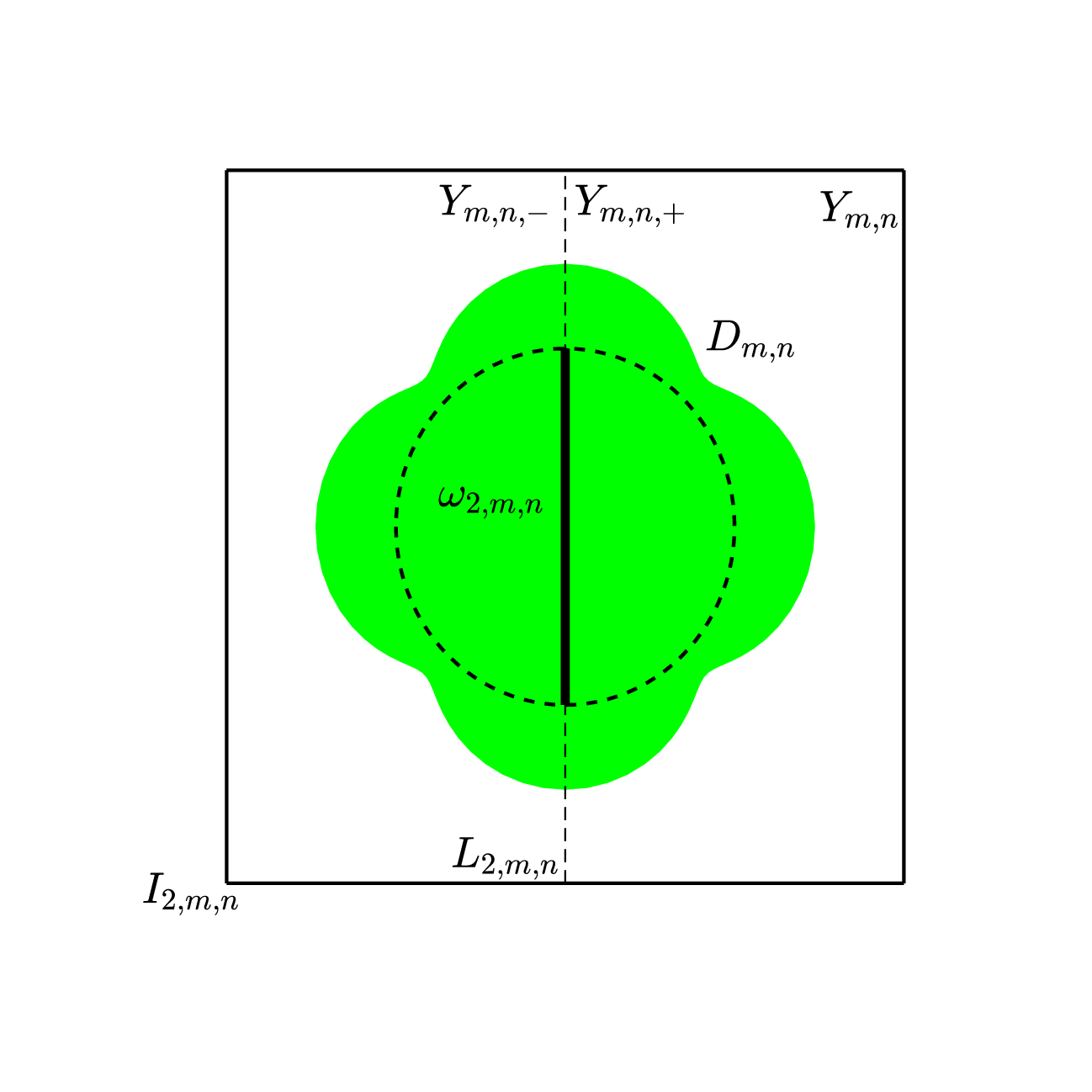}
	\caption{Cell $Y_{m,n}$. The line $L_{2,m,n}$ divides the cell into two parts, denoted as $Y_{m,n,+}$ and $Y_{m,n,-}$.}
	\label{fig:TruncateCell}
\end{figure} 
The procedure can also be split into three recursive steps. Given the boundary value $g_{_1} = U_{0,0}>0$ on $\wI_{2,1}$, for $m>0$,
\begin{itemize}
	\item[(1)] First, solve the following boundary value problem for the Laplace equation on the right half space for all $n\in\mathbb{Z}$:
	\begin{gather}
		\left\{\begin{aligned}
			&\Delta \cV_{m,n}(x) = 0,\quad x\in \cup_{l\ge m} \wY_l,\\
			& \cV_{m,n}(y) = g_{_m}(y),\quad y\in I_{2,m,n},\\
			& \cV_{m,n}(y) = 0,\quad y\in \wI_{2,m}\setminus I_{2,m,n},\\
			& \sup_{x\in \cup_{l\ge m} \wY_l} |\cV_{m,n}|<+\infty. 
		\end{aligned}\right.
	\end{gather}
	We can take its trace on the lines $\{L_{2,m,n}\}_{n\in\mathbb{Z}}$, denoted by $\{h_{m,n}\}_{n\in\mathbb{Z}}$.
	
	\item[(2)] Next, solve the boundary value problems for all $n\in \mathbb{Z}$:
	\begin{gather}
		\left\{\begin{aligned}
			&\Delta \cW_{m,n}(x) = 0,\quad x\in \wY_{m,n,+}\cup\Big(\cup_{l\ge m+1} \wY_l\Big),\\
			& \cW_{m,n}(y) = h_{m,n}(1-\chi_{_{m,n}})(y),\quad y\in L_{2,m,n},\\
			&\sup_{x\in \wY_{m,n,+}\cup(\cup_{l\ge m+1} \wY_l)} |\cW_{m,n}|<+\infty. 				
		\end{aligned}\right.
	\end{gather}
	Here, $\{\chi_{_{m,n}}\}_{n\in\mathbb{Z}}$ are smooth functions with $0\le\chi_{_{m,n}}\le1$ on $L_{2,m,n}$ satisfying 
	\[ \chi_{_{m,n}}(y) = 1,\quad y\in \omega_{2,m,n},\quad \text{and }\operatorname{supp}\chi_{_{m,n}}\subset L_{2,m,n}\cap D_{m,n}. \]
	Its trace on $\wI_{2,m+1}$ is denoted by $\{g_{_{m+1,n}}\}_{n\in\mathbb{Z}}$. Taking the summation with respect to $n$ gives $g_{_{m+1}}$
	\begin{equation}\label{eqn:reconssum}
		g_{_{m+1}} \coloneqq \sum_{n\in\mathbb{Z}} g_{_{m+1,n}} \in \mathbb{L}^1(\wI_{2,m+1})\cap \mathbb{L}^2(\wI_{2,m+1}). 
	\end{equation}
	\item[(3)] Repeat Step (1) and Step (2).
\end{itemize} 
\begin{remark}\label{rem:SeriesTruncation}
    Here we have to clarify the convergence of the series in \eqref{eqn:reconssum}. Let $d=\operatorname{dist}(\wI_{2,m},\wI_{2,m+1})>0$. The solution can be represented by convolution with Poisson kernel:
    \begin{align*}
      h_{m,n}&=P_{\operatorname{dist}(x_{m,n},\wI_{2,m})}*( \chi_{{I_{2,m,n}}} g_{m}),\\
      0\le g_{m+1,n}
      &=P_{\operatorname{dist}(x_{m,n},\wI_{2,m+1})}*\big[(1-\chi_{m,n})h_{m,n}\big]
      \le P_d*(\chi_{{I_{2,m,n}}} g_{m}).
    \end{align*}
    And Young's inequality yields
    \begin{align*}
      \sum_n\|g_{m+1,n}\|_{\mathbb{L}^1}
      &\le\sum_n\|\chi_{{I_{2,m,n}}} g_{m}\|_{\mathbb{L}^1}=\|g_m\|_{\mathbb{L}^1},\\
      \sum_n\|g_{m+1,n}\|_{\mathbb{L}^2}
      &\le\|P_d\|_{\mathbb{L}^2}\sum_n\|\chi_{{I_{2,m,n}}} g_{m}\|_{\mathbb{L}^1}
      =\frac{\|g_m\|_{\mathbb{L}^1}}{\sqrt{2\pi d}}.
    \end{align*}
    Thus \eqref{eqn:reconssum} is an absolutely convergent series in both $\mathbb{L}^1$ and $\mathbb{L}^2$. The maximum principle gives $0\le g_1\le1$, and $g_{m+1}\le P_d*g_m$ preserves this bound.
\end{remark}

\begin{remark}
	In step (1), we split the boundary data into components supported on	$I_{2,m,n}$ so that each component can be truncated on its
	corresponding line $L_{2,m,n}$ using $\chi_{m,n}$. 
	These correction lines need not coincide, since the centers of the interior balls $x_{m,n}$ may vary with $n$. By the Poisson representation, summing	$\mathcal V_{m,n}$ over $n$ recovers the harmonic extension of $g_m$ prescribed on the whole line $\wI_{2,m}$.
\end{remark}

We can also prove the following theorem for the truncation method on $\mathbb{R}$:
\begin{theorem}
	Given the solutions $\{\{ \cV_{m,n}\}_{n\in\mathbb{Z}},\{ \cW_{m,n} \}_{n\in\mathbb{Z}} \}_{m\ge 1} $ obtained by the truncation method on $\mathbb{R}$, it holds that for all $m\ge1$,
	\[ U_{0,0}\le g_{_{m}},\quad y\in \wI_{2,m}. \]
\end{theorem}
\begin{proof}
	The proof proceeds by induction. First, $U_{+}(y) = g_{_1}$ on $\wI_{2,1}$. Assuming $U_{+}\le g_{_m}$ on $\wI_{2,m}$, consider its restriction to the perforated half-space
    \[
      \Omega_m\coloneqq\bigcup_{l\ge m}\Big( \wY_l\setminus\overline{\wD_l} \Big).
    \]
    For each $n$, let $U_{+,m,n}$ be the bounded solution to the following boundary value problems 
	\begin{gather*}
		\begin{aligned}
			\left\{\begin{aligned}
				&\Delta U_{+,m,n}(x) = 0,\quad x \in \Omega_m,\\
				&U_{+,m,n}(y)  = 0,\quad y \in \cup_{l\ge m}\partial \wD_l,\\
				&U_{+,m,n}(y)  = U_{0,0}(y),\quad y \in I_{2,m,n},\\
				&U_{+,m,n}(y)  = 0,\quad y \in \wI_{2,m}\setminus I_{2,m,n}.
			\end{aligned}\right.
		\end{aligned}
	\end{gather*} 
	The boundary data are understood almost everywhere on the line. The bounded solutions may be obtained by exhaustion and comparison with the half-plane Poisson extension. Put $f_m=U_+|_{\wI_{2,m}}$, $E_N=\bigcup_{|n|>N}I_{2,m,n}$ and $S_N=\sum_{|n|\le N}U_{+,m,n}$. Since $0\le f_m\le g_m\in\mathbb L^1$, bounded comparison gives
	\[
	0\le U_+(x)-S_N(x)
	\le \big(P_t*(f_m\chi_{E_N})\big)(s)
	\le\frac{\|f_m\chi_{E_N}\|_{\mathbb L^1}}{\pi t}\longrightarrow0,
	\]
	where $t=\operatorname{dist}(x,\wI_{2,m})>0$ and $s$ is the common tangential coordinate. Thus $U_+=\sum_nU_{+,m,n}$ locally uniformly in $\Omega_m$ and uniformly on $\wI_{2,m+1}$.
	From this, consider the harmonic function $\cV_{m,n}-U_{+,m,n}$. It can be directly verified that for all $n\in \mathbb{Z}$, $\cV_{m,n}-U_{+,m,n}$ satisfies
	\begin{gather*}
		\left\{\begin{aligned}
			&\Delta \big(\cV_{m,n}-U_{+,m,n}(x)\big) = 0,\quad x \in \cup_{l\ge m}\wY_l\setminus \overline{ \wD}_l,\\
			&\cV_{m,n}(y)-U_{+,m,n}(y)  > 0,\quad y \in \cup_{l\ge m}\partial \wD_l,\\
			&\cV_{m,n}(y)-U_{+,m,n}(y)  \ge 0 ,\quad y \in I_{2,m,n},\\
			&\cV_{m,n}(y)-U_{+,m,n}(y)  = 0,\quad y \in \wI_{2,m}\setminus I_{2,m,n}.
		\end{aligned}\right.
	\end{gather*} 
	Thus, it follows that $\cV_{m,n}(y)-U_{+,m,n}$ is nonnegative in $\Omega_m$, which implies $h_{m,n}\ge U_{+,m,n}$ on $L_{2,m,n}$. Since $ \{U_{+,m,n}\}_{n\in\mathbb{Z}}$ is equal to zero in $\wD_{m}$, it follows that 
	\[ h_{m,n}(1-\chi_{_{m,n}})(y)\ge U_{+,m,n},\quad y\in L_{2,m,n}. \]
	Next, consider the harmonic function $\cW_{m,n}-U_{+,m,n}$. Again, by definition, one has 
	\begin{gather*}
		\left\{\begin{aligned}
			&\Delta \big(\cW_{m,n}-U_{+,m,n}(x)\big) = 0,\quad x \in (\wY_{m,n,+}\setminus \overline{ \wD}_m)\bigcup(\bigcup_{l\ge m+1}\wY_l\setminus \overline{\wD}_l),\\
			&\cW_{m,n}(y)-U_{+,m,n}(y)\ge0,\quad y\in\partial\wD_m\cap\wY_{m,n,+},\\
			&\cW_{m,n}(y)-U_{+,m,n}(y)\ge0,\quad y\in\cup_{l\ge m+1}\partial\wD_l,\\
			&\cW_{m,n}(y)-U_{+,m,n}(y)  \ge 0 ,\quad y \in L_{2,m,n}.
		\end{aligned}\right.
	\end{gather*} 
	Therefore, we conclude that $\cW_{m,n}\ge U_{+,m,n}$ in the relevant domain. Taking the trace of $\cW_{m,n}$ on $\wI_{2,m+1}$, it can be verified that $\cW_{m,n} \ge U_{+,m,n}$ on $\wI_{2,m+1}$. 
	
	By \Cref{rem:SeriesTruncation}, the traces $g_{m+1,n}$ sum in $\mathbb L^1$ and $\mathbb L^2$ on $\wI_{2,m+1}$. Summing the comparison there, and using the locally uniform decomposition just established, gives $U_+\le\sum_n g_{m+1,n}=g_{m+1}$. No integrability of the Poisson extensions over the entire half-space is needed.
\end{proof}

\subsection[Exponential decay of the longitudinal integral]{\texorpdfstring{Exponential decay of the integral on $\wI_{2,m}$}{Exponential decay of the longitudinal integral}}
In this subsection, we estimate the decay of $\{g_{_{m}}\}_{m\ge 1}$ with respect to $m$, where the functions $g_{_m}$ are determined by the truncation method on $\mathbb{R}$. We prove the following theorem.	
\begin{theorem}\label{thm:DecayFull}
	Given the functions $\{ g_{_m} \}_{m\ge 1}$ determined by the truncation method on $\mathbb{R}$, we have 
	\begin{equation}
		\int_{\wI_{2,m+1}}g_{_{m+1}}\du\sigma(y) \le \rho\int_{\wI_{2,m}}g_{_m}\du\sigma(y).
	\end{equation}
	Here $0<\rho<1$ is independent of $m$.  
\end{theorem}
\begin{proof}
We use tangential arclength and normal coordinates as specified in \Cref{sec:setting}. Let $d_{m,n}=d(x_{m,n},\wI_{2,m})$ and $f_{m,n}=g_m\chi_{I_{2,m,n}}$. Define $V_{m,n}$ as the bounded harmonic function on
\[
\wY_{m,n,-}\cup\Big(\bigcup_{l\le m-1}\wY_l\Big)
\]
with boundary value $\chi_{m,n}$ on $L_{2,m,n}$. Thus its trace on $\wI_{2,m}$ is $P_{d_{m,n}}*\chi_{m,n}$, while $h_{m,n}=P_{d_{m,n}}*f_{m,n}$. Conservation of mass under Poisson propagation and symmetry of the kernel give
\begin{align*}
\int_{\wI_{2,m+1}}g_{m+1,n}\du\sigma
&=\int_{L_{2,m,n}}(1-\chi_{m,n})h_{m,n}\du\sigma\\
&=\int_{I_{2,m,n}}g_m(1-V_{m,n})\du\sigma.
\end{align*}
All integrands are nonnegative, so these interchanges follow from Tonelli's theorem.

Let $R$ be the diameter of a cell. The geometric assumptions give $r+c\le d_{m,n}\le R$. If $y\in I_{2,m,n}$ and $z\in\omega_{2,m,n}$, their tangential separation is at most $R$. Since $\chi_{m,n}=1$ on a segment of length at least $2c$, the Poisson kernel gives
\[
V_{m,n}(y)\ge\frac{2c(r+c)}{2\pi R^2}\eqqcolon\theta>0.
\]
Also $V_{m,n}<1$, so $\rho\coloneqq1-\theta\in(0,1)$. The bound is uniform in $m,n$. Summing the mass identity and using \Cref{rem:SeriesTruncation} proves
\[
\int_{\wI_{2,m+1}}g_{m+1}\du\sigma
\le\rho\int_{\wI_{2,m}}g_m\du\sigma.
\]
\end{proof}

For $m\le -1$, similar estimates follow by considering the restriction of $U_{0,0}$ to $\bigcup_{m\le -1}(\wY_m\setminus\wD_m)$. Thus, the functions $$\{\{ \cV_{m,n}\}_{n\in\mathbb{Z}},\{ \cW_{m,n} \}_{n\in\mathbb{Z}} \}_{m\le -1} $$ and the related trace on $ \wI_{2,m}$ can also be proved to satisfy
\[
\int_{\wI_{2,-m}}U_{0,0}\du \sigma(y) \le \int_{\wI_{2,-m}}g_{_{-m}}\du\sigma(y) \le \rho^{m} \int_{\wI_{2,0}}U_{0,0}\du\sigma(y),\quad m\in\mathbb Z^+. \]
Combining the above arguments, we obtain the full estimate \eqref{eqn:Decay_Condition2}.

\begin{remark}
	\Cref{sec:Decay_periodic} and \Cref{sec:estimate_full} have established the exponential decay of the quasi-periodic capacitance matrix $\widehat{\mathbf{C}}^{\alpha}$, which implies that the capacitance coefficients decay exponentially in terms of the distance between the corresponding resonators. The error of a banded approximation is therefore controlled by \Cref{cor:banded_operator_error}; the accuracy of the nearest-neighbor choice depends on the decay rate and the required tolerance. The uniform presence of separated inclusions in every cell is essential to this argument. Finite and lower-dimensional arrays need not have the same decay; see \cite{ammari2023convergence}.
\end{remark}

\begin{remark}
	The methodology suggests an extension to higher-dimensional systems. More specifically, consider resonators $D_{\fu}$ in $\mathbb{R}^3$ that extend to infinity, with basis vectors $v_{_1},v_{_2},v_{_3}$ spanning the lattice points \[\Lambda = \{ mv_{_1}+nv_{_2}+pv_{_3}:m,n,p\in\mathbb{Z} \}. \] If the resonators are periodic with respect to $v_{_1} $ and $v_{_2}$, a corresponding decay result would require uniform geometric and Poincar\'e estimates and the appropriate higher-dimensional comparison argument.  
\end{remark}

\section{Subwavelength problem formulation}\label{sec:Subwavelength_capa}
Capacitance reductions for infinite resonator crystals are developed in \cite{ammari2025analysisnonlinearresonancesresonator}. In this section, we adapt the linear reduction to the waveguide fiber for fixed $\alpha\neq0$. For subwavelength frequencies $\omega=\mathcal O(\sqrt\delta)$ corresponding to modes square-integrable in the fiber, we prove a spectral-distance estimate and an eigenvalue approximation under an isolation condition. 
Our analysis relies on layer potential theory. The solution to \eqref{eqn:ProbFormq} can be represented by the single-layer potential: 
\begin{gather}
	w(x) = \left\{\begin{aligned}
		&\mathcal{S}^{\alpha,k_1}_{\wD}[\phi](x),\quad x\in \wD,\\
		&\mathcal{S}^{\alpha,k_0}_{\wD}[\psi](x),\quad x\in \wY \setminus \overline{\wD},
	\end{aligned}\right.\quad \phi,\psi \in \mathbb{H}^{-\frac12}(\partial \wD).
\end{gather}

To present the main theorem, we first introduce some lemmas. 
\begin{lemma}\label{lem:kernel}
For fixed $\alpha\ne0$, the kernel of $-\frac12\operatorname{Id}+(\mathcal K^{-\alpha,0}_{\wD})^\ast$ on $\mathbb H^{-\frac12}(\partial\wD)$ is
\[
\left\{\sum_{n\in\mathbb Z}c_n\psi_n^\alpha:\{c_n\}_{n\in\mathbb Z}\in l^2(\mathbb Z)\right\},
\]
where $\psi_n^\alpha=\nu\cdot\nabla U_n^\alpha|_+\in\mathbb H^{-\frac12}(\partial\wD)$. The series converges in $\mathbb H^{-\frac12}(\partial\wD)$ and defines a bounded isomorphism from $l^2(\mathbb Z)$ onto the kernel. Here $U_n^\alpha$ is continued as the constant $\delta_{nj}$ in $D_j$.
\end{lemma}
\begin{proof}
The uniform trace bounds imply that $\{c_n\}\mapsto\sum_nc_n\chi_{\partial D_n}$ is bounded from $l^2(\mathbb Z)$ to $\mathbb H^{\frac12}(\partial\wD)$. Apply the bounded inverse of $\mathcal S^{\alpha,0}_{\wD}$ from \Cref{apthm:inverti}. Its potential is constant in each inclusion, so the interior jump relation places the resulting density in the kernel and identifies it with $\sum_nc_n\psi_n^\alpha$.

Conversely, if $\psi$ belongs to the kernel, the interior normal derivative of $\mathcal S^{\alpha,0}_{\wD}[\psi]$ vanishes. Since each $D_n$ is connected, Green's identity makes this potential constant, say $c_n$, in $D_n$. By \Cref{aplem:boundedH1} and $|D_n|\ge\pi c^2$,
\[
\sum_n|c_n|^2\le\frac1{\pi c^2}\|\mathcal S^{\alpha,0}_{\wD}[\psi]\|_{\mathbb L^2(\wD)}^2
\le C_\alpha\|\psi\|_{\mathbb H^{-\frac12}(\partial\wD)}^2.
\]
Invertibility of the single-layer operator gives the asserted representation. These two estimates also prove boundedness of the synthesis map and its inverse.
\end{proof}

The following operator asymptotic expansions hold on the infinite boundary $\partial \wD$.
\begin{lemma}\label{exps}
	For $\alpha \neq 0$, there exist bounded operators $\mathcal S_2^\alpha:\mathbb{H}^{-\frac{1}{2}}(\partial \wD)\to \mathbb{H}^{\frac{1}{2}}(\partial \wD)$ and $\mathcal K_2^\alpha:\mathbb{H}^{-\frac{1}{2}}(\partial \wD)\to \mathbb{H}^{-\frac{1}{2}}(\partial \wD)$ such that the following expansions hold for sufficiently small $k$:
	\begin{align*}
		\mathcal S_{\wD}^{\alpha,k}
		&=\mathcal S_{\wD}^{\alpha,0}+k^2\mathcal S_2^\alpha+\mathcal R_1^\alpha(k),\\
		(\mathcal K_{\wD}^{-\alpha,k})^\ast
		&=(\mathcal K_{\wD}^{-\alpha,0})^\ast+k^2\mathcal K_2^\alpha+\mathcal R_2^\alpha(k),
	\end{align*}
	where
	\[
	\|\mathcal R_1^\alpha(k)\|_{\mathcal L(\mathbb{H}^{-\frac{1}{2}}(\partial \wD),\mathbb{H}^{\frac{1}{2}}(\partial \wD))}
	+\|\mathcal R_2^\alpha(k)\|_{\mathcal L(\mathbb{H}^{-\frac{1}{2}}(\partial \wD))}
	\le C_\alpha k^4.
	\]
	Moreover, for $\phi\in \mathbb{H}^{-\frac{1}{2}}(\partial \wD)$ and $n\in\mathbb{Z}$,
	\begin{equation}\label{eqn:second_order}
		\int_{\partial D_n}\mathcal{K}_2^\alpha[\phi]\du\sigma
		=-\int_{D_n}\mathcal S_{\wD}^{\alpha,0}[\phi]\du x.
	\end{equation}
	
\end{lemma}
\begin{proof}
	For $\phi\in\mathbb H^{-\frac12}(\partial\wD)$, consider the single-layer potential on $\wY\setminus\partial\wD$ and write
	\[
	\mathcal{S}^{\alpha,k}_{\wD}[\phi](x) = \mathcal{S}^{\alpha,0}_{\wD}[\phi](x) + k^2 u_2^{\alpha}(x)+ u_{R}^{\alpha}(x),\quad x\in\wY\setminus\partial \wD. 
	\]
	Here $u_2^\alpha$ and $u_{R}^{\alpha}$ satisfy
	\begin{gather}\label{eqn:residue_eqn}
		\left\{\begin{aligned}
			&\Delta u_{2}^{\alpha} =- \mathcal{S}^{\alpha,0}_{\wD}[\phi],\\
			&u_2^\alpha|_{+} = u_2^\alpha|_{-},\\
			&\nu\cdot \nabla u^\alpha_2|_{+} = \nu\cdot \nabla u^\alpha_2|_{-},\\
			& u^{\alpha}_2(x+mv_{_1}) = \eu^{\iu m \alpha}u^\alpha_2(x),
		\end{aligned}\right.
		\quad
		\left\{\begin{aligned}
			&\Delta u_{R}^{\alpha}+ k^2u_R^\alpha = -k^4 u_2^\alpha,\\
			&u_R^\alpha|_{+} = u_R^\alpha|_{-},\\
			&\nu\cdot \nabla u^\alpha_R|_{+} = \nu\cdot \nabla u^\alpha_R|_{-},\\
			& u^{\alpha}_R(x+mv_{_1}) = \eu^{\iu m \alpha}u^\alpha_R(x).
		\end{aligned}\right.
	\end{gather}
	For $u_2^\alpha$, we first seek the quasi-periodic solution $u_{2,A}^{\alpha}\in \mathbb{U}^{\alpha}$ that satisfies
	\[
	\int_{\wY}\overline{\nabla u_{2,A}^{\alpha}}\cdot \nabla v^{\alpha} \du x = \int_{\wY} \overline{\mathcal{S}^{\alpha,0}_{\wD}[\phi]} v^{\alpha}\du x,\quad \forall v^{\alpha}\in\mathbb{U}^{\alpha}.
	\]
	Here the space for test function $\mathbb{U}^{\alpha}$ is defined by
	\[
	\mathbb{U}^{\alpha}:= \{ v^{\alpha}\in\mathbb H^1_\alpha(\wY):\operatorname{Tr}_{\partial\wD}v^{\alpha} = 0 \}.
	\] 
	From \Cref{apthm:existenceHar}, there is a unique weak solution to the above equation satisfying
	\begin{equation}\label{eqn:asym_estimate_1}
	\| u^{\alpha}_{2,A} \|_{\mathbb{H}^1} \le C\|\mathcal{S}^{\alpha,0}_{\wD}[\phi]\| _{\mathbb{H}^1}\le C\| \phi \|_{\mathbb{H}^{-\frac{1}{2}}}.
	\end{equation}
	However, the normal derivative is not continuous across the boundary. For any non-zero $f\in\mathbb{H}^{\frac{1}{2}}(\partial \wD) $, let $\phi_f := (\mathcal{S}^{\alpha,0}_{\wD})^{-1}[f]$ and consider the following pairing
	\[
	\langle \psi,f\rangle = \int_{\wY}\overline{\mathcal{S}^{\alpha,0}_{\wD}[\phi]} \mathcal{S}^{\alpha,0}_{\wD}[\phi_f] \du x,\quad f\in \mathbb{H}^{\frac{1}{2}}(\partial \wD).
	\]
	By \Cref{aplem:boundedH1}, it can be directly verified that 
	\begin{equation}\label{eqn:asym_estimate_2}
		\left|\int_{\wY}\overline{\mathcal{S}^{\alpha,0}_{\wD}[\phi]} \mathcal{S}^{\alpha,0}_{\wD}[\phi_f] \du x\right|\le C\| \phi \|_{\mathbb{H}^{-\frac{1}{2}}}
		\|f \|_{\mathbb{H}^{\frac{1}{2}}},
	\end{equation}
	which determines $\psi\in\mathbb{H}^{-\frac{1}{2}} (\partial \wD)$ uniquely. 
	
	Now we consider the function $ u_{2,A}^{\alpha}-\mathcal{S}^{\alpha,0}_{\wD}[\psi] $, it can be directly verified that for all $f\in \mathbb{H}^{\frac{1}{2}}(\partial \wD)$, we have 
	\[
	\begin{aligned}
		 -\int_{\partial\wD}(\overline{\nu\cdot\nabla u^{\alpha}_{2,A}}|_{+}-\overline{\nu\cdot\nabla u^{\alpha}_{2,A}}|_- )f\du \sigma&=\int_{\wY}\overline{\Delta u^{\alpha}_{2,A}}\mathcal{S}^{\alpha,0}_{\wD}[\phi_f]-\overline{ u^{\alpha}_{2,A}}\Delta\mathcal{S}^{\alpha,0}_{\wD}[\phi_f]\du x\\
		 &=-\int_{\wY}\overline{\mathcal{S}^{\alpha,0}_{\wD}[\phi]}\mathcal{S}^{\alpha,0}_{\wD}[\phi_f] \du x = -\langle \psi,f\rangle.
	\end{aligned}
	\]
	Therefore, 
	\[ \nu\cdot\nabla u^{\alpha}_{2,A}|_{+}-\nu\cdot\nabla u^{\alpha}_{2,A}|_- = \psi. \]
	Since 
	\[ \nu\cdot\nabla \mathcal{S}^{\alpha,0}_{\wD}[\psi]|_{+}-\nu\cdot\nabla \mathcal{S}^{\alpha,0}_{\wD}[\psi]|_- = \psi, \]
	the normal derivative is continuous across the boundary. Moreover, from \eqref{eqn:asym_estimate_1} and \eqref{eqn:asym_estimate_2}, one has 
	\[ \|u^{\alpha}_{2}\|_{\mathbb{H}^{1}}\le C\| \phi \|_{\mathbb{H}^{-\frac{1}{2}}} .
	\]
	
	For $u_R^\alpha$, solve $(\Delta+k^2)u_R^\alpha=-k^4u_2^\alpha$ on the whole fiber, with continuous value and normal derivative across the interfaces. Fourier expansion gives $\|v\|_{\mathbb L^2(\wY)}\le|\alpha|^{-1}\|\partial_{x_1}v\|_{\mathbb L^2(\wY)}$. Thus the form $\int_{\wY}(|\nabla v|^2-k^2|v|^2)$ is uniformly coercive for $0\le k\le k_*<|\alpha|$, and
	\[
	\|u_R^\alpha\|_{\mathbb H^1(\wY)}\le C_{\alpha,k_*}k^4\|\phi\|_{\mathbb H^{-\frac12}(\partial\wD)}.
	\]
	The sum $\mathcal S^{\alpha,0}_{\wD}[\phi]+k^2u_2^\alpha+u_R^\alpha$ has the same equation and jumps as $\mathcal S^{\alpha,k}_{\wD}[\phi]$. Their difference solves the homogeneous whole-fiber equation and vanishes by the same coercivity. Taking the trace proves the first expansion.
	
	To justify the asymptotic expansion of the Neumann-Poincar\'e operators, consider the interior normal derivative of $\mathcal{S}^{\alpha,k}_{\wD}[\phi]$. We have 
	\[
	-\frac{1}{2}\phi + (\mathcal{K}^{-\alpha,k}_{\wD})^\ast[\phi] = -\frac{1}{2}\phi + (\mathcal{K}^{-\alpha,0}_{\wD})^\ast[\phi] + k^2\nu\cdot \nabla u_{2}^{\alpha}|_{-} + \nu\cdot \nabla u_{R}^{\alpha}|_{-}. 
	\]
	For $f\in \mathbb H^{\frac12}(\partial\wD)$, we can construct $F\in\mathbb H^1(\wD)$ with $\operatorname{Tr}_{\wD}F=f$ and $\|F\|_{\mathbb H^1(\wD)}\le C\|f\|_{\mathbb H^{\frac12}}$. For $u^{\alpha}_2$ we have 
	\begin{align*}
		|\langle \nu \cdot \nabla u^{\alpha}_2|_-,f\rangle|
		&=\Big|\int_{\wD}
		\bigl( \overline{\nabla u^{\alpha}_2}\cdot\nabla F+ \overline{\Delta u_{2}^{\alpha}}F\bigr)\du x\Big|\\
		&\le C\bigl(\|u^{\alpha}_2\|_{\mathbb H^1(\wD)}
		+\| \mathcal{S}^{\alpha,0}_{\wD}[\phi]\|_{\mathbb L^2(\wD)}\bigr)\|f\|_{\mathbb H^{\frac12}(\partial\wD)}.
	\end{align*}
	Thus
	\[
	\| \nu \cdot \nabla u^{\alpha}_2|_-\|_{\mathbb H^{-\frac12}(\partial\wD)}
	\le C\|\phi\|_{\mathbb H^{-\frac12}}.
	\]
	And similarly, we have 
	\[
	\| \nu \cdot \nabla u^{\alpha}_R|_-\|_{\mathbb H^{-\frac12}(\partial\wD)}
	\le Ck^4\|\phi\|_{\mathbb H^{-\frac12}}.
	\]
	Therefore $\mathcal{K}^{\alpha}_2$ and $\mathcal R_2^\alpha(k)$ are bounded operators on $\mathbb H^{-\frac12}(\partial \wD)$. 
	
	Finally, for every $n\in\mathbb Z$, Green's identity on the bounded inclusion $D_n$ gives
	\[
	\int_{\partial D_n}\mathcal K_2^\alpha[\phi]\du\sigma
	=\int_{D_n}\Delta u_2^\alpha\du x
	=-\int_{D_n}\mathcal S^{\alpha,0}_{\wD}[\phi]\du x.
	\]
	we have completed the proof.
\end{proof}

The following argument adapts the linear case of the capacitance reduction in \cite[Theorem 2.6 and Section 4]{ammari2025analysisnonlinearresonancesresonator} to a fixed Bloch fiber; see in particular Propositions 4.3--4.4 therein. In the present line-defect geometry, the uniform extension, trace and Poincar\'e estimates in \Cref{apsec:extensiontrace,apsec:poincare} provide the bounds needed uniformly in the transverse index. The discrete space is $l^2(\mathbb Z)$, and the inclusion volumes enter through the mass matrix $\mathbf M$. We give a layer-potential proof, keeping the remainder estimates in $l^2(\mathbb Z)$.

\begin{theorem}\label{thm:capamat_approx}
Fix $\alpha\in[-\pi,\pi)\setminus\{0\}$ and $K>0$, and set
\[
\mathbf M=\operatorname{diag}\{|D_n|\}_{n\in\mathbb Z},\qquad
A^\alpha=\mathbf M^{-1/2}\widehat{\mathbf C}^\alpha\mathbf M^{-1/2}.
\]
There exist $C_{\alpha,K}>0$ and $\delta_0>0$ such that, for $0<\delta<\delta_0$, every non-zero guided mode solving \eqref{eqn:ProbFormq} with frequency $\omega>0$ and
\[
z_\delta\coloneqq\frac{\omega^2}{\delta\mu_{_1}^2}\le K
\]
satisfies
\[
\operatorname{dist}\bigl(z_\delta,\sigma(A^\alpha)\bigr)\le C_{\alpha,K}\delta.
\]
Suppose, moreover, that $\lambda^\alpha>0$ is an isolated eigenvalue of $A^\alpha$ of finite multiplicity. Choose $g_\alpha>0$ such that
\[
\sigma(A^\alpha)\cap[\lambda^\alpha-2g_\alpha,\lambda^\alpha+2g_\alpha]
=\{\lambda^\alpha\}.
\]
For the frequencies above with $|z_\delta-\lambda^\alpha|\le g_\alpha$, we have
\begin{equation}
\omega(\alpha)=\mu_{_1}\sqrt{\delta\lambda^\alpha}+\mathcal O(\delta).
\end{equation}
Here $\lambda^\alpha$ is a generalized eigenvalue of the quasi-periodic capacitance matrix:
\begin{equation}\label{eqn:genEigProb}
\widehat{\mathbf C}^\alpha\mathbf v^\alpha
=\lambda^\alpha\mathbf M\mathbf v^\alpha,
\qquad 0\ne\mathbf v^\alpha\in l^2(\mathbb Z).
\end{equation}
\end{theorem}
\begin{proof}
We seek solutions $(\phi^\alpha,\psi^\alpha)$ with
$\|\phi^\alpha\|_{\mathbb H^{-\frac12}(\partial\wD)}=1$.
By the operator expansions in \Cref{exps},
\begin{equation}\label{eqn:Singlematch}
\mathcal S^{\alpha,0}_{\wD}[\phi^\alpha]
-\mathcal S^{\alpha,0}_{\wD}[\psi^\alpha]=\mathcal O(\omega^2),
\end{equation}
and
\begin{equation}\label{eqn:NPmatch}
\begin{split}
\left\{-\frac12\operatorname{Id}+(\mathcal K^{-\alpha,0}_{\wD})^\ast
+\frac{\omega^2}{\mu_{_1}^2}\mathcal K_2^\alpha\right\}[\phi^\alpha]
&-\delta\left\{\frac12\operatorname{Id}+(\mathcal K^{-\alpha,0}_{\wD})^\ast\right\}[\psi^\alpha]\\
&=\mathcal O(\omega^4+\delta\omega^2).
\end{split}
\end{equation}
The remainders are in $\mathbb H^{\frac12}(\partial\wD)$ and $\mathbb H^{-\frac12}(\partial\wD)$, respectively. Invertibility of the single-layer operator gives
\[
\psi^\alpha=\phi^\alpha+\mathcal O(\omega^2)
\quad\text{in }\mathbb H^{-\frac12}(\partial\wD).
\]
Here boundedness of $\psi^\alpha$ follows first from
$\psi^\alpha=(\mathcal S^{\alpha,k_0}_{\wD})^{-1}\mathcal S^{\alpha,k_1}_{\wD}[\phi^\alpha]$ and \Cref{exps}.

Using \Cref{lem:kernel} and the orthogonal decomposition for $(\cdot,\cdot)_\ast$ in \Cref{apsec:Invert_NP}, write
\begin{equation}
\phi^\alpha=\sum_{n\in\mathbb Z}c_n^\alpha\psi_n^\alpha+\varphi^\alpha,
\qquad \mathbf c^\alpha\in l^2(\mathbb Z),\quad\varphi^\alpha\in\mathcal X.
\end{equation}
Equation \eqref{eqn:NPmatch} gives
\[
\left\{-\frac12\operatorname{Id}+(\mathcal K^{-\alpha,0}_{\wD})^\ast\right\}[\varphi^\alpha]
=\mathcal O(\omega^2+\delta).
\]
The restriction of this operator to $\mathcal X$ is invertible by \Cref{apthm:invertiK}. Hence
\[
\|\varphi^\alpha\|_{\mathbb H^{-\frac12}(\partial\wD)}
\le C_\alpha(\omega^2+\delta).
\]
Since the synthesis map in \Cref{lem:kernel} is bounded and $\|\phi^\alpha\|=1$, it follows that $\|\mathbf c^\alpha\|_{l^2}\ge c_\alpha>0$ for sufficiently small $\delta$.

To control the summed remainders, the uniform trace bounds and the uniform upper bound on $|D_n|$ give
\begin{align*}
\left\|\left(\int_{\partial D_n}\eta\du\sigma\right)_n\right\|_{l^2}
&\le C\|\eta\|_{\mathbb H^{-\frac12}(\partial\wD)},\\
\left\|\left(\int_{D_n}F\du x\right)_n\right\|_{l^2}
&\le C\|F\|_{\mathbb L^2(\wD)}.
\end{align*}
Boundary integrals here are understood by duality. Since $\varphi^\alpha\in\mathcal X$, its charge on each inclusion vanishes, and
\[
\int_{\partial D_n}\phi^\alpha\du\sigma
=-\bigl(\widehat{\mathbf C}^\alpha\mathbf c^\alpha\bigr)_n.
\]
The jump relation \eqref{eqn:jumprela}, Green's identity and \eqref{eqn:second_order} now yield
\begin{align*}
\int_{\partial D_n}\left\{-\frac12\operatorname{Id}+(\mathcal K^{-\alpha,0}_{\wD})^\ast\right\}[\phi^\alpha]\du\sigma&=0,\\
\int_{\partial D_n}\left\{\frac12\operatorname{Id}+(\mathcal K^{-\alpha,0}_{\wD})^\ast\right\}[\psi^\alpha]\du\sigma
&=-\bigl(\widehat{\mathbf C}^\alpha\mathbf c^\alpha\bigr)_n+\mathcal O(\omega^2),\\
\int_{\partial D_n}\mathcal K_2^\alpha[\phi^\alpha]\du\sigma
&=-|D_n|c_n^\alpha+\mathcal O(\omega^2+\delta).
\end{align*}
In these identities the remainder sequences are bounded in $l^2(\mathbb Z)$ by the indicated orders. This follows from the two bounds above and \Cref{aplem:boundedH1}; in particular, $\mathcal S^{\alpha,0}_{\wD}[\psi_j^\alpha]=\delta_{nj}$ in $D_n$.
Applying the charge map to \eqref{eqn:NPmatch}, we therefore obtain
\[
\left\|\delta\widehat{\mathbf C}^\alpha\mathbf c^\alpha
-\frac{\omega^2}{\mu_{_1}^2}\mathbf M\mathbf c^\alpha\right\|_{l^2}
\le C_\alpha(\omega^4+\delta\omega^2).
\]
The mass matrix and its inverse are bounded, since $\pi c^2\le |D_n|\le |Y_{0,0}|$. Dividing by $\delta$, multiplying by $\mathbf M^{-1/2}$ and normalizing gives
\[
\|(A^\alpha-z_\delta)\mathbf w^\alpha\|_{l^2}\le C_{\alpha,K}\delta,
\qquad
\mathbf w^\alpha=\frac{\mathbf M^{1/2}\mathbf c^\alpha}
{\|\mathbf M^{1/2}\mathbf c^\alpha\|_{l^2}}.
\]
Since $A^\alpha$ is self-adjoint, the spectral theorem proves the spectral-distance estimate.

Under the additional isolation and neighborhood conditions, the closest spectral point is $\lambda^\alpha$, so $|z_\delta-\lambda^\alpha|\le C_{\alpha,K}\delta$. Consequently,
\[
\left|\omega-\mu_{_1}\sqrt{\delta\lambda^\alpha}\right|
=\mu_{_1}\sqrt\delta\,
\frac{|z_\delta-\lambda^\alpha|}{\sqrt{z_\delta}+\sqrt{\lambda^\alpha}}
=\mathcal O(\delta^{3/2}),
\]
which implies the stated $\mathcal O(\delta)$ approximation.
\end{proof}

\begin{remark}
	Compared to the derivation of the discrete tight-binding model of electronic band structures, we use combinations of single-layer potentials $\{\mathcal{S}^{\alpha,0}_{\wD}[\psi_{n}^{\alpha}]\}_{n\in\mathbb{Z}}$ to recover the approximate eigenfunction. Here, we recall that the functions $\psi_{n}^{\alpha}$ are given in \Cref{lem:kernel}. These functions behave similarly to Wannier functions in physics or atomic orbitals in mathematical analysis; see, for example, \cite{Marzari1997,Pelinovsky2008,Pelinovsky2010,Ablowitz_2012a}.  \par 
	
\end{remark}

\begin{remark}
	In the previous sections, we have proved the exponential decay of the off-diagonal elements of the quasi-periodic capacitance matrix, when $\alpha\in [-\pi,\pi)$. However, we are unable to prove the correspondence between subwavelength frequencies and eigenvalues of $\widehat{\mathbf{C}}^{\alpha}$ for $\alpha=0$, since the corresponding single-layer potential $\mathcal{S}^{0,0}_{\wD}$ is no longer invertible. This is a limitation of the layer-potential representation: the variational capacitance matrix at $\alpha=0$ is already defined by \eqref{eqn:Harfuncp}. The present spectral derivation does not cover this phase. Asymptotic analysis is further required; see previous works \cite{Ammari_2021,Miao2024} on this issue.
\end{remark}

\begin{remark}
	\label{rem:combined_error}
	Let $m_*\coloneqq\inf_n|D_n|\ge\pi c^2$. By \Cref{cor:banded_operator_error}, the error between the mass-normalized full and banded capacitance operators is at most
	\[
	\varepsilon_N\coloneqq\frac{2C}{m_*(1-\widetilde\rho)}\widetilde\rho^{N+1}.
	\]
	For an isolated eigenvalue cluster separated from zero, take $\varepsilon_N$ smaller than its spectral separation. The induced frequency error is then $\mathcal O(\sqrt\delta\,\varepsilon_N)$. Choosing $\widetilde\rho^{N+1}\le\sqrt\delta$, for example
	\[
	N\ge\frac{\log(1/\delta)}{2|\log\widetilde\rho|},
	\]
	with $N$ an integer, makes this error $\mathcal O(\delta)$, of the order appearing in \Cref{thm:capamat_approx}. The nearest-neighbor choice $N=1$ must be assessed against the spectral separation and the desired tolerance.
\end{remark}
\section{Numerical experiments}\label{sec:numerics}
In this section, we present numerical experiments illustrating the decay of the off-diagonal elements of the capacitance matrix $\widehat{\mathbf C}^{\alpha}$. We compare full and nearest-neighbor capacitance matrices for defect and dimerized interface configurations. These comparisons assess matrix truncation, rather than the high-contrast approximation of the continuous problem. \par 
We take the square lattice for simplicity, i.e., 
\[ v_{_1} = (1,0),\quad v_{_2} = (0,1). \] 
\subsection{Exponential decay of the off-diagonal elements} 
First, we let the inclusions $\{ D_{m,n}\}_{m,n}$ be disks of radius $0.35$ centered in each cell $ Y_{m,n} $. By translation invariance with respect to $v_{_2}$, we only calculate $\{\widehat{C}^{\alpha}_{n,0}\}$ for various $\alpha\in [-\pi,\pi]$. We discretize the single-layer operator and compute 
\[ \mathcal{S}_{\wD}^{\alpha,0}[\phi_{n}^{\alpha}] = \chi_{_{\partial D_{n}}}. \] \par 
To discretize the single-layer potential operator, we consider a simple collocation boundary element method for a finite number of inclusions $D_{t}$:
\begin{equation}
	D_{\text{t}} \coloneqq \cup_{n=-14}^{14}D_{n}.
\end{equation}
This is a finite truncation in the transverse direction. It differs from a periodically repeated supercell approximation, whose convergence is studied in \cite{Soussi2005}; its accuracy must be assessed by increasing the number of retained inclusions. \par  
The diagonal element of the discretized single-layer potential operator is evaluated using Gauss quadrature with weight equal to $\log(x)$, while the off-diagonal part can be calculated efficiently by standard Gauss quadrature. We set $M=200$ to approximate the Green's functions
\begin{gather*}
	\begin{aligned}
		G^{\alpha,0}&\approx -\sum_{|k|\le M}\frac{\eu^{\iu(2\pi k+\alpha)x_1}\eu^{-|2\pi k+\alpha||x_2|}}{2|2\pi k+\alpha|},\\
		G^{0,0} &\approx \frac{|x_2|}{2} - \sum_{0<|k|\le M}\frac{\eu^{\iu2\pi kx_1}\eu^{-|2\pi kx_2|}}{|4\pi k |}.
	\end{aligned}
\end{gather*}

The first formula applies to $\alpha\ne0$. The second separates the zero Fourier mode; the nonzero-mode sum excludes $k=0$.
We compute $\{|\widehat{C}^{\alpha}_{i,0}|\}_{i=1}^{5}$ as given in \Cref{fig:ExpConvergence}. The computed entries illustrate the predicted off-diagonal decay. This motivates us to truncate the capacitance matrix to simplify the calculation.  
\begin{figure}[htbp]
	\centering
	\includegraphics[width = 0.7\textwidth]{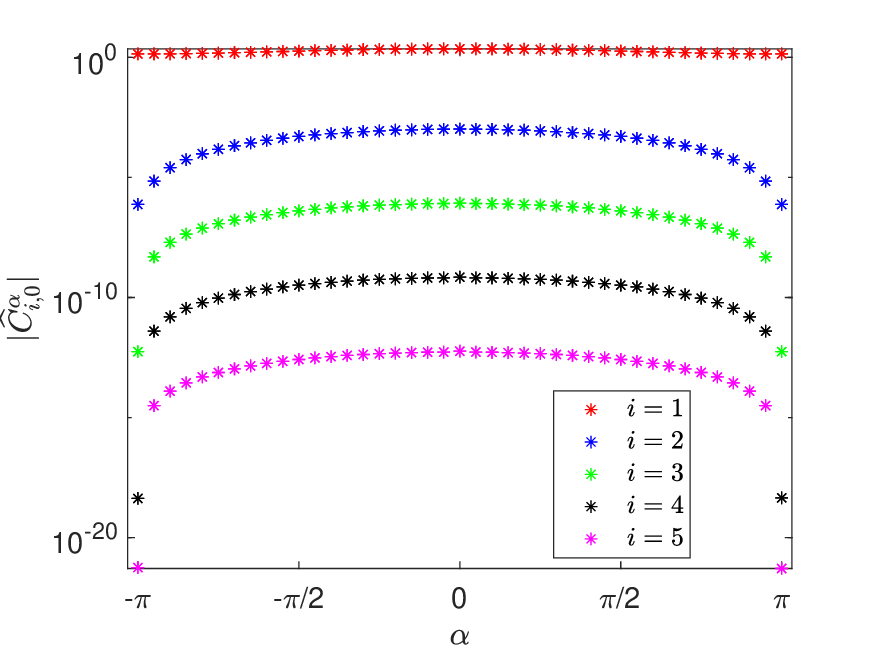}
	\caption{The exponential decay of off-diagonal elements of capacitance matrices $\widehat{\mathbf{C}}^{\alpha}$. Here, we numerically calculate $\{|\widehat{C}^{\alpha}_{i,0}|\}_{i=1}^{5}$. The curves illustrate the decrease of the computed entries as $i$ grows; their smallest values require a numerical error assessment. }
	\label{fig:ExpConvergence}
\end{figure} 

\subsection{Defect mode and truncation of the capacitance matrix}
In this subsection, we illustrate that truncating the capacitance matrix still captures the defect mode. To this end, we let the inclusion $D_0$ in the center be a disk of radius $0.2$. Radius perturbations producing line-defect modes are studied in \cite{jems2021}.  We keep the diagonal and the two adjacent diagonals of $\widehat{\mathbf{C}}^{\alpha}$ while setting all other elements to zero. The truncated matrix is denoted by $\widehat{\mathbf{C}}^{\alpha}_t$. We plot the related generalized eigenvalues, as given in \eqref{eqn:genEigProb}. We numerically calculate the eigenvalues of the original capacitance matrix $\widehat{\mathbf{C}}^{\alpha}$ and those of the truncated capacitance matrix $\widehat{\mathbf{C}}_{t}^{\alpha}$ that correspond to the region $D_{t}$. We can see from \Cref{fig:TruncaEig} that for each $\alpha$, the defect eigenvalue of the capacitance matrix can be well approximated by the truncated capacitance matrix. The maximum reported over the sampled Bloch phases is
\[ \max_{\alpha\in\mathcal A_{\mathrm{num}}}|\lambda^\alpha-\lambda^\alpha_t| = 0.02655, \]
where $\mathcal A_{\mathrm{num}}\subset[-\pi,\pi]$ denotes the numerical phase grid,
and $\lambda^\alpha$ and $\lambda^\alpha_t$ correspond to the defect eigenvalue at $\alpha$. 
The remaining finite-matrix eigenvalues approximate the bulk spectral bands; they are not themselves continuous spectrum. 

\begin{figure}[htbp]
	\centering
	\includegraphics[width = 0.7\textwidth]{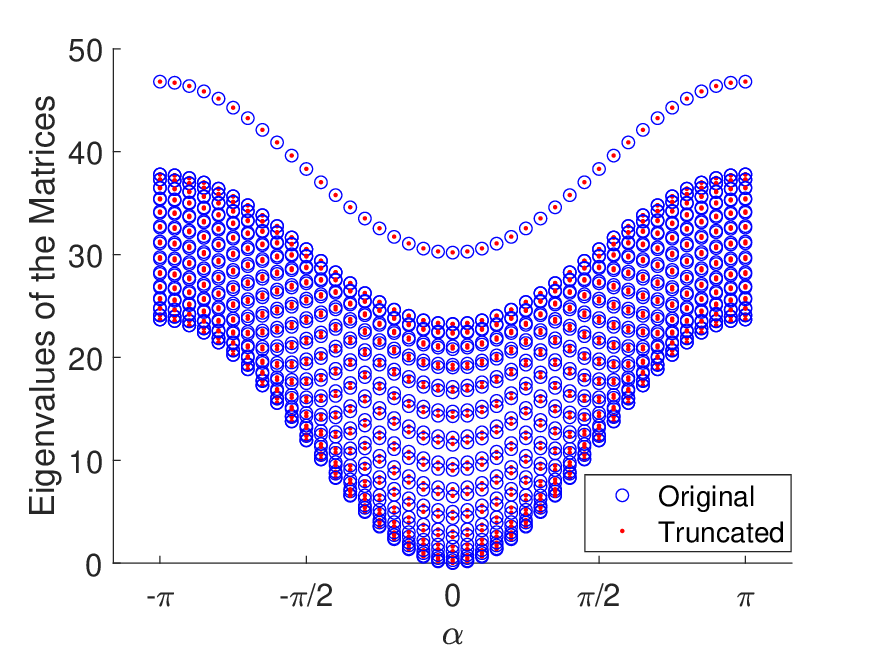}
	\caption{Comparison of the eigenvalues of the original capacitance matrix and those of the truncated capacitance matrix.  }
	\label{fig:TruncaEig}
\end{figure} 

\subsection{Interaction between defects}
In this subsection, we consider the interaction between inclusions. We let $D_{\pm l}$ be disks of radius $0.2$. We expect two defect modes. An illustration is shown in \Cref{fig:twoDefects}.
\begin{figure}[htbp]
	\centering
	\includegraphics[width = 0.7\textwidth]{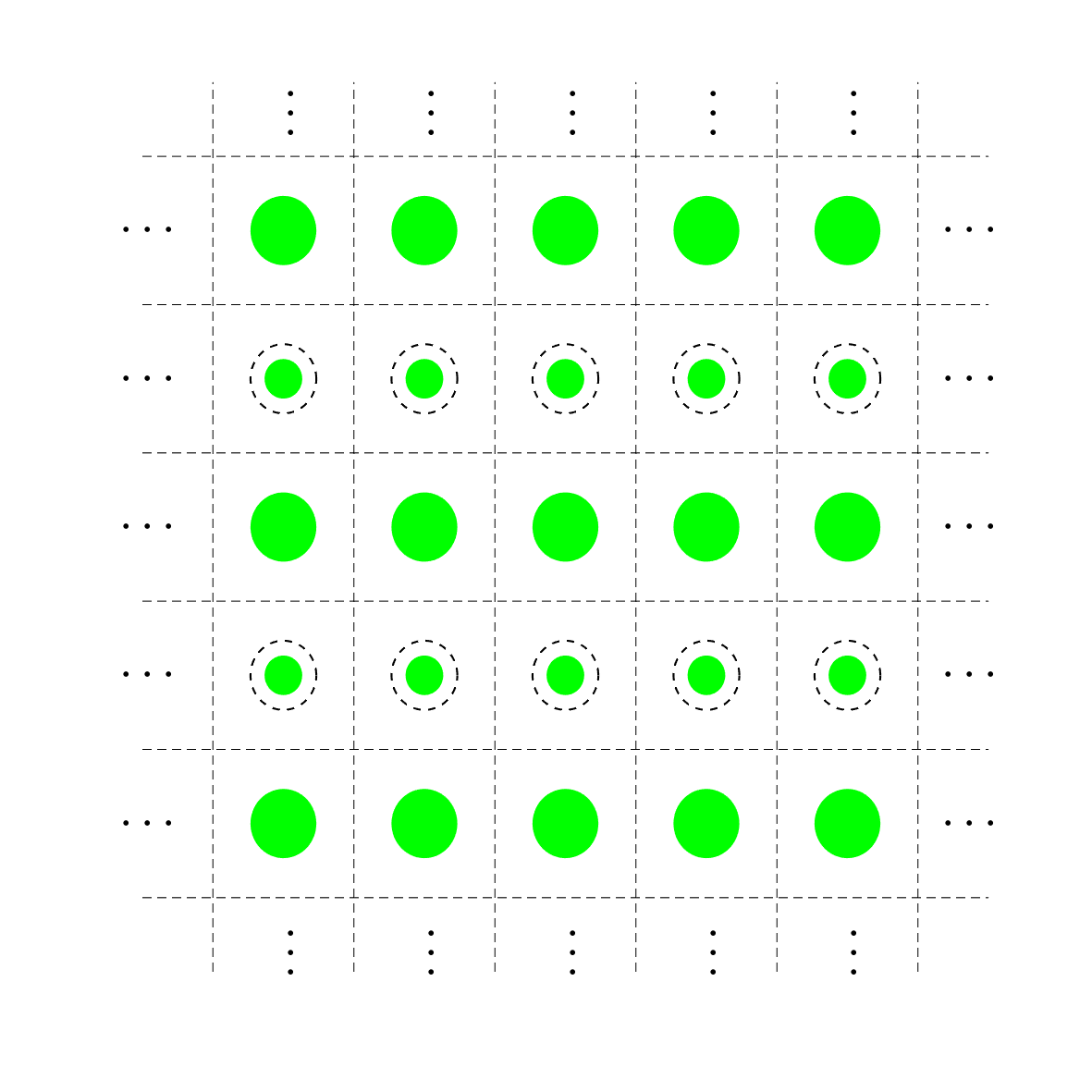}
	\caption{Illustration of a waveguide system with two defects, $l=1$. }
	\label{fig:twoDefects}
\end{figure} 

First, we calculate the corresponding eigenvalue of $\widehat{\mathbf{C}}^{\alpha}$; see \Cref{fig:TwoDefectsEig}. The defect eigenvalues are marked by larger red dots and blue circles. It is clear that as $l$ increases, the eigenvalue splitting between the two defect modes decreases.
\begin{figure}[htbp] 
	\centering
	\begin{subfigure}{0.35\linewidth}
		\centering
		\includegraphics[width=\textwidth]{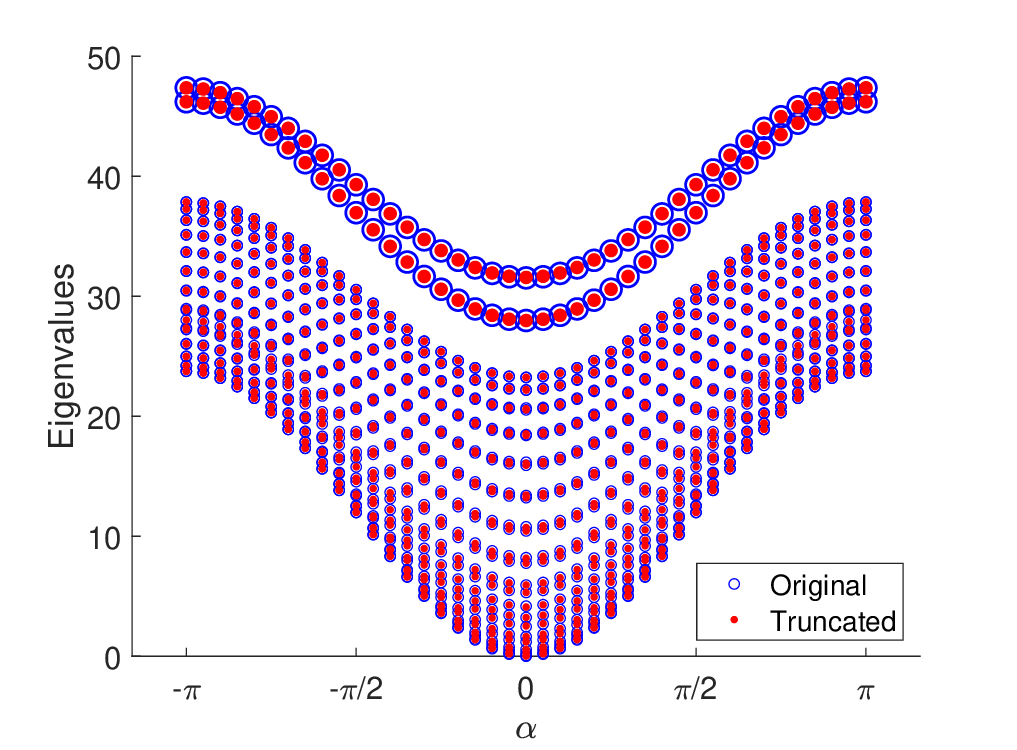}
		\caption{$l=1$.}
	\end{subfigure}
	\hspace{-0.7cm}
	\begin{subfigure}{0.34\linewidth}
		\centering
		\includegraphics[width=\textwidth]{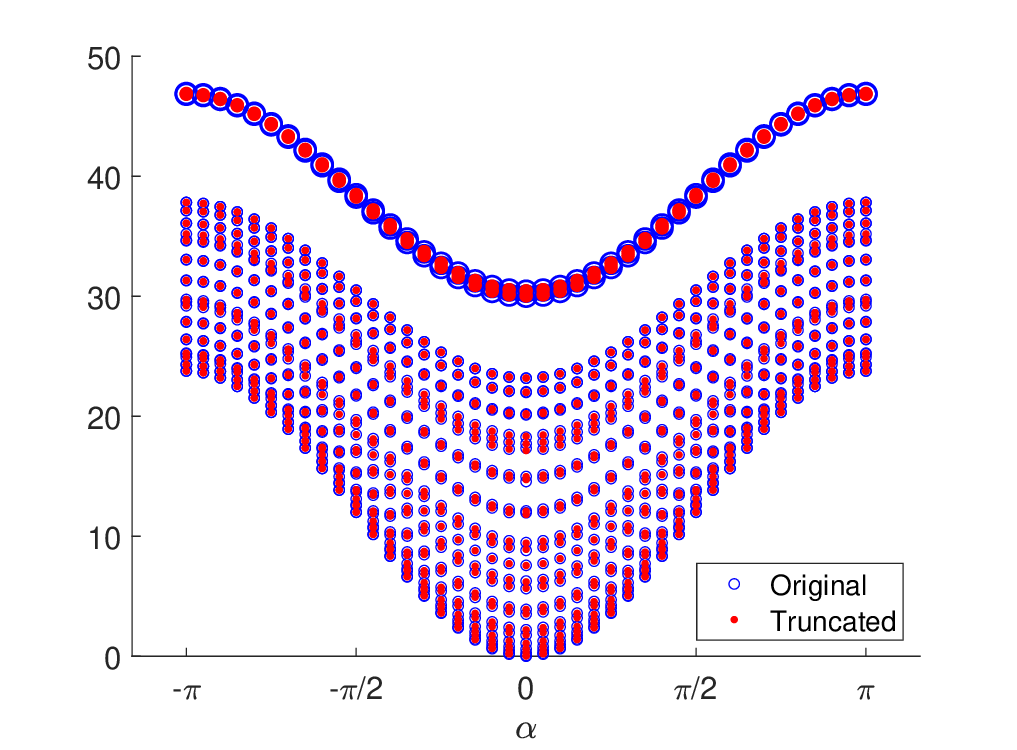}
		\caption{$l=2$.}
	\end{subfigure}
	\hspace{-0.7cm}
	\begin{subfigure}{0.34\linewidth}
		\centering
		\includegraphics[width=\textwidth]{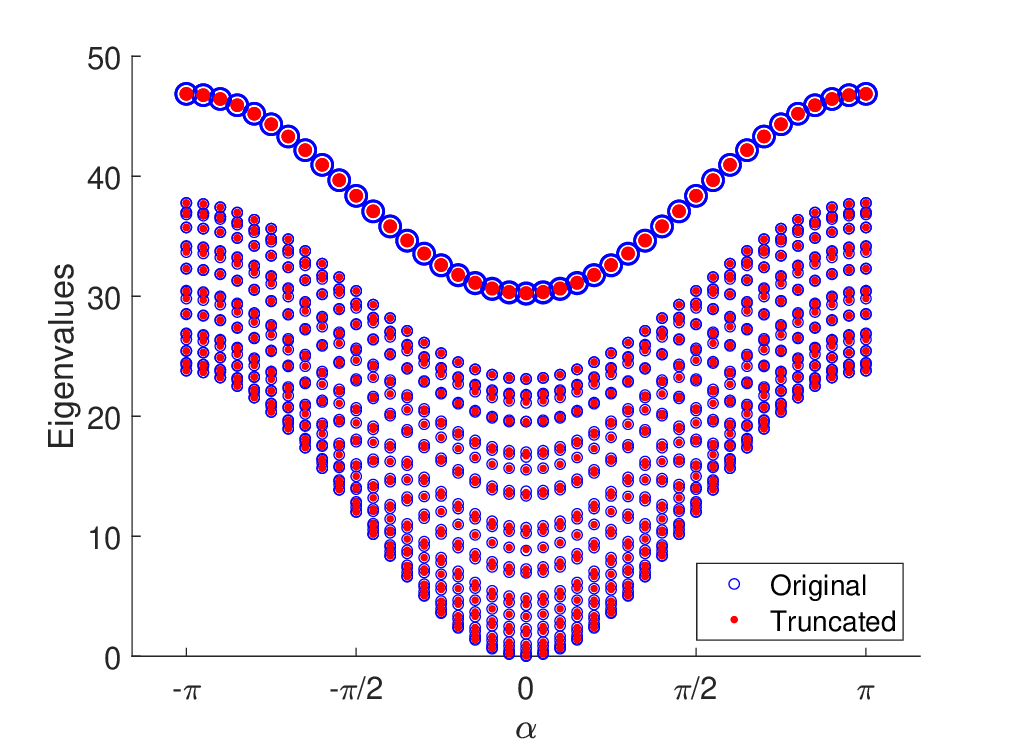}
		\caption{$l=3$.}
	\end{subfigure}
	\caption{Illustration of the interaction between defects. Here, we plot the eigenvalues of $\widehat{\mathbf{C}}^{\alpha}$ for the region $D_t$, letting $D_{\pm l}$ be disks with radius $0.2$. The defect eigenvalues are marked by larger red dots and blue circles. As $l$ increases, the eigenvalue splitting between the two defect modes becomes smaller.}
	\label{fig:TwoDefectsEig}
\end{figure}

The eigenvalue splitting can be seen more clearly in \Cref{table:compareGap}. In this table, we compute the maximal eigenvalue splitting and the maximum relative difference between the defect modes over the sampled Bloch phases. Here, the relative difference for each $\alpha$ is defined by 
\[RD(\alpha) \coloneqq 1-\frac{\lambda^\alpha_1}{\lambda_2^{\alpha}},\quad \lambda^{\alpha}_1\le \lambda^\alpha_2.\]  
\begin{table}[htbp]
	\centering
	\caption{Comparison of the eigenvalue splitting and relative difference between defect eigenvalues}
	\begin{tabular}{cccc}
		\toprule
		&{$l=1$}   &  {$l=3$}   &  {$l=5$} \\ 
		\midrule
		Maximal eigenvalue splitting for $\widehat{\mathbf{C}}^{\alpha}$  & 3.55087  & 0.41703  & $5.40939\times 10^{-2}$\\
		Maximal relative difference & 11.26227\% & 1.37087\% &0.17881\%\\
		Maximal eigenvalue splitting for $\widehat{\mathbf{C}}_t^{\alpha}$ & 3.59281  & 0.42113  & $5.46282\times 10^{-2}$\\
		Maximal relative difference & 11.38087\% & 1.38304\% &0.18042\%\\
		\bottomrule
	\end{tabular}
	\label{table:compareGap}
\end{table}
The reported splittings are consistent with exponential decrease as $l$ grows. The figure uses $l=1,2,3$, whereas the table uses $l=1,3,5$. The relative difference $RD$ measures the splitting of the two defect eigenvalues, not the error of the banded approximation. \par 
In all the cases above, numerical experiments show that the nearest-neighbor truncation captures the interactions between defects, which is similar to the analysis presented in \cite{Fefferman2025}.

\subsection{Dimerized interface modes}
We consider the two-dimensional analogue of the Su-Schrieffer-Heeger (SSH) model \cite{Su1979,Ammari2020b,2dssh}. An illustration is shown in \Cref{fig:SSHmodel}.
\begin{figure}[htbp]
	\centering
	\includegraphics[width = 0.6\textwidth]{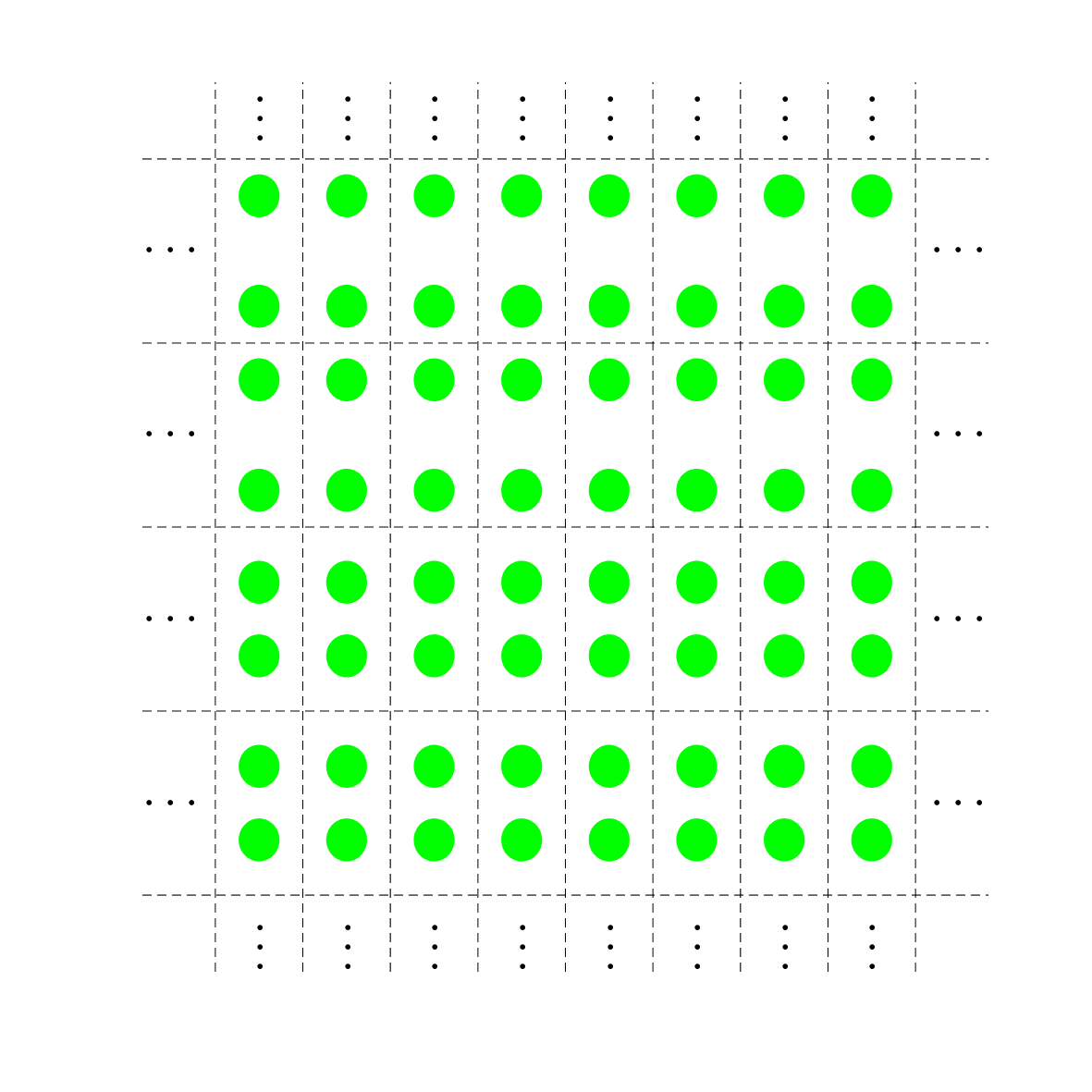}
	\caption{An illustration of the two-dimensional analogue of the SSH model. }
	\label{fig:SSHmodel}
\end{figure}
We assume that the inclusions are identical disks of radius $0.3$. The computational domain is given by 
\[ D_{t} = D_{+,t}\cup D_{-,t}, \]
where the regions $ D_{+,t} $ and $ D_{-,t} $ are given by
\begin{align*}
	D_{-,t} &= \bigcup_{n=-15}^{-1} \Big[B((0.5,2n+0.65),0.3)\cup B((0.5,2n+1.35),0.3)\Big],\\
	D_{+,t} &= \bigcup_{n=1}^{15} \Big[B((0.5,2n-1.65),0.3) \cup B((0.5,2n-0.35),0.3)\Big].
\end{align*}

The two half-arrays have a transverse bulk period of two cells, with two inclusions per bulk period. This configuration is motivated by the one-dimensional SSH resonator chain in \cite{Ammari2020b}. Establishing a topological invariant and protection for the present continuum geometry requires a separate analysis of its bulk symbols and projected gaps.

In this subsection, we directly compute the eigenvalues corresponding to the truncated capacitance matrix $\widehat{\mathbf{C}}^{\alpha}_t$. In \Cref{fig:TopoEigValue}, we plot the corresponding spectrum of $\widehat{\mathbf{C}}^{\alpha}_{t}$ of the domain $D_t$, with the original mode selection retained. The branch selected as an interface-mode candidate is marked by larger red dots. Its localization and separation from the bulk spectrum must be checked at each fixed Bloch phase; the plotted finite-matrix branch alone does not establish topological protection. 

\begin{figure}[htbp]
	\centering
	\includegraphics[width = 0.6\textwidth]{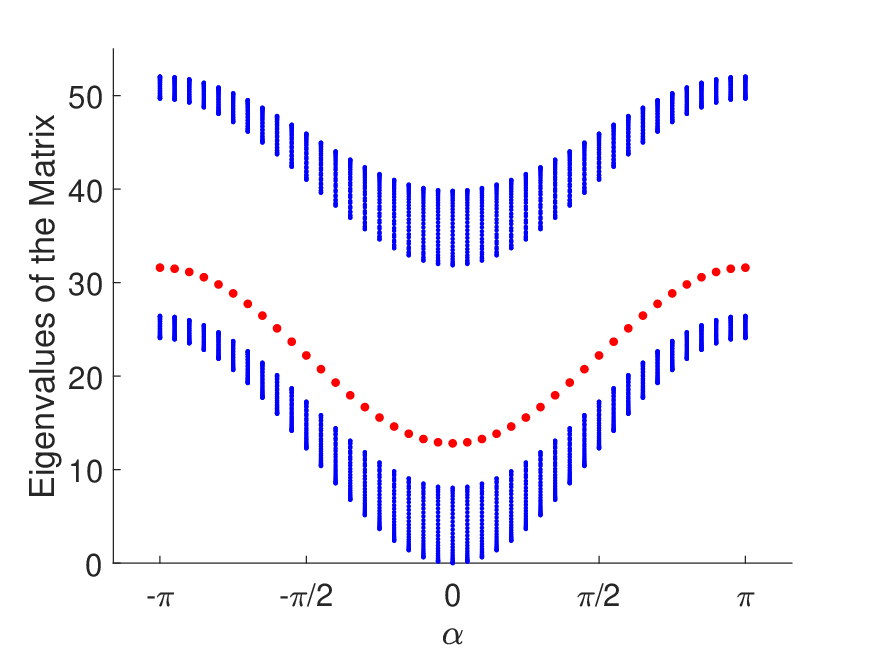}
	\caption{Interface branch in the spectrum of the truncated capacitance matrix $\widehat{\mathbf{C}}^{\alpha}_t$ for the two-dimensional analogue of the SSH model. The selected interface-mode candidate (for each quasi-periodicity $\alpha$) is plotted in larger red dots. }
	\label{fig:TopoEigValue}
\end{figure}
These results illustrate an interface-mode candidate in the nearest-neighbor capacitance model.

\section{Concluding remarks}\label{sec:conclude}
In this paper, we have introduced a new, computationally efficient, and accurate method for computing subwavelength defect bands in high-contrast subwavelength systems with line defects. Our method is based on a tight-binding approximation of the capacitance matrix associated with a periodic structure containing a non-compact line defect, with quantitative error control. The numerical examples compare full and banded capacitance models for defect states and illustrate a dimerized interface configuration. A topological characterization of the latter and a convergence study for the underlying continuous problem remain separate questions. We anticipate that our method can be applied to curved line defects and can serve as a building block for optimizing guiding properties in crystals of subwavelength resonators. We leave these extensions to future work.

\appendix

\section{Preliminaries}\label{apsec:prelim}
Before proceeding with the full analysis, we collect several preliminary definitions and results. We define the quasi-periodic Green function $ G^{\alpha,\omega} $ of the Helmholtz operator $\Delta + \omega^2$ as
\begin{equation}\label{eqn:quasiGreen}
	G^{\alpha,\omega}(x) \coloneqq \sum_{m\in\mathbb{Z}} G^{\omega}(x-(m,0))\eu^{\iu \alpha m},
\end{equation} 
where $G^\omega$ is the Green function in free space that is given by 
\begin{gather}
	G^{\omega} \coloneqq\left\{\begin{aligned}
		&\frac{1}{2\pi}\ln|x|,\quad \omega = 0,\\
		& -\frac{\iu}{4}H^{(1)}_0(\omega|x|),\quad \omega > 0.
	\end{aligned}\right. 
\end{gather}
The lattice sum in \eqref{eqn:quasiGreen} is understood in the limiting-absorption sense for positive frequency. At zero frequency the logarithmic series is not an ordinary convergent sum; for nonzero $\alpha$ we define it by the Fourier representation below. Poisson summation gives, for $|\alpha|>\omega\ge0$,
\begin{equation}\label{eqn:GreenFunc}
	G^{\alpha,\omega}(x_1,x_2) = - \sum_{k\in\mathbb{Z}}\frac{\eu^{\iu(2\pi k+\alpha )x_1}\eu^{-\sqrt{|2\pi k+\alpha|^2-\omega^2}|x_2|}}{2\sqrt{|2\pi k +\alpha|^2-\omega^2}}.
\end{equation}
The series and its derivatives converge normally on sets $|x_2|\ge\epsilon>0$. On $x_2=0$ it is understood distributionally, with logarithmic singularities at $(m,0)$, $m\in\mathbb Z$. Based on this, we define the single-layer potential $ \mathcal{S}^{\alpha,\omega}_{\wD}:\mathbb{H}^{-\frac12}(\partial \wD) \to \mathbb{H}^{\frac12}(\partial \wD) $ and the Neumann--Poincar\'e operator $ (\mathcal{K}^{-\alpha,\omega}_{\wD})^{\ast}:\mathbb{H}^{-\frac12}(\partial \wD) \to \mathbb{H}^{-\frac12}(\partial \wD) $ by
\begin{align}
	\mathcal{S}^{\alpha,\omega}_{\wD}[\varphi] & \coloneqq \int_{\partial \wD} G^{\alpha,\omega}(x-y)\varphi(y)\du\sigma(y),\label{eqn:SingleLayer}\\
	(\mathcal{K}^{-\alpha,\omega}_{\wD})^{\ast}[\varphi] & \coloneqq \int_{\partial \wD} \nu_x \cdot \nabla G^{\alpha,\omega}(x-y)\varphi(y)\du\sigma(y).
\end{align}
We prove their boundedness in \Cref{apsec:invert}. For a given $ x\in \wY \setminus \partial \wD$, the integral over $ \partial \wD $ can be decomposed as 
\[ 	\mathcal{S}^{\alpha,\omega}_{\wD}[\varphi] = \sum_{n\in\mathbb{Z}} \int_{\partial D_{0,n}}G^{\alpha,\omega}(x-y)\varphi(y)\du \sigma(y). \]
Therefore, it can be proved that $ \mathcal{S}^{\alpha,\omega}_{\wD}: \mathbb{H}^{-\frac{1}{2}}(\partial \wD) \to \mathbb{H}^{\frac{1}{2}}(\partial \wD) $ is a bounded operator when $|\alpha|>\omega\ge 0$. Furthermore, from the decay of the Green function as $ |x|\to \infty $, one has 
\begin{equation*}
	\mathcal{S}^{\alpha,\omega}_{\wD}[\varphi](x) \in \mathbb{H}^1(\wY\setminus \partial \wD)\cap \mathbb{H}^2_{\mathrm{loc}}(\wY\setminus \partial \wD),\quad x\in \wY\setminus \partial \wD ,\varphi\in \mathbb{L}^2(\partial \wD).
\end{equation*}
Moreover, the single-layer potentials satisfy the standard jump relation
\begin{equation}\label{eqn:jumprela}
	\nu\cdot \nabla \mathcal{S}^{\alpha,\omega}_{\wD}[\phi]|_{\pm} = \pm\frac12\phi + (\mathcal{K}^{-\alpha,\omega}_{\wD})^{\ast}[\phi].
\end{equation}
The invertibility of the single-layer operator is proved in \Cref{apthm:inverti}.\par 
Now, we invoke the differential Harnack's inequality, whose proof can be deduced from \cite[Theorem 3.4]{Axler2001}.
\begin{lemma}\label{lem:Harnackineq}
	Suppose that $ f $ is a positive harmonic function in $B(x_0,r)$ and continuous in $ \overline{B(x_0,r)} $. Then for all $ x\in B(x_0,r/2) $, we have
	\[ |\nabla f|(x) < \frac{4}{r}f(x). \]
\end{lemma}

\section{Extension and trace theorem}\label{apsec:extensiontrace}
The uniform Lipschitz charts in \Cref{def:decaycond2}, together with separation from the cell sides, give extension and trace bounds with constants independent of the inclusion. We use the local extension construction in \cite{Evans2010} and sum the resulting cellwise estimates.
\begin{theorem}\label{appthm:Extend}
There exist bounded linear operators
\[
E_1:\mathbb H^1(\wD)\to\mathbb H^1_\alpha(\wY),\qquad
E_2:\mathbb H^1(D_{\fu})\to\mathbb H^1(\mathbb R^2),
\]
equal to the input on $\wD$ and $D_{\fu}$, respectively. Their restrictions to each cell are supported in open neighborhoods $V_1,V_2$ satisfying
\[
D_{0,n}\Subset V_1\cap Y_{0,n}\Subset Y_{0,n},\qquad
D_{m,n}\Subset V_2\cap Y_{m,n}\Subset Y_{m,n}.
\]
Moreover,
\begin{align*}
\|E_1u\|_{\mathbb H^1(\wY)}&\le C\|u\|_{\mathbb H^1(\wD)},\\
\|E_2u\|_{\mathbb H^1(\mathbb R^2)}&\le C\|u\|_{\mathbb H^1(D_{\fu})},
\end{align*}
with $C$ independent of $\alpha$ and the cell indices.
\end{theorem}
\begin{proof}
Cover each boundary by a uniformly bounded number of charts of the fixed size in \Cref{def:decaycond2}. Local reflection across the Lipschitz graphs and a partition of unity give bounded extensions to a neighborhood of each inclusion. Multiply by a cutoff equal to one on the inclusion and zero near the cell sides. The separation assumption bounds the cutoff derivatives uniformly. Summing the squared cellwise estimates gives the assertion for $E_2$. For $E_1$, perform the same construction on $Y_{0,n}$ and continue quasi-periodically in the $v_{_1}$ direction. The extensions vanish near the cell sides, so this continuation belongs to $\mathbb H^1_\alpha(\wY)$ without an additional boundary term.
\end{proof}
Similarly, we have the following trace theorem.
\begin{theorem}\label{apthm:trace}
	For the regions $\wD$ and $ D_{\fu} $, there exist bounded linear operators 
	\[ \operatorname{Tr}_{\wD}:\mathbb{H}^1(\wD ) \to \mathbb{H}^{\frac{1}{2}}(\partial \wD),\quad  \operatorname{Tr}_{D_{\fu}}:\mathbb{H}^1( D_{\fu}) \to \mathbb{H}^{\frac{1}{2}}(\partial D_{\fu}), \]
	such that, for all $w\in \mathbb{H}^{1}(\wD)\cap C(\overline{\wD}) $ and $ v\in \mathbb{H}^{1}(D_{\fu})\cap C(\overline{D_{\fu}}) $, we have
	\begin{gather*}
		\left\{
		\begin{aligned}
			&w\mapsto \operatorname{Tr}_{\wD}w(y) = \left.w\right|_{\partial \wD}(y),\quad y\in \partial \wD,\\
			&v\mapsto \operatorname{Tr}_{D_{\fu}}v(y) = \left.v\right|_{\partial D_{\fu}}(y),\quad y\in \partial D_{\fu}.
		\end{aligned}
		\right.	
	\end{gather*}
\end{theorem}
Consider the interior boundary value problem for the Laplace equation 
\begin{gather}
	\left\{\begin{aligned}
		\Delta w_{\mathrm{int}}(x) &= 0,\quad x \in  D_{\fu},\\
		w_{\mathrm{int}}(y) &= g(y),\quad y\in \partial D_{\fu},
	\end{aligned}\right.
\end{gather}
where $ g\in \mathbb{H}^{\frac{1}{2}}(\partial D_{\fu})$. For each inclusion $ D_{m,n}\subset D_{\fu} $, it can be shown that the weak solution satisfies
\begin{equation}\label{appeqn:estimateelliptic}
	\Vert w_{\mathrm{int}}\Vert_{\mathbb{H}^{1}(D_{m,n})} \le C \Vert g \Vert_{\mathbb{H}^{\frac{1}{2}}(\partial D_{m,n})},\quad n\in \mathbb{Z}.
\end{equation}
The constant $C$ is independent of $m,n$: the chart construction gives a uniform boundary lifting, and the zero-trace remainder satisfies the Dirichlet Poincar\'e inequality after zero extension to the fixed cell. 
It immediately follows that 
\begin{equation}\label{appeqn:InteriorEstimate}
	\Vert w_{\mathrm{int}}\Vert_{\mathbb{H}^{1}(D_{\fu})} \le C \Vert g \Vert_{\mathbb{H}^{\frac{1}{2}}(\partial D_{\fu})}.
\end{equation}
Combining \Cref{appthm:Extend} and \eqref{appeqn:estimateelliptic}, we obtain the following result.
\begin{theorem}\label{apthm:extension}
	For the regions $\wD$ and $ D_{\fu} $, there exist bounded linear operators 
	\[ \mathcal{E}^{\alpha}_{1}:\mathbb{H}^{\frac{1}{2}}(\partial \wD)\to \mathbb{H}^1_{\alpha}(\wY), \quad \mathcal{E}_{2}:\mathbb{H}^{\frac{1}{2}}(\partial D_{\fu})\to \mathbb{H}^{1}(\mathbb{R}^2),\] 
	such that for $f_1\in\mathbb{H}^{\frac{1}{2}}(\partial \wD)$ and $f_2\in\mathbb{H}^{\frac{1}{2}}(\partial D_{\fu})$,
	 \begin{gather*}
	 	\left\{\begin{aligned}
	 		&(\operatorname{Tr}\mathcal{E}_{1}^{\alpha})f_1(y) \equiv f_1(y),\quad y\in \partial \wD,\\
	 		&(\operatorname{Tr}\mathcal{E}_2)f_2(y) \equiv f_2(y),\quad y\in \partial D_{\fu},
	 	\end{aligned}\right.
	 \end{gather*}
	with supports in the neighborhoods supplied by \Cref{appthm:Extend}. Moreover, the extended function $ \mathcal{E}_1^{\alpha}f_1 $ and $\mathcal{E}_2 f_{2}$ satisfy
	\begin{equation}\label{appeqn:boundedextend}
		\Vert \mathcal{E}_1^{\alpha} f_1 \Vert_{\mathbb{H}^{1}(\wY)} \le C\Vert f_1 \Vert_{\mathbb{H}^{\frac{1}{2}}(\partial \wD)},\quad \Vert \mathcal{E}_2 f_2 \Vert_{\mathbb{H}^{1}(\mathbb{R}^2)} \le C\Vert f_2 \Vert_{\mathbb{H}^{\frac{1}{2}}(\partial D_{\fu})}.
	\end{equation}
\end{theorem}

\section{Poincar\'e inequalities and existence of harmonic functions}\label{apsec:poincare}
\subsection{Poincar\'e inequalities}
In this section, we prove the following Poincar\'e inequalities on $ \wY\setminus \wD $ and $\mathbb{R}^2\setminus D_{\fu}$.
\begin{lemma}\label{aplem:poincareineq}
	For $ w\in \mathbb{H}^1_{\alpha}(\wY\setminus \wD) $ and $v\in \mathbb{H}^1(\mathbb{R}^2\setminus D_{\fu})$, such that 
	\begin{gather*}
		\left\{\begin{aligned}
			&\operatorname{Tr}_{\wD} w(y) = 0,\quad y\in \partial \wD,\\
			&\operatorname{Tr}_{D_{\fu}} v(y) = 0,\quad y\in \partial D_{\fu}.
		\end{aligned}\right.
	\end{gather*}
	there exists a constant $C$, independent of $\alpha$ and the cell indices, such that the following inequalities hold:
	\begin{align}
		\Vert w \Vert_{\mathbb{L}^2(\wY\setminus \wD)} &\le C \Vert \nabla w \Vert_{\mathbb{L}^2(\wY\setminus \wD)},\label{apeqn:poincare_q}\\
		\Vert v \Vert_{\mathbb{L}^2(\mathbb{R}^2\setminus D_{\fu})} &\le C \Vert \nabla v \Vert_{\mathbb{L}^2(\mathbb{R}^2\setminus D_{\fu})}\label{apeqn:poincare_fu}.
	\end{align}
\end{lemma}
\begin{proof}
Extend $w$ by zero into the inclusions and restrict it to a cell $Y=Y_{0,n}$. The zero trace ensures that this extension belongs to $\mathbb H^1(Y)$. Let $B\subset D_{0,n}$ be the ball from \Cref{def:decaycond2}, and let $\overline w$ denote the mean of $w$ on $Y$. The cell Poincar\'e inequality gives
\[
\|w-\overline w\|_{\mathbb L^2(Y)}\le C_Y\|\nabla w\|_{\mathbb L^2(Y)}.
\]
Since $w=0$ on $B$,
\[
|B|^{1/2}|\overline w|\le\|w-\overline w\|_{\mathbb L^2(B)},
\]
and hence
\[
\|w\|_{\mathbb L^2(Y)}
\le C_Y\left(1+\sqrt{|Y|/|B|}\right)\|\nabla w\|_{\mathbb L^2(Y)}.
\]
All cells have the same shape and $|B|\ge\pi c^2$. Summing the squared inequalities proves \eqref{apeqn:poincare_q}. Applying the same argument to every $Y_{m,n}$ proves \eqref{apeqn:poincare_fu}.
\end{proof}

In this article, we also invoke the following form of Poincar\'e's inequality.
\begin{theorem}
	For any bounded connected Lipschitz domain $\Omega$, the following estimate holds for $u\in\mathbb{H}^{1}(\Omega)$:
	\[ \Big\Vert u-\frac{1}{|\Omega|}\int_{\Omega} u\du x \Big\Vert_{\mathbb{L}^2(\Omega)}\le C \Vert \nabla u \Vert_{\mathbb{L}^2(\Omega)}. \]
\end{theorem}

For the connected inclusions $D_{m,n}$, the constant in this inequality can be chosen uniformly under \Cref{def:decaycond2}. Indeed, the fixed chart size and Lipschitz bound give uniform local Poincar\'e inequalities on boundary patches and interior balls. Connectedness and the fixed cell diameter allow a bounded chain of overlapping patches; comparing their means gives the global inequality with a constant depending only on the stated geometric parameters.

\subsection{Existence of harmonic functions}
Now, we investigate the existence of exterior boundary value problems \eqref{eqn:Harfuncp} and \eqref{eqn:Harfunc} on $ \wY\setminus \wD $ and $\mathbb{R}^2\setminus D_{\fu}$, respectively. 

\begin{theorem}\label{apthm:existenceHar}
	The boundary value problems \eqref{eqn:Harfuncp} and \eqref{eqn:Harfunc} have unique solutions in the $\mathbb H^1$ classes specified there.
\end{theorem}

\begin{proof}
	Define the spaces for test functions:
	\begin{align*}
		\mathbb{V}_\alpha&=\{z\in\mathbb H^1_\alpha(\wY):z=0\text{ a.e.\ in }\wD\},\\
		\mathbb{V}&=\{z\in\mathbb H^1(\mathbb R^2):z=0\text{ a.e.\ in }D_{\fu}\}.
	\end{align*}
	First we investigate the existence and uniqueness of the solution to \eqref{eqn:Harfuncp} with quasi-periodic boundary condition when $g\in\mathbb{H}^{\frac{1}{2}}(\partial \wD)$: 
	\begin{gather}
		\left\{\begin{aligned}
			&\Delta w^{\alpha}_{\mathrm{ext}}(x) = 0,\quad x \in \wY\setminus \wD,\\
			&w^\alpha_{\mathrm{ext}}(y) = g(y),\quad y\in \partial \wD,\\
			&w^{\alpha}_{\mathrm{ext}}(x+mv_{_1}) = \eu^{\iu m \alpha}w^{\alpha}_{\mathrm{ext}}(x),\quad x\in \wY\setminus \wD.
		\end{aligned}\right.
	\end{gather}
	Consider the following weak formulation:
	\[
	\int_{\wY\setminus \wD} \overline{\nabla U^{\alpha}}\cdot \nabla v^{\alpha} = -\int_{\wY\setminus \wD} \overline{\nabla \mathcal{E}^{\alpha}_1g}\cdot \nabla v^{\alpha},\quad U^{\alpha},v^{\alpha}\in\mathbb{V}_\alpha.
	\]
	By Poincar\'e's inequality and the extension theorem, the sesquilinear form on the left hand side is coercive. Therefore, the weak solution $U^\alpha$ exists and is unique by Lax--Milgram's theorem. The required exterior solution is $w_{\mathrm{ext}}^\alpha=(\mathcal E_1^\alpha g+U^\alpha)|_{\wY\setminus\overline{\wD}}$. 
	
	Repeating the above proof on $\mathbb R^2\setminus\overline{D_{\fu}}$, we conclude that the solution to \eqref{eqn:Harfunc} exists and is unique.
\end{proof}

\section{Poisson's Kernel for the half space}\label{apsec:Poisson}
In this section, we briefly review some facts about the Poisson kernel of the half space. For simplicity, we take the Poisson kernel in $ \mathbb{R}^{2}_{+} $ as an example. \par 
Consider the boundary value problem for the Laplace equation in half space $ \mathbb{R}^{2}_{+} $:
\begin{gather}
	\left\{\begin{aligned}
		\Delta w(x) &= 0,\quad x \in \mathbb{R}^2_+,\\
		w(x_1,0) &= g(x_1),\quad x_1\in \mathbb{R},\\
		\sup_{x\in\mathbb{R}^{2}_{+}}|w(x)|&<+\infty.
	\end{aligned}\right.
\end{gather}
Here we investigate the following two cases:
\begin{itemize}
	\item The function $ w $ satisfies the periodic boundary condition:
	\[ w(x+mv_{_1}) = w(x),\quad x\in \mathbb{R}^2_+,m\in\mathbb{Z}. \]
	The boundary value $g\in \mathbb L^\infty(\mathbb R)$ is also periodic;
	\[ g(x_1+m) = g(x_1),\quad m\in\mathbb{Z}. \]
	\item The boundary value $ g(x_1) $ belongs to $\mathbb L^1(\mathbb R)\cap\mathbb L^2(\mathbb R)\cap\mathbb L^\infty(\mathbb R)$.
\end{itemize}
Boundary values are understood as nontangential limits almost everywhere. The Poisson formulas give the unique bounded harmonic solutions in these classes; for continuous bounded data see \cite[Theorems 7.3--7.5]{Axler2001}, and the bounded measurable case follows by the Poisson representation and approximation. 
\begin{proposition}\label{appprop:representation}
	For the boundary value problem on the half space, one has the following results:
	\begin{itemize}
		\item[(1)] If the bounded solution $w$ satisfies a periodic boundary condition, then it is given by
		\begin{equation}\label{apeqn:periodpoisson}
			w(x_1,x_2) = \sum_{n\in \mathbb{Z}} \widehat{g}(n)\eu^{-2\pi |n|x_2 }\eu^{\iu 2\pi n x_1 } = \int_{0}^{1} P_{x_2}^{\mathrm{per}}(x_1-s) g(s)ds.
		\end{equation}
		Here, $(x_1,x_2)\in\mathbb{R}^2_{+}$, $ \widehat{g}(n) $ denotes the Fourier coefficients, and $P_{t}^{\mathrm{per}}$ is the periodic Poisson's kernel given by 
		\begin{equation}\label{appeqn:poissonkernel}
			P_{t}^{\mathrm{per}}(s) \coloneqq 
			\frac{\sinh(2\pi t)}{\cosh(2\pi t)-\cos(2\pi s)}.
		\end{equation}
		\item[(2)] If $g\in\mathbb L^1(\mathbb R)\cap\mathbb L^2(\mathbb R)\cap\mathbb L^\infty(\mathbb R)$, then the bounded solution $ w $ is given by
		\begin{equation}\label{apeqn:freepoisson}
			w(x_1,x_2) = \mathcal{F}^{-1}\left[ \widehat{g}(\xi) \eu^{-2\pi x_2|\xi|} \right] = \int_{\mathbb{R}}P_{x_2}(x_1-s)g(s)ds,\quad (x_1,x_2)\in \mathbb{R}^2_+. 
		\end{equation}
		Here, $(x_1,x_2)\in\mathbb{R}^2_{+}$, $ \widehat{g}(\xi) $ is the Fourier transform of $ g(x_1) $ and $ \mathcal{F}^{-1} $ denotes the inverse Fourier transform. The Poisson kernel on the half space $\mathbb{R}^2_{+}$ is given by
		\begin{equation}\label{appeqn:poissonkernel2}
			P_{t}(s) \coloneqq 
			\frac{t}{\pi(t^2+s^2)}.
		\end{equation}
		\item[(3)] In either case, if $g\ge0$ almost everywhere and $g$ is not identically zero, then its Poisson extension is strictly positive in $\mathbb R^2_+$.
	\end{itemize}
\end{proposition}

\section[Properties of layer potential operators]{\texorpdfstring{Properties of layer potential operators on $\partial \wD$}{Properties of layer potential operators}}\label{apsec:invert}
In this section, we give some properties of the layer-potential operators on the infinite boundary $\partial \wD$. For $\alpha\neq 0$, we will prove the boundedness of the single-layer operator $\mathcal{S}^{\alpha,0}_{\wD}:\mathbb{H}^{-\frac{1}{2}}(\partial \wD) \to \mathbb{H}^{\frac{1}{2}}(\partial \wD)$ and the Neumann-Poincar\'e operator $(\mathcal{K}^{-\alpha,0}_{\wD})^{\ast}:\mathbb{H}^{-\frac{1}{2}}(\partial \wD) \to \mathbb{H}^{-\frac{1}{2}}(\partial \wD)$. Furthermore, we prove the invertibility of the single-layer potential when $0\le \omega\ll 1$. Denoting
\begin{equation}\label{apeqn:decomposi}
	\phi = \sum_{n\in \mathbb{Z}}\phi_{n}\chi_{_{\partial D_{n}}},\quad \phi_{n}\in\mathbb{H}^{-\frac{1}{2}}(\partial D_n), 
\end{equation} 
we first separate the single-layer potential and the Neumann-Poincar\'e operator on the infinite boundary $\partial \wD = \cup_{m\in\mathbb{Z}} \partial D_{m}$ by 
\begin{align*}
	\mathcal{S}^{\alpha,0}_{\wD}[\phi]\big|_{\partial D_m} = \sum_{n\in \mathbb{Z}} \mathcal{S}^{\alpha,0}_{mn}[\phi_{n}],\quad 
	(\mathcal{K}^{-\alpha,0}_{\wD})^{\ast}[\phi]\big|_{\partial D_{m}} = \sum_{n\in \mathbb{Z}} \mathcal{K}^{\alpha,0}_{mn}[\phi_{n}].
\end{align*}
Here, the block operators
\begin{align*}
\mathcal S_{mn}^{\alpha,0}&:\mathbb H^{-\frac12}(\partial D_n)\to\mathbb H^{\frac12}(\partial D_m),\\
\mathcal K_{mn}^{\alpha,0}&:\mathbb H^{-\frac12}(\partial D_n)\to\mathbb H^{-\frac12}(\partial D_m)
\end{align*}
are defined by 
\begin{align}
	\mathcal{S}^{\alpha,0}_{mn}[\phi_{n}]&\coloneqq \int_{\partial D_{n}} G^{\alpha,0}(x-y)\phi_n\,d\sigma(y),\quad x\in \partial D_{m},\\
	\mathcal{K}^{\alpha,0}_{mn}[\phi_{n}]&\coloneqq \int_{\partial D_{n}}\nu_x\cdot \nabla G^{\alpha,0}(x-y)\phi_n\,d\sigma(y),\quad x\in \partial D_{m}.
\end{align} 
We recall that the Green's function $G^{\alpha,0}$ is given by \eqref{eqn:GreenFunc}. \par 
From \Cref{def:decaycond2}, the single-layer potential and the Neumann-Poincar\'e operators on each boundary, $\mathcal{S}_{nn}^{\alpha,0}$ and $\mathcal{K}^{\alpha,0}_{nn}$, are uniformly bounded, i.e., there exist $C$ such that 
\begin{gather}
	\left\{\begin{aligned}
		\Vert \mathcal{S}^{\alpha,0}_{nn}[\phi_{n}] \Vert_{\mathbb{H}^{\frac{1}{2}}(\partial D_{n})}\le C \Vert \phi_{n} \Vert_{\mathbb{H}^{-\frac{1}{2}}(\partial D_{n})},\\
		\Vert \mathcal{K}^{\alpha,0}_{nn}[\phi_{n}] \Vert_{\mathbb{H}^{-\frac{1}{2}}(\partial D_{n})}\le C \Vert \phi_{n} \Vert_{\mathbb{H}^{-\frac{1}{2}}(\partial D_{n})}
	\end{aligned}\right., \quad \forall \phi_{n}\in\mathbb{H}^{-\frac{1}{2}}(\partial D_n),\; n\in\mathbb{Z}.
\end{gather}
We first discuss the boundedness of the layer potential operators $\mathcal{S}^{\alpha,0}_{\wD}$ and $(\mathcal{K}^{-\alpha,0}_{\wD})^{\ast}$.
\begin{theorem}\label{apthm:boundsingle}
	The single-layer potential $\mathcal{S}^{\alpha,0}_{\wD}:\mathbb{H}^{-\frac{1}{2}}(\partial \wD)\to \mathbb{H}^{\frac{1}{2}}(\partial \wD)$ and the Neumann-Poincar\'e operator $(\mathcal{K}^{-\alpha,0}_{\wD})^{\ast}:\mathbb{H}^{-\frac{1}{2}}(\partial \wD)\to \mathbb{H}^{-\frac{1}{2}}(\partial \wD)$ are bounded operators. 
\end{theorem}
\begin{proof}
	First, we consider $\mathcal{S}^{\alpha,0}_{\wD}[\phi]$ on each $\partial D_m$. We have
	\[ \mathcal{S}^{\alpha,0}_{\wD}[\phi]\big|_{\partial D_{m}} = \mathcal{S}^{\alpha,0}_{mm}[\phi_{m}]+ \sum_{n\neq m}\mathcal{S}^{\alpha,0}_{mn}[\phi_{n}]. \]
	By the exponential decay of the Green function \eqref{eqn:GreenFunc} with respect to $x_2$, we obtain that
	\[ \Vert \mathcal{S}^{\alpha,0}_{mn}[\phi_{n}] \Vert_{\mathbb{H}^{\frac{1}{2}}(\partial D_m)}\le C_\alpha\eu^{-c_\alpha|m-n|} \Vert \phi_{n} \Vert_{\mathbb{H}^{-\frac{1}{2}}(\partial D_n)}.\] 
	Here one may take any fixed $0<c_\alpha<|\alpha|h$, increasing $C_\alpha$ for neighboring cells. The uniform boundary charts convert the smooth off-diagonal kernel bounds to the displayed Sobolev bounds. The summability of $\{\eu^{-c_\alpha|j|}\}_{j\in\mathbb Z}$ and the discrete Schur test give the boundedness of the single-layer operator. Similar procedures can be carried out to prove the boundedness of the Neumann-Poincar\'e operator $(\mathcal{K}^{-\alpha,0}_{\wD})^{\ast}:\mathbb{H}^{-\frac{1}{2}}(\partial \wD)\to \mathbb{H}^{-\frac{1}{2}}(\partial \wD)$. 
\end{proof}
From \Cref{apsec:extensiontrace} and \Cref{apsec:poincare}, we can also prove the following lemma.
\begin{lemma}\label{aplem:boundedH1}
	For $\alpha\neq 0$, the single-layer operator $\mathcal{S}^{\alpha,0}_{\wD}[\phi]$ defines a harmonic function on $\wY\setminus \partial \wD$. Moreover, it satisfies 
	\begin{equation}
		\Vert \mathcal{S}^{\alpha,0}_{\wD}[\phi] \Vert_{\mathbb{H}^{1}(\wD)} + \Vert \mathcal{S}^{\alpha,0}_{\wD}[\phi] \Vert_{\mathbb{H}^{1}(\wY\setminus \wD)} \le C\Vert  \phi\Vert_{\mathbb{H}^{-\frac{1}{2}}(\partial \wD)}.
	\end{equation}
\end{lemma}
\begin{proof}
For fixed $\alpha\ne0$, Fourier expansion gives $\|v\|_{\mathbb L^2(\wY)}\le|\alpha|^{-1}\|\partial_{x_1}v\|_{\mathbb L^2(\wY)}$. The trace bound and Lax--Milgram therefore give a unique $u\in\mathbb H^1_\alpha(\wY)$ satisfying
\[
\int_{\wY}\overline{\nabla u}\cdot\nabla v\du x
=-\langle\phi,\operatorname{Tr}v\rangle,
\qquad v\in\mathbb H^1_\alpha(\wY),
\]
with $\|u\|_{\mathbb H^1(\wY)}\le C_\alpha\|\phi\|_{\mathbb H^{-\frac12}(\partial\wD)}$. For smooth densities supported on finitely many components, the Green function identifies $u$ with $\mathcal S^{\alpha,0}_{\wD}[\phi]$. Density gives the same identity for general $\phi$ and proves the estimate.
\end{proof}
Now, we proceed to prove the invertibility of the single-layer operator. Define the sesquilinear form $a(\cdot,\cdot)$ by
\[ a(\phi,\psi) \coloneqq -\langle \phi,\mathcal{S}^{\alpha,0}_{\wD}[\psi] \rangle. \]
We first establish the coercivity of this sesquilinear form 
\begin{lemma}\label{aplem:coer}
	The sesquilinear form $a(\cdot,\cdot)$ is Hermitian and satisfies the coercivity estimate
	\begin{equation}\label{apeqn:coercivity}
		a(\phi,\phi)\ge c\Vert \phi  \Vert^2_{\mathbb{H}^{-\frac{1}{2}}(\partial \wD)} ,\quad \phi\in\mathbb{H}^{-\frac{1}{2}}(\partial \wD),
	\end{equation}
	where $c$ is a positive real constant.
\end{lemma}

\begin{proof}
	From the definition of the Green's function, one can directly verify that the sesquilinear form is Hermitian:   
	\begin{align*}
		a(\phi,\psi) &= -\int_{\partial \wD}\int_{\partial \wD} \overline{\phi(x)}G^{\alpha,0}(x-y)\psi(y)\du \sigma(y)\du\sigma(x)\\
		&=-\int_{\partial \wD}\int_{\partial \wD} \overline{G^{\alpha,0}(y-x)\phi(x)}\psi(y)\du \sigma(y)\du\sigma(x) = \overline{a(\psi,\phi)}.
	\end{align*}
	Green's identity gives $a(\phi,\psi)=\int_{\wY}\overline{\nabla\mathcal S^{\alpha,0}_{\wD}[\phi]}\cdot\nabla\mathcal S^{\alpha,0}_{\wD}[\psi]\du x$. To prove coercivity, consider the following duality pairing: 
	\[
	\langle \phi,f\rangle,\quad \phi\in \mathbb{H}^{-\frac{1}{2}}(\partial \wD),\quad  f\in \mathbb{H}^{\frac{1}{2}}(\partial \wD).
	\]
	By the extension theorem in \Cref{apsec:extensiontrace}, we can extend the function $f\in \mathbb{H}^{\frac{1}{2}}(\partial \wD)$ to a function $\mathcal{E}_1^{\alpha}f\in \mathbb{H}^{1}_\alpha(\wY)$ satisfying
	\[
	\Vert \mathcal{E}_1^{\alpha} f \Vert_{\mathbb{H}^{1}(\wY)} \le C\Vert f \Vert_{\mathbb{H}^{\frac{1}{2}}(\partial \wD)}.
	\]
	Applying Green's identity and the jump relations, we have
	\[
	\langle \phi,f\rangle = -\int_{\wY} \overline{\nabla \mathcal{S}^{\alpha,0}_{\wD}[\phi] } \cdot \nabla \mathcal{E}_1^{\alpha} fdx.
	\]
	By the boundedness of the single-layer operator and \Cref{apthm:extension}, we obtain 
	\begin{equation}
		|\langle \phi,f\rangle| \le C \Vert \nabla \mathcal{S}^{\alpha,0}_{\wD}[\phi] 
		\Vert_{\mathbb{L}^2(\wY)}\Vert f
		\Vert_{\mathbb{H}^{\frac{1}{2}}(\partial \wD)} = C\sqrt{a(\phi,\phi)}\Vert f
		\Vert_{\mathbb{H}^{\frac{1}{2}}(\partial \wD)}.
	\end{equation}
	
	Taking the supremum over all $f\in \mathbb{H}^{\frac{1}{2}}(\partial \wD)$ with $\Vert f \Vert_{\mathbb{H}^{\frac{1}{2}}(\partial \wD)}=1 $, we arrive at 
	\[
	\frac{1}{C}\Vert \phi \Vert_{\mathbb{H}^{-\frac{1}{2}}(\partial \wD)}^2 \le a(\phi,\phi),
	\]
	which implies the desired estimate \eqref{apeqn:coercivity}.
\end{proof}
By the Lax-Milgram theorem, we directly deduce the following result.
\begin{theorem}\label{apthm:inverti}
	The single-layer potential $ \mathcal{S}_{\wD}^{\alpha,0}:\mathbb{H}^{-\frac{1}{2}}(\partial \wD)\to \mathbb{H}^{\frac{1}{2}}(\partial \wD) $ is invertible for $ \alpha\neq 0 $.
\end{theorem}
From \Cref{aplem:coer} and the boundedness of the single-layer potential, we notice that $a(\cdot,\cdot)$ defines an equivalent inner product on $\mathbb{H}^{-\frac12}(\partial \wD)$.
\begin{definition}
	The equivalent inner product on $\mathbb{H}^{-\frac12}(\partial \wD)$ is given by 
	\[
	( \phi,\psi)_\ast \coloneqq a(\phi,\psi).
	\] 
\end{definition}
The following subthreshold extension follows from coercivity on the whole fiber.
\begin{corollary}\label{apcor:singlelayer_smallk}
	Fix $\alpha\in[-\pi,\pi)\setminus\{0\}$. For every real
	$0\le k<|\alpha|$, the single-layer operator
	\[
	\mathcal S_{\wD}^{\alpha,k}:
	\mathbb H^{-1/2}(\partial\wD)
	\longrightarrow
	\mathbb H^{1/2}(\partial\wD)
	\]
	is invertible. Its inverse norm is bounded independently of real $k\in[0,|\alpha|)$. Moreover, it satisfies 
	\begin{equation}
		\Vert \mathcal{S}^{\alpha,k}_{\wD}[\phi] \Vert_{\mathbb{H}^{1}(\wD)} + \Vert \mathcal{S}^{\alpha,k}_{\wD}[\phi] \Vert_{\mathbb{H}^{1}(\wY\setminus \wD)} \le C_{\alpha,k}\Vert  \phi\Vert_{\mathbb{H}^{-\frac{1}{2}}(\partial \wD)}.
	\end{equation}
The field-norm constant is uniform for $0\le k\le k_*<|\alpha|$.
\end{corollary}

\begin{proof}
	Set $E_\pm=\mathbb H^{\pm1/2}(\partial\wD)$ and
	$V_\alpha=\mathbb H^1_\alpha(\wY)$. Fourier expansion
	in the $v_{_1}$ direction gives
	\[
	\|v\|_{\mathbb L^2(\wY)}
	\le |\alpha|^{-1}
	\|\partial_{x_1}v\|_{\mathbb L^2(\wY)},
	\qquad v\in V_\alpha.
	\]
	Consequently, the Hermitian form
	\[
	B_k(u,v)
	=\int_{\wY}
	\left(\overline{\nabla u}\cdot\nabla v
	-k^2\overline u\,v\right)\du x
	\]
	satisfies
	\[
	B_k(v,v)
	\ge
	\frac{|\alpha|^2-k^2}{1+|\alpha|^2}
	\|v\|_{\mathbb H^1(\wY)}^2.
	\]
	Here the integrals include both the inclusions and
	their exterior.
	
	By the trace estimate and the Lax--Milgram theorem,
	for each $\phi\in E_-$ there exists a unique
	$u_\phi\in V_\alpha$ satisfying
	\[
	B_k(u_\phi,v)
	=-\langle\phi,\operatorname{Tr}v\rangle,
	\qquad v\in V_\alpha.
	\]
	The Green function representation
	\eqref{eqn:GreenFunc} identifies $u_\phi$ with
	$\mathcal S_{\wD}^{\alpha,k}[\phi]$ for smooth
	densities supported on finitely many boundary
	components. The bounded solution map and density
	extend this identity to all $\phi\in E_-$.
	In particular, its boundary trace is bounded
	from $E_-$ to $E_+$.
	
	Define
	\[
	a_k(\phi,\psi)
	=-\langle\phi,\mathcal S_{\wD}^{\alpha,k}[\psi]\rangle
	=B_k(u_\phi,u_\psi).
	\]
	Thus $a_k$ is bounded and Hermitian.
	Using the boundary lifting from
	Appendix~\ref{apsec:extensiontrace}, take
	$\mathcal E_1^\alpha f\in V_\alpha$ such that
	\[
	\operatorname{Tr}(\mathcal E_1^\alpha f)=f,
	\qquad
	\|\mathcal E_1^\alpha f\|_{\mathbb H^1(\wY)}
	\le C_\alpha\|f\|_{E_+}.
	\]
	Cauchy--Schwarz for $B_k$ yields
	\begin{align*}
		|\langle\phi,f\rangle|
		&=|B_k(u_\phi,\mathcal E_1^\alpha f)|\\
		&\le
		\sqrt{a_k(\phi,\phi)}
		\sqrt{B_k(\mathcal E_1^\alpha f,
			\mathcal E_1^\alpha f)}\\
		&\le
		C_\alpha\sqrt{a_k(\phi,\phi)}\,\|f\|_{E_+}.
	\end{align*}
	Taking the supremum over $\|f\|_{E_+}=1$ gives
	\[
	a_k(\phi,\phi)\ge c_\alpha\|\phi\|_{E_-}^2,
	\]
	where $c_\alpha>0$ is independent of $k$.
	The Lax--Milgram theorem applied to $a_k$ therefore
	gives, for every $f\in E_+$, a unique $\psi\in E_-$
	such that
	\[
	a_k(\phi,\psi)=-\langle\phi,f\rangle,
	\qquad \phi\in E_-.
	\]
	Hence $\mathcal S_{\wD}^{\alpha,k}[\psi]=f$, and
	the coercivity estimate gives the asserted
	inverse bound.
\end{proof}

\section[Invertibility on the zero-charge subspace]{\texorpdfstring{Invertibility of the operator $-\frac{1}{2}\operatorname{Id} + (\mathcal{K}^{-\alpha,0}_{\wD})^\ast$}{Invertibility on the zero-charge subspace}}\label{apsec:Invert_NP}
In this section, we prove the invertibility of the operator $-\frac{1}{2}\operatorname{Id} + (\mathcal{K}^{-\alpha,0}_{\wD})^\ast$ on a subspace of $\mathbb{H}^{-\frac{1}{2}}(\partial \wD)$. Since $\wD$ is unbounded, the Neumann-Poincar\'e operator is no longer compact. Therefore,  we cannot adopt the standard Fredholm theory. We first give some definitions. 
\begin{definition}
	Assuming the decomposition of $\phi\in\mathbb{H}^{-\frac{1}{2}}(\partial \wD)$, its charge pairing on each inclusion $\partial D_n$ is given by duality as
	\begin{equation}
		\langle  \phi  \rangle_n \coloneqq \langle \phi_n,\chi_{_{\partial D_{n} }} \rangle.
	\end{equation}
	Furthermore, we define the following closed subspaces of $ \mathbb{H}^{-\frac12}(\partial \wD)$
	\begin{align}
		&\mathcal{D} \coloneqq  \ker \Big\{-\frac{1}{2}\operatorname{Id} + (\mathcal{K}^{-\alpha,0}_{\wD})^\ast\Big\},\\
		&\mathcal X \coloneqq\Big\{\phi\in \mathbb H^{-\frac12}(\partial \wD) : \langle \phi\rangle_n=0,\ \forall n\in\mathbb{Z}\Big\}.
	\end{align}
\end{definition}
We have the following result concerning the mapping property of the Neumann-Poincar\'e operator. 
\begin{lemma}\label{aplem:mapping}
	For every $\phi\in \mathbb{H}^{-\frac{1}{2}}(\partial \wD)$ and $n\in\mathbb{Z}$,
	\[
	\big\langle (\mathcal{K}^{-\alpha,0}_{\wD})^\ast[ \phi] \big\rangle_n \ = \frac12\langle \phi\rangle_n.
	\]
\end{lemma}
\begin{proof}
	Let $w^\alpha=\mathcal{S}^{\alpha,0}_{\wD}[\phi]$. From the jump relation \eqref{eqn:jumprela}, we have
	\[ \langle -\tfrac12\phi+ (\mathcal{K}^{-\alpha,0}_{\wD})^\ast[ \phi]  , \chi_{_{\partial D_{n} }}\rangle = \langle \nu\cdot \nabla w^\alpha|_{-}  , \chi_{_{\partial D_{n} }}\rangle. \]
	Here, $\nu\cdot \nabla w^\alpha|_{-}$ denotes the interior limit of the normal derivative of $w^{\alpha}$. By Green's identity,  we have 
	\[ \langle \nu\cdot \nabla w^\alpha|_{-}  , \chi_{_{\partial D_{n} }}\rangle = \int_{D_{n}}\overline{\nabla w^\alpha}\cdot\nabla U^{\alpha}_n \du x = 0. \]
	The second equality holds since the harmonic function $U^{\alpha}_n$ is constant in $D_{n}$. Therefore,
	\[  \big\langle (\mathcal{K}^{-\alpha,0}_{\wD})^\ast[ \phi] \big\rangle_n \ = \frac12\langle \phi\rangle_n. \]
\end{proof}
We now consider the quotient space $\mathbb{H}^{-\frac12}(\partial \wD)/\mathcal{D}$. It can be proved directly by noting that for all $\varphi\in\mathbb{H}^{-\frac{1}{2}}(\partial \wD)$, one has 
\[ (\varphi,\psi_{n}^{\alpha})_{\ast} = -\langle \varphi,\mathcal{S}^{\alpha,0}_{\wD}[\psi_{n}^{\alpha}]  \rangle = -\langle \varphi,\chi_{_{\partial D_{n}}}  \rangle = -\langle \varphi\rangle_{_n}. \]
\begin{lemma}
	With the norm induced by $(\cdot,\cdot)_\ast$, the quotient space $\mathbb H^{-\frac12}(\partial\wD)/\mathcal D$ is isometrically equivalent to $\mathcal X=\mathcal D^{\perp_\ast}$: 
	\[ \mathbb{H}^{-\frac12}(\partial \wD)/\mathcal{D}\simeq \mathcal{X}. \] 
\end{lemma}

We turn to proving the coercivity of the operator $ \frac{1}{2}\operatorname{Id} - (\mathcal{K}^{-\alpha,0}_{\wD})^\ast$. 
\begin{theorem}
	The operator $A \coloneqq\frac{1}{2}\operatorname{Id} - (\mathcal{K}^{-\alpha,0}_{\wD})^\ast$ is coercive on $\mathcal{X}$. Specifically, there exists $c>0$ such that 
	\begin{equation}
		(A\psi,\psi)_\ast \ge c \Vert \psi \Vert^2_{\mathbb{H}^{-\frac{1}{2}}(\partial \wD)},\quad \psi\in\mathcal X.
	\end{equation} 
\end{theorem}
\begin{proof}
	We argue by contradiction. Suppose that there exists a sequence $\{\psi_l\}_{l\ge1}\subset\mathcal X$ such that 
	\[ (A\psi_l,\psi_{l})_{\ast}\to 0,\quad \Vert\psi_l \Vert_{\mathbb{H}^{-\frac{1}{2}}(\partial \wD)}=1. \]
	Subtracting the constant part in each inclusion $D_{0,n}$, one has, by Poincare's inequality,
	\begin{equation}\label{apeqn:sequencecontrol}
		\Vert \mathcal{S}^{\alpha,0}_{\wD}[\psi_{l}] - c_{n}^{(l)}\chi_{_{D_{0,n}}} \Vert_{\mathbb{L}^{2}(D_{0,n})} \le C \Vert \nabla \mathcal{S}^{\alpha,0}_{\wD}[\psi_{l}] \Vert_{\mathbb{L}^{2}(D_{0,n})}. 
	\end{equation}
	Here, the constants $\{c_{n}^{(l)}\}_{n\in\mathbb{Z}}$ are defined by 
	\[ c_{n}^{(l)} \coloneqq \frac{1}{|D_{0,n}|}\int_{D_{0,n}} \mathcal{S}^{\alpha,0}_{\wD}[\psi_{l}]\du x. \]
	It can be verified from \Cref{aplem:boundedH1} that $\{c_{n}^{(l)}\}_{n\in\mathbb{Z}} \in l^2(\mathbb{Z})$. \par 
	From Green's identity, 
	\[
	(A\phi,\psi)_\ast=\int_{\wD}\overline{\nabla \mathcal{S}^{\alpha,0}_{\wD}[\phi]}\cdot\nabla \mathcal{S}^{\alpha,0}_{\wD}[\psi]dx.
	\]
	The interior Poincar\'e constant is uniform by \Cref{apsec:poincare}. Squaring \eqref{apeqn:sequencecontrol}, adding the gradient term, and summing over $n$ gives
\[
\Big\|\mathcal S^{\alpha,0}_{\wD}[\psi_l]-\sum_nc_n^{(l)}\chi_{D_{0,n}}\Big\|_{\mathbb H^1(\wD)}^2
\le C\int_{\wD}|\nabla\mathcal S^{\alpha,0}_{\wD}[\psi_l]|^2\du x\longrightarrow0.
\]
Therefore, by \Cref{apthm:trace}, we have  
	\[ \Big\Vert \mathcal{S}^{\alpha,0}_{\wD}[\psi_{l}] - \sum_{n\in\mathbb{Z}}c_{n}^{(l)}\chi_{_{\partial D_{0,n}}} \Big\Vert_{\mathbb{H}^{\frac{1}{2}}(\partial \wD)} \to 0.\]
	This implies that
	\[ \operatorname{dist}(\psi_{l},\mathcal{D}) \to 0,\]
	which contradicts the fact that $\psi_l\in\mathcal X=\mathcal D^{\perp_\ast}$ and the $\ast$ norm is equivalent to the original Sobolev norm.
\end{proof}
Therefore, we have the following invertibility theorem of $A$.
\begin{theorem}\label{apthm:invertiK}
	The operator $-A = -\frac{1}{2}\operatorname{Id} + (\mathcal{K}^{-\alpha,0}_{\wD})^\ast:\mathcal X\to\mathcal X$ is a bounded isomorphism. In particular, $\operatorname{Im}A=\mathcal X$.
\end{theorem}
\begin{proof}
	Since the operator $A$ is bounded and coercive on $\mathcal{X}$, it follows from Lax--Milgram's theorem that the operator $-A$ is invertible. The preceding energy identity gives self-adjointness with respect to $(\cdot,\cdot)_\ast$. \par 
	Now we prove that $\operatorname{Im}A = \mathcal X$. From \Cref{aplem:mapping}, it follows that $\operatorname{Im}A\subset \mathcal{X}$. From the coercivity of $A$ on $\mathcal{X}$, it follows that for all $g\in \mathcal X$, there exists $\psi\in\mathcal{X}$ such that $A\psi = g$. This proves the desired result.
\end{proof}

\section*{Acknowledgment}
This work is partially supported by the National Key R\&D Program of China grant number 2021YFA0719200 and 2024YFA1016000. This work is also supported by the Fundamental Research Funds for the Central Universities grant number 226-2025-00192. It was completed while the first author was visiting the Hong Kong Institute for Advanced Study as a Senior Fellow.

\bibliographystyle{amsplain}
\bibliography{references3}

\end{document}